\documentclass{article}

\usepackage{abstract}
\usepackage{authblk}
\usepackage{float}
\usepackage[margin=1in]{geometry}
\usepackage{graphicx}
\usepackage[colorlinks,linkcolor=blue,citecolor=blue,urlcolor=blue]{hyperref}
\usepackage{natbib}
\usepackage{subfig}

\usepackage{amsthm}
\usepackage{amsmath}
\usepackage{amssymb}
\usepackage{commath}
\usepackage{mathtools}
\usepackage{mathrsfs}

\usepackage{dsfont}

\newcommand{\und}[1]{\underline{#1}}
\newcommand{\mb}[1]{\boldsymbol{#1}}
\newcommand{\mbb}[1]{\mathbb{#1}}
\newcommand{\mcal}[1]{\mathcal{#1}}
\newcommand{\mscr}[1]{\mathscr{#1}}
\newcommand{\mfrak}[1]{\mathfrak{#1}}
\newcommand{\ip}[2]{\left\langle #1 , #2 \right\rangle}

\newtheorem{definition}{Definition}[section]
\newtheorem{theorem}{Theorem}[section]
\newtheorem{lemma}{Lemma}[section]

\title{Functional Factor Modeling of Brain Connectivity}

\author[1]{Kyle Stanley}
\author[2]{Nicole Lazar}
\author[3]{Matthew Reimherr}

\affil[1]{Department of Statistics, Pennsylvania State University\\
\texttt{kms8227@psu.edu}}
\affil[2]{Department of Statistics and Huck Institutes of the Life Sciences, Pennsylvania State University\\
\texttt{nfl5182@psu.edu}}
\affil[3]{Department of Statistics, Pennsylvania State University, and Amazon Science\\
\texttt{mlr36@psu.edu}}

\date{}

\begin{document}

\maketitle

\begin{abstract}
Many fMRI analyses examine functional connectivity, or statistical dependencies among remote brain regions. Yet popular methods for studying whole-brain functional connectivity often yield results that are difficult to interpret. Factor analysis offers a natural framework in which to study such dependencies, particularly given its emphasis on interpretability. However, multivariate factor models break down when applied to functional and spatiotemporal data, like fMRI. We present a factor model for discretely-observed multidimensional functional data that is well-suited to the study of functional connectivity. Unlike classical factor models which decompose a multivariate observation into a ``common'' term that captures covariance between observed variables and an uncorrelated ``idiosyncratic'' term that captures variance unique to each observed variable, our model decomposes a functional observation into two uncorrelated components: a ``global'' term that captures long-range dependencies and a ``local'' term that captures short-range dependencies. We show that if the global covariance is smooth with finite rank and the local covariance is banded with potentially infinite rank, then this decomposition is identifiable. Under these conditions, recovery of the global covariance amounts to rank-constrained matrix completion, which we exploit to formulate consistent loading estimators. We study these estimators, and their more interpretable post-processed counterparts, through simulations, then use our approach to uncover a rich covariance structure in a collection of resting-state fMRI scans.
\end{abstract}


\section{Introduction}

There exist a number of imaging modalities that allow researchers to measure the physiological changes that accompany neuronal activation. Among the most widely used techniques is functional magnetic resonance imaging (fMRI). FMRI is based on the haemodynamic response wherein blood delivers oxygen to active brain regions at a greater rate than to inactive regions. Blood-oxygen-level-dependent (BOLD) imaging uses magnetic fields to measure relative oxygenation levels across the brain, allowing researchers to use BOLD signal as a proxy for neuronal activation. During an fMRI experiment, a subject is placed in a scanner that collects a temporal sequence of three-dimensional brain images. Each image is partitioned into a grid of three-dimensional volume elements, called voxels, and contains BOLD measurements at each voxel \citep{lindquist-2008}. The resulting data are both large and complex. The data from a single session typically contains hundreds of brain images, each with BOLD observations at more than 100,000 voxels; large fMRI datasets include hundreds of such scans. Moreover, these data exhibit spatiotemporal dependencies that complicate analysis. 

A common goal of fMRI analyses is to characterize \textit{functional connectivity}, or the statistical dependencies among remote neurophysiological events \citep{friston-2011}. This is of particular interest in resting-state fMRI, wherein subjects do not perform an explicit task while in the scanner. In this work, we study functional connectivity in resting-state data from the PIOP1 dataset of the Amsterdam Open MRI Collection (AOMIC) \citep*{snoek-et-al-2021}, which contain six-minute scans collected from 210 subjects at rest. To do so, we develop a new approach that integrates two fields of statistics: factor analysis and functional data analysis.  

Techniques used to study resting-state functional connectivity fall in two categories: seed methods and whole-brain methods. To conduct a seed-based connectivity analysis, researchers first select a seed voxel (or region), then correlate the time series of that seed with the time series of all other voxels (or regions), resulting in a functional connectivity map \citep*{biswal-etal-1995}. Although readily interpretable, these functional connectivity maps only describe dependencies of the selected seed region \citep{heuvel-pol-2010}. Principal component analysis (PCA) is a rudimentary whole-brain approach to functional connectivity that finds the collection of orthogonal spatial maps (a.k.a., principal components) that maximally explains the variation -- both between- and within-voxel -- in a set of brain images \citep{friston-etal-1993}. The orthogonality constraint, however, means that the low-order subspace is expressed by dense spatial maps containing many high-magnitude weights that complicate interpretation. Independent component analysis (ICA; \citealp{beckmann-smith-2004}), perhaps the most popular whole-brain approach, offers a more parsimonious representation of this subspace. After computing the principal spatial components and associated temporal components from an fMRI scan, ICA rotates these components until the spatial set is maximally independent. Although the resulting spatial maps are typically more parsimonious than those of PCA, ICA may fragment broad areas of activation into multiple maps with highly correlated time courses in its pursuit to maximize spatial independence \citep{friston-1998}. This means ICA findings are not robust to dimension overspecification. 

To address the shortcomings of the above approaches, we develop a method based on exploratory factor analysis. The goal of factor analysis is to describe the covariance relationships between many observed variables using a few unobserved random quantities. Classical factor analyses use the \textit{orthogonal factor model (OFM)}, which posits that $M$ zero-mean observed quantities $\und{X} \in \mbb{R}^M$ having covariance $\mb{C} \in \mbb{R}^{M\times M}$ are the sum of a \textit{common component}, given by linear combinations of $K < M$ unobserved uncorrelated factors $\und{F} \in \mbb{R}^{K}$ with zero-mean and unit variance, and an \textit{idiosyncratic component}, given by a vector of uncorrelated errors $\und{\epsilon} \in \mbb{R}^M$ with zero-mean. The loading matrix $\mb{L} \in \mbb{R}^{M\times K}$ specifies the coefficients of these linear combinations in the common component:
\begin{align*}
    \und{X} = \mb{L}\und{F} + \und{\epsilon}.
\end{align*}
If $\mb{D} \in \mbb{R}^{M\times M}$ denotes the diagonal covariance matrix of $\und{\epsilon}$, then $\mb{C}$ decomposes as
\begin{align*}
    \mb{C} = \mb{L}\mb{L}^T + \mb{D}.
\end{align*}
The OFM is identifiable only up to an orthogonal rotation since, for $\und{F}^* = \mb{R}^T\und{F}$ and $\mb{L}^* = \mb{L}\mb{R}$, where $\mb{R} \in \mbb{R}^{K \times K}$ is a rotation matrix, we have the equivalences
\begin{align*}
    \und{X} = \mb{L}\und{F} + \und{\epsilon} = \mb{L}^*\und{F}^* + \und{\epsilon} \hspace{1em} \text{and} \hspace{1em} \mb{C} = \mb{L}\mb{L}^T + \mb{D} = \mb{L}^*(\mb{L}^*)^T + \mb{D}.
\end{align*}
Practitioners address this rotation problem by choosing from a suite of factor rotation algorithms, like quartimax \citep{neuhaus-wrigley-1954} and varimax \citep{kaiser-1958} rotation, that bring the loading matrix closer to a so-called \textit{simple structure}, where in (i) each variable has a high loading on a single factor but near-zero loadings on other factors, and (ii) each factor is described by only a few variables with high loadings on that factor while other variables have near-zero loadings on that factor. This structure provides a parsimonious representation of the common component and one can opt for a rotation procedure that is robust to dimension overspecification. 

To grasp the value of factor analysis in functional connectivity, consider a naive application wherein an OFM is fit to a vector of voxel-wise BOLD observations. Used in conjunction with factor rotation, the OFM would provide $K$ simple spatial maps amenable to interpretation. This sparse and low-dimensional representation of the entire brain makes factor analysis an attractive alternative to seed and ICA methods which, respectively, are not whole-brain and lack such interpretable results. However, there is a problem; the OFM crucially assumes that its errors are uncorrelated, a condition violated in this application given that errors for nearby voxels will be correlated. Such autocorrelation may dissipate as $K$ grows, but this has two undesirable effects. First, the parsimony of the low-order model is lost. Second, a high-order common component will capture variation at small scales -- possibly even between adjacent voxels -- an outcome at odds with the purpose of functional connectivity studies: to study correlations between distant brain regions. In a more appropriate factor model, the common component would be low-rank and capture only large-scale variation while its error term would permit variation at a smaller scale. 

To develop this better suited factor model, it helps to view fMRI data in the \textit{functional data analysis (FDA)} framework. Data are said to be functional if we can naturally assume that they arise from some smooth curve. FDA is then the statistical analysis of sample curves in such data \citep{ramsay-silverman-2005, wang-et-al-2016, kokoszka-reimherr-2017}. Since it is reasonable to view BOLD observations of a brain image as noisy realizations from some smooth function over a three-dimensional domain, fMRI data is an example of multidimensional functional data. Though not widely adopted by the neuroimaging community, fMRI has been studied using FDA tools such as functional principal component analysis (FPCA; \citealp{viviani-etal-2005}), functional regression \citep{reiss-etal-2015, reiss-etal-2017}, and functional graphical models \citep{li-solea-2018}. To adapt factor analytic techniques to fMRI data, we can extend classical factor methods to this functional setting. 

Early extensions of the OFM come from studies of large-dimensional econometric panel data, in which $M$ time series (e.g., one time series per asset) are observed at the same $n$ points in time. When $M$ is large, traditionally reliable time series tools (e.g., vector autoregressive models) require estimation of too many parameters. Factor models are a practical alternative as they provide more parsimonious parameterizations. These factor models should, however, possess two characteristics to accommodate properties inherent to economic time series. First, they ought to be dynamic as problems involving time series are dynamic in nature. Second, they should permit correlated errors as uncorrelatedness is often unrealistic with large $M$. \cite{sargent-sims-1977} and \cite{geweke-1977} innovated on the OFM by making latent factors dynamic in their \textit{dynamic factor model (DFM)}. However, these early DFMs were \textit{exact} in that they assumed uncorrelated errors. Factor models did not permit such error correlation until \cite{chamberlain-rothschild-1983} introduced their \textit{approximate factor model (AFM)}, which required only that the largest eigenvalue of the error covariance be bounded. Although this AFM was static, much subsequent study focused on factor models that were both approximate and dynamic \citep{forni-et-al-2000, stock-watson-2002, bai-2003}.  

Factor models are well-studied in large dimensions, but the literature on infinite-dimensional factor analysis is scarce. \cite{hays-et-al-2012} were the first to merge ideas from FDA with factor analysis in their \textit{dynamic functional factor model (DFFM)} which they used to forecast yield curves. \cite{liebl-2013} proposed a different DFFM which, unlike the model of \cite{hays-et-al-2012}, placed no a priori constraints on the latent process of factor scores. \cite{kowal-et-al-2017} presented a Bayesian DFFM for modeling multivariate dependent functional data. None of these early functional factor models permit any degree of error dependence, a structure that certainly arises when representing potentially infinite-rank observations with low-rank objects. \cite{otto-salish-2022} were the first to address this issue using their approximate DFFM wherein the error covariance is left unrestricted. However, as this is a model for \textit{completely}-observed functional data, one must either observe these data at infinite resolution (a practical impossibility) or obtain functional representations of discretely-observed sample curves  via spline smoothing \citep{ramsay-silverman-2005} prior to model estimation.

Though spline smoothing (and other forms of pre-smoothing) may be an effective first step in many analyses of functional data, it is not in factor analysis. Analyses that begin by smoothing functional data assume that a functional observation $X(s)$ is the sum of two uncorrelated components: a smooth signal $Y(s)$ and a white noise term $\epsilon(s)$ whose covariance kernel is zero off the diagonal (or some infinitesimally narrow band). When observed at discrete points, 
\begin{align*}
    X(s_m) = Y(s_m) + \epsilon(s_m), \ m = 1, \dots, M,
\end{align*}
the smoothed curve $\Tilde{X}(s)$ is defined as
\begin{equation}
\label{eqn:sm-obj}
    \Tilde{X}(s) = \arg \min_{f \in C^2[0,1]} \left\{ \sum_{m=1}^M (f(s_m) - X(s_m))^2 + \alpha \norm{\partial_s^2 f}_{L^2} \right\}
\end{equation}
When $\epsilon(s)$ is white noise, $\Tilde{X}(s)$ is a good approximation to $Y(s)$. If, instead, the signal $Y(s)$ and the error $\epsilon(s)$ are defined as in our theorized factor model -- the former capturing only large-scale variation and the latter permitting small-scale variation -- then minimizing the ``uncorrelated'' objective of (\ref{eqn:sm-obj}) yields a smoothed curve $\Tilde{X}(s)$ contaminated by short-range variation of the error term, making it a poor approximation to $Y(s)$. Using $\Tilde{X}(s)$ as a stand-in for $Y(s)$ in subsequent analyses is ill-advised. Thus, a factor model for functional data should not presume completely-observed sample curves, but rather deal directly with discretely-observed data. 

As it turns out, an orthogonal factor structure indeed underlies discretely-observed functional data. \cite{hormann-jammoul-2020} formalized this observation, noting, however, that the assumption of uncorrelated errors is rarely justified. They suggest that a more realistic factor model would permit error correlations that ``taper to zero with increasing lag.'' But this poses an identifiability question: is it possible to distinguish such error dependence from that residing in the common component? \cite{descary-panaretos-2019} show that this is, indeed, possible if one assumes the covariance of the common term is smooth with low rank and that of the error term is banded with potentially infinite rank. These assumptions give rise to a distinctly functional re-characterization of the common and idiosyncratic components of the OFM: the former becomes ``global'', capturing dependencies between distant regions of the domain, while the latter becomes ``local'', describing dependencies between nearby regions. As this is the precise perspective put forth by our theorized factor model for functional connectivity, we let the ideas of \cite{descary-panaretos-2019} guide the development of this model.   

In this paper, we develop a factor analytic approach for discretely-observed multidimensional functional data that is well-suited to the study of functional connectivity in fMRI. Our methodology is inspired by the framework of \cite{descary-panaretos-2019}, which we build upon in three key respects. First, we extend their conditions for identifying the global and local components of a full covariance to multidimensional functional data, and show that these conditions become less restrictive as the dimensionality grows. Second, we improve upon their global covariance estimator by incorporating a roughness penalty. Third, we develop a post-processing procedure -- inspired by methods from classical factor analysis -- that facilitates interpretation of global estimates. Using this new approach, we uncover a rich covariance structure in the 210 resting-state scans from the AOMIC-PIOP1 dataset.  

The rest of this paper proceeds as follows. In Section \ref{sec:mod}, we introduce functional factor models for both completely- and discretely-observed functional data, and present identifiability conditions for both. In Section \ref{sec:est}, we develop estimation methods for the discretely-observed functional factor model. We then take a brief detour, in Section \ref{sec:ica}, to review an ICA-based approach used throughout the neuroimaging community. In Section \ref{sec:sim}, we use simulations to compare our estimator to that of this ICA-based approach and several other comparators. Section \ref{sec:da} contains an application of our approach to the AOMIC resting-state data which we contrast with an application of the ICA-based estimator. We conclude, in Section \ref{sec:discussion}, with some final remarks and future work.

\section{The Functional Factor Model and its Identification}
\label{sec:mod}

This section considers two data observation paradigms: (i) data observed completely (e.g., a brain image with infinite spatial resolution, making it a function), and (ii) data observed discretely on a grid (e.g., a brain image with finite spatial resolution, making it a tensor). Given that our methodology forgos the usual data smoothing step in many analyses of functional data, the completely-observed paradigm is not practically possible. We instead use this first paradigm as an intermediate step in the development of our theory for the second. Unless otherwise specified, let $a$ or $A$ denote scalars, $\und{a}$ or $\und{A}$ denote vectors, $\mb{A}$ denote matrices, $\mcal{A}$ denote tensors, and $\mscr{A}$ denote operators.

\subsection{Completely-Observed Data}
\label{sec:mod-co}

In the completely-observed paradigm, we assume $X$ is a mean-zero random function in $L^2([0,1]^D)$, which is endowed with the usual inner product and norm: 
\begin{align*}
    \ip{f}{g}_{L^2} = \int_{[0,1]^D} f(\und{s})g(\und{s})d\und{s} \hspace{1em} \text{and} \hspace{1em} \norm{f}_{L^2} = \ip{f}{f}_{L^2}.
\end{align*}
We may view $X$, in the fMRI setting, as a 3-dimensional ($D = 3$) brain image at infinite spatial resolution, an object conceivable only in theory. The covariance operator $\mscr{C}$ of $X$ is the integral operator with kernel $$c(\und{s}_1, \und{s}_2) = \mbb{E}[X(\und{s}_1) X(\und{s}_2)].$$ We say $X$ follows a \textit{functional factor model (FFM)} with $K$ factors if
\begin{equation*}
    X(\und{s}) = \underbrace{\sum_{k=1}^K l_k(\und{s})f_k}_{Y(\und{s})} + \epsilon(\und{s}),
\end{equation*}
where $l_k$ and $f_k$ denote the $k$th \textit{loading function} and \textit{factor}, respectively, and $\epsilon$ is the \textit{error function}. In the spirit of classical factor models, this model assumes the $f_k$ are uncorrelated, mean-zero with unit variance, and independent of the mean-zero $\epsilon$, and that the $l_k$ are linearly independent. In this functional setting, we want the components $Y$ and $\epsilon$ to capture global and local variation in $X$, respectively. If we assume the $l_k$ are sufficiently smooth (in a sense to be defined in Section \ref{sec:mod-co-id}), then $Y$, being a finite sum of these smooth functions, will contain the global variation in $X$ but fail to capture all of the function's local variation. If we further assume $\epsilon$ has $\und{\delta}$-banded covariance, i.e., $\text{Cov}(\epsilon(\und{s}_1), \epsilon(\und{s}_2)) = 0$ when $\abs{s_{1,d} - s_{2,d}} \geq \delta_d$ for some $\delta_d > 0$, $d = 1, \dots, D$, then the error function will contain only that local variation not captured by $Y$. When $X$ is an infinite-resolution brain image, this model decomposes $X$ into the sum of two images: one defined by a finite linear combination whose components (i.e., loading functions) describe dependencies between distant brain regions, and another that captures dependencies between nearby regions. 

If the global component $Y$ has covariance $\mscr{G}$ with kernel $g$, and the local component $\epsilon$ has covariance $\mscr{B}$ with kernel $b$, then
\begin{equation*}
\begin{split}
    \mscr{C} & = \mscr{G} + \mscr{B}, \\
    c(\und{s}_1, \und{s}_2) & = g(\und{s}_1, \und{s}_2) + b(\und{s}_1, \und{s}_2).
\end{split}
\end{equation*}
Moreover, being covariances, $g$ and $b$ admit the Mercer decompositions
\begin{align*}
    g(\und{s}_1, \und{s}_2) & = \sum_{k=1}^K \lambda_k \eta_k (\und{s}_1) \eta_k (\und{s}_2), \\
    b(\und{s}_1, \und{s}_2) & = \sum_{j=1}^{\infty} \beta_j \psi_j (\und{s}_1) \psi_j (\und{s}_2) = \mathds{1}\{ \abs{s_{1,d} - s_{2,d}} < \delta_d, \ \forall d \} \sum_{j=1}^{\infty} \beta_j \psi_j (\und{s}_1) \psi_j (\und{s}_2).
\end{align*}
Since the assumptions of the model also imply
\begin{align*}
    g(\und{s}_1, \und{s}_2) & = \sum_{k=1}^K l_k(\und{s}_1)l_k(\und{s}_2)
\end{align*}
the loading functions $l_k$ are equal to the scaled eigenfunctions of $\mscr{G}$, $\lambda_k^{1/2}\eta_k$, up to an orthogonal rotation. This makes it possible to impose smoothness on the loading functions through the eigenfunctions of $\mscr{G}$, as we do in Section \ref{sec:mod-co-id}. 

Like many latent models, the FFM is rotationally indeterminate. That is, if loading functions $l_1, \dots, l_K$ satisfy model constraints, then so do any orthogonal rotation of these functions (see Section 2 of Supplement A for a detailed explanation). In the fMRI context, this means there are infinitely many ways to express the global component via spatial map sets of size $K$, and each set is a rotation away from the rest. Classical factor analysis addresses this rotation problem by bringing the loadings closer to a simple structure. We explore such approaches in Section \ref{sec:est-pp}.

\subsubsection{Identification of $(\mscr{G}, \mscr{B})$ from $\mscr{C}$}
\label{sec:mod-co-id}

Despite rotational indeterminacy of the loading functions, we may still hope to identify the global and local covariance components $(\mscr{G}, \mscr{B})$ from $\mscr{C}$. Thus far, we have imposed two critical conditions on this covariance decomposition: (i) $\mscr{G}$ has finite rank, and (ii) $\mscr{B}$ has banded covariance. These restrictions, however, do not sufficiently constrain the problem. Given $\mscr{C}$, we know $\mscr{G}$ off the band, but it is not clear how to uniquely extend $\mscr{G}$ onto the band. Now, if $\mscr{C}$ were contaminated by $\mscr{B}$ only on its diagonal (i.e., $\mscr{B}$ is $\und{0}$-banded), then we could simply assume $\mscr{G}$ to be continuous and smooth over this diagonal. However, this strategy does not suffice in the presence of a nontrivial band. Recall that in Section \ref{sec:mod-co}, we alluded to an additional assumption on the smoothness of the eigenfunctions of $\mscr{G}$. In Theorem 1 from Section 3.1 of Supplement A, which extends Theorem 1 from \cite{descary-panaretos-2019}, we formalize this condition through the notion of \textit{analyticity}, showing that real analytic loading functions provide the formulation of smoothness needed to identify the covariance decomposition. To prove this result, we exploit the so-called analytic continuation property, which states that if a function is analytic over some domain, but is known only on an open subset of that domain, then the function extends uniquely to the rest of the domain. This property allows us to uniquely extend $\mscr{G}$ onto the band, thereby establishing a unique decomposition of the full covariance. When $X$ is an infinite-resolution brain image, this condition implies that the loading functions, which capture correlations between distant brain regions, are analytically smooth over the domain of the image.

\subsection{Discretely-Observed Data}
\label{sec:mod-do}

In practice, we can measure a function $X$ at only a finite number of locations. In fMRI, these locations are an evenly-spaced 3-dimensional grid of voxels. This manuscript assumes such a grid, although Section 3.2 of Supplement A presents theory for a more general discrete-observation framework. Let $\{ \und{s}_{\und{m}} \}_{\und{m} \in \mbb{N}_{\und{M}}}$ be an evenly-spaced grid where $\und{M} = (M_1, \dots, M_D)$ is the grid resolution and $\mbb{N}_{\und{M}} = \mbb{N}_{M_1} \times \dots \times \mbb{N}_{M_D}$ is the set of all $D$-dimensional multi-indices with $d$th component no greater than $M_d$ (i.e., $\mbb{N}_{M_d} = \{1, \dots, M_d\}$). In the case of the AOMIC dataset which contains brain images of dimension $65 \times 77 \times 60$, we have resolution $\und{M} = (65, 77, 60)$ with dimension $D = 3$, and $\und{s}_{\und{m}}$ denotes the location of the voxel at grid position $\und{m} = (m_1, m_2, m_3)$. Suppose we observe each of the $n$ samples of $X$ at the $M = M_1 \dots M_D$ grid points,
\begin{equation*}
    X_i(\und{s}_{\und{m}}) = \sum_{k=1}^K l_k(\und{s}_{\und{m}})f_{i,k} + \epsilon_i(\und{s}_{\und{m}}), \ i = 1, \dots, n, \ \und{m} \in \mbb{N}_{\und{M}}.
\end{equation*}
If we summarize the functional terms in the model with the tensors 
\begin{align*}
    \mcal{X}^{\und{M}}(\und{m}) = X(\und{s}_{\und{m}}),  \hspace{1em} \mcal{L}^{\und{M}}_k(\und{m}) = l_k(\und{s}_{\und{m}}), \hspace{1em} \mcal{E}^{\und{M}}(\und{m}) = \epsilon(\und{s}_{\und{m}}),
\end{align*} 
then we can compactly write the \textit{\und{M}-resolution functional factor model} with $K$ factors as
\begin{align*}
    \mcal{X}^{\und{M}} = \sum_{k=1}^K \mcal{L}^{\und{M}}_k f_k + \mcal{E}^{\und{M}}.
\end{align*}
Defining the $M_1 \times \dots \times M_D \times M_1 \times \dots \times M_D$ covariance tensors 
\begin{align*}
    \mcal{C}^{\und{M}} (\und{i},\und{j})  = c(\und{s}_{\und{i}},  \und{s}_{\und{j}}), \hspace{1em}
    \mcal{G}^{\und{M}} (\und{i},\und{j}) = g(\und{s}_{\und{i}}, \und{s}_{\und{j}}), \hspace{1em}
    \mcal{B}^{\und{M}} (\und{i},\und{j}) = b(\und{s}_{\und{i}}, \und{s}_{\und{j}}),
\end{align*}
we can invoke the model's assumptions, which carry over from the completely-observed setting, to decompose the covariance as 
\begin{equation*}
    \mcal{C}^{\und{M}} = \underbrace{\sum_{k=1}^K \mcal{L}^{\und{M}}_k \otimes \mcal{L}^{\und{M}}_k}_{\mcal{G}^{\und{M}}} + \mcal{B}^{\und{M}},
\end{equation*}
where $\otimes$ denotes the usual tensor product, and the loading tensor $\mcal{L}_k^{\und{M}}$ is equal to the $k$th scaled eigentensor of $\mcal{G}^{\und{M}}$, $\lambda_k^{1/2} \mcal{H}_k^{\und{M}}$, up to a rotation. For notational simplicity, we will often suppress the $\und{M}$ in the superscript when writing tensors. We will also say that tensors with dimension $M_1 \times \dots \times M_D$ and $M_1 \times \dots \times M_D \times M_1 \times \dots \times M_D$ simply have dimensions $\mfrak{M}$ and $\mfrak{M}\times\mfrak{M}$, respectively. This allows us abbreviate $\mbb{R}^{M_1 \times \dots \times M_D}$ and $\mbb{R}^{M_1 \times \dots \times M_D \times M_1 \times \dots \times M_D}$ with $\mbb{R}^{\mfrak{M}}$ and $\mbb{R}^{\mfrak{M}\times\mfrak{M}}$, respectively.

\subsubsection{Identification of $(\mcal{G}, \mcal{B})$ from $\mcal{C}$}
\label{sec:mod-do-id}

The rotation problem discussed in Section \ref{sec:mod-co} extends to this finite resolution framework. That is, given $\mcal{G}$, we can only identify the loading tensors $\mcal{L}_k$ up to a rotation. Yet we may still hope to identify ($\mcal{G}$, $\mcal{B}$) from $\mcal{C}$. This, however, is not obviously possible as analyticity, which was used to establish uniqueness in the completely-observed setting, is a property of functions, not of tensors. \cite{descary-panaretos-2019} showed that when $D=1$ this covariance decomposition is indeed unique at finite resolution and that analyticity plays a central role. In Theorem 4 from Section 3.2 of Supplement A, we extend their result to general $D$, and show that the central identification condition becomes less restrictive as $D$ grows:
\begin{align}
\label{eqn:max-K}
    K \leq K^* = \prod_{d=1}^D  \left\lfloor \left( \frac{1}{2} - \delta_d \right) M_d - 1\right\rfloor.
\end{align}
Condition (\ref{eqn:max-K}) reveals that at finite resolution, identification depends crucially on the relationship between the bandwidth $\und{\delta}$, resolution $\und{M}$, and rank of the smooth operator $K$. From a factor analytic perspective, for fixed $\und{\delta}$ and $\und{M}$, this condition gives the maximum number of identifiable factors $K^*$ (up to a rotation). Figure \ref{fig:max-ranks} shows that these parameter constraints are not very restrictive. Of particular note is the order-$D$ polynomial growth of the maximal rank in $M$. This means that in fMRI, where $M$ is typically large and $D>1$, the functional $K$-factor model will be identifiable up to a rotation for any practically reasonable $K$.  

\begin{figure}[!h]
    \centering
    \includegraphics[width=0.95\linewidth]{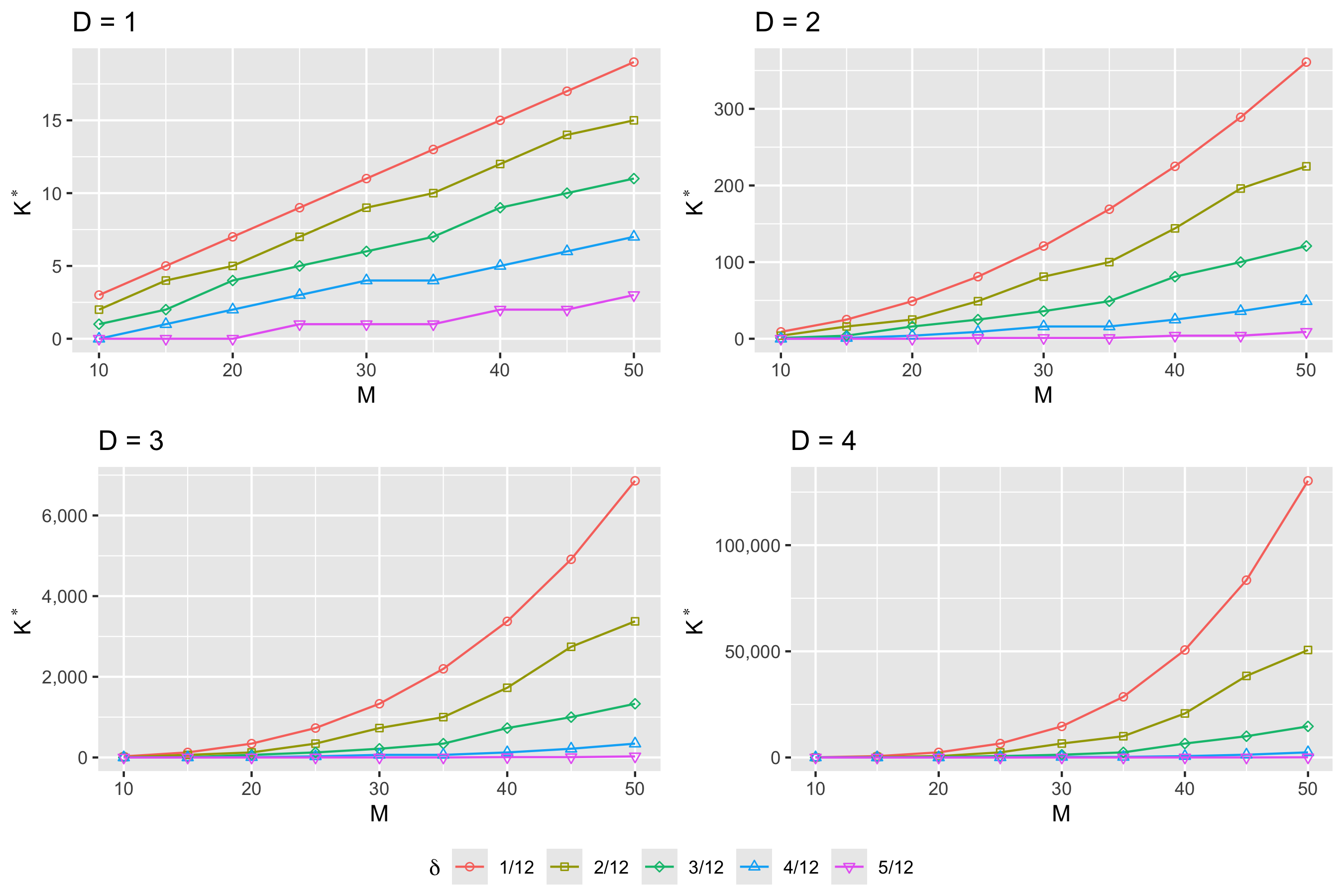}
    \caption{Plots of the maximal rank $K^*$ as a function of $M = M_1 = \dots = M_D$ for different values of $D$ and $\delta$. Note that each plot has a different vertical axis.}
    \label{fig:max-ranks}
\end{figure}

\section{Estimation}
\label{sec:est}

In this section, we formulate estimation in the discrete observation paradigm. The methodology is split in two phases: initial loading estimation and post-processing. In the first phase, we define an estimator for $\mcal{G}$, from which we derive estimates for the number of factors $K$ and the loading tensors $\mcal{L}_k$. To estimate the $\mcal{L}_k$ through $\mcal{G}$, we must assume that these $\mcal{L}_k$ are orthogonal. This, of course, is merely a mathematical convenience; there is no reason to suspect that a collection of brain maps describing connectivity structures is orthogonal. We depart from this assumption during post-processing, wherein the initial loading estimates are brought into a more interpretable form via rotation and shrinkage. This section emphasizes practical implementation of our method; Section 3.3 of Supplement A develops theory for these estimators and provides derivations of their asymptotic properties. 

\subsection{Initial Loading Estimation}
\label{sec:est-ile}

Recall that when the data are observed at resolution $\und{M}$, the covariance decomposes as
\begin{equation*}
    \mcal{C} = \underbrace{\sum_{k=1}^K \mcal{L}_k \otimes \mcal{L}_k}_{\mcal{G}} + \mcal{B}.
\end{equation*}
Two central goals of factor analysis are to estimate the number of factors $K$ and the loading tensors $\mcal{L}_k$, $k = 1, \dots, K$, from the empirical covariance $\hat{\mcal{C}}$. Since loading tensors are only identifiable up to a rotation, we temporarily assume they are orthogonal so that the $k$th loading tensor may be identified with the $k$th scaled eigentensor of $\mcal{G}$ (i.e., $\mcal{L}_k = \lambda_k^{1/2} \mcal{H}_k$, where $(\lambda_k^{1/2}, \mcal{H}_k)$ is the $k$th eigen-pair of $\mcal{G}$). In this section, we propose an estimator of $\mcal{G}$ from which we derive estimators for $K$ and the $\mcal{L}_k$.

Given a sample of discretely-observed functions summarized by the tensors $\mcal{X}_1, \dots, \mcal{X}_n$, our goal is to estimate the smooth covariance component $\mcal{G}$ from the empirical covariance $\hat{\mcal{C}}$ defined by
\begin{align}
\label{eqn:emp-cov}
    \hat{\mcal{C}} = \frac{1}{n - 1} \sum_{i=1}^n (\mcal{X}_i - \bar{\mcal{X}}) \otimes (\mcal{X}_i - \bar{\mcal{X}}),
\end{align}
where $\bar{\mcal{X}} = n^{-1} \sum_{i=1}^{n} \mcal{X}_i$. We do so by finding a smooth low-rank covariance tensor that is a good approximation to $\hat{\mcal{C}}$ off the band: 
\begin{equation}
\label{eqn:g-def}
    \hat{\mcal{G}} = \arg \min_{\theta \in \Theta_{\und{M}^*}} \left\{ \norm{\mcal{A} \circ (\hat{\mcal{C}} - \theta)}_F^2 + \alpha \mscr{P}(\theta) + \tau \text{rank}(\theta) \right\}.
\end{equation}
In the estimator of (\ref{eqn:g-def}), $\Theta_{\und{M}}^*$ is the space of $\mfrak{M} \times \mfrak{M}$ covariance tensors satisfying constraints specified in Section 3.3 of Supplement A, $\mcal{A}$ is the $\mfrak{M} \times \mfrak{M}$ ``band-deleting'' tensor defined by $\mcal{A}^{\und{M}}(\und{i}, \und{j}) = \mathds{1} \{ \abs{i_d - j_d} > \lceil M_d / 4 \rceil$ for $d = 1, \dots, D\}$, $\mscr{P}: \mbb{R}^{\mfrak{M} \times \mfrak{M}} \to \mbb{R}$ is some roughness-penalizing functional, $\alpha > 0$ is a roughness-penalizing parameter, and $\tau > 0$ is a rank-penalizing parameter. Using $\hat{\mcal{G}}$, we define estimators for $K$ and the $\mcal{L}_k$,
\begin{equation}
\label{eqn:k-l-def}
    \hat{K} = \text{rank}(\hat{\mcal{G}}) \hspace{1em} \text{and} \hspace{1em} \hat{\mcal{L}}_k = \hat{\lambda}_k \hat{\mcal{H}}_k,
\end{equation}
$k = 1, \dots, \hat{K}$, where $\hat{\lambda}_k$ and $\hat{\mcal{H}}_k$ are the $k$th scaled eigenvalue and eigentensor of $\hat{\mcal{G}}$, respectively. 

To describe how one performs the optimization in (\ref{eqn:g-def}), let us first assume we have already chosen a rank parameter $\tau^*$ and a smoothing parameter $\alpha^*$. Define the functional $f: \mbb{R}^{\mfrak{M}\times \mfrak{M}} \to \mbb{R}$ by $f(\theta) = \norm{ \mcal{A} \circ \left(\hat{\mcal{C}} - \theta \right) }_F^2 + \alpha^* \mscr{P}(\theta)$. We then estimate $\mcal{G}$ as follows: 

\begin{enumerate}
    \item[(a)] Solve the optimization problem 
    \begin{align*}
        \min_{0 \preceq \theta \in \mbb{R}^{\mfrak{M}\times\mfrak{M}}} f(\theta) \hspace{1em} \text{subject to} \hspace{1em} \text{rank}(\theta) \leq j,
    \end{align*}
    for $j \in \{ 1, \dots, K^* \}$, obtaining minimizers $\hat{\theta}_1, \dots, \hat{\theta}_{K^*}$.
    \item[(b)] Compute the quantities $\{ f(\hat{\theta}_j) + \tau^* j: j = 1, \dots , K^* \}$ and determine the $j$ furnishing the minimum quantity. Declare the corresponding $\hat{\theta}_j$ to be the estimator $\hat{\mcal{G}}$.
\end{enumerate}

Now consider selection of the penalty parameters, $\alpha$ and $\tau$ (see Table \ref{tab:params} for a summary of information related to these parameters and others). Choice of the smoothing parameter $\alpha$ depends on the roughness penalty $\mscr{P}$. Although many choices for $\mscr{P}$ are possible, we use one that promotes smoothness via the eigentensors. In particular, we define $\mscr{P}(\theta) = M^{-1} \sum_{k=1}^K \Tilde{\mscr{P}} \mcal{E}_k (\theta)$ where $\mcal{E}_k: \Theta_{\und{M}} \to \mbb{R}^{\mfrak{M}}$ maps a tensor to its $k$th scaled eigentensor, and $\Tilde{\mscr{P}}: \mbb{R}^{\mfrak{M}} \to \mbb{R}$ is a roughness-penalizing functional defined by $\Tilde{\mscr{P}}(\mcal{V}) = \mcal{V}_{\text{vec}}^T \mcal{R}_{\text{sq}} \mcal{V}_{\text{vec}}$ where $\mcal{R} \in \mbb{R}^{\mfrak{M}\times\mfrak{M}}$ is a generalization of the second difference matrix given by $\mcal{R}_{\und{m}_1,\und{m}_2} = 3^D - 1$ when $\norm{\und{m}_1 - \und{m}_2}_{\infty} = 0$, $\mcal{R}_{\und{m}_1,\und{m}_2} = -1$ when $\norm{\und{m}_1 - \und{m}_2}_{\infty} = 1$, and $\mcal{R}_{\und{m}_1,\und{m}_2} = 0$ otherwise. Here, $\mcal{V}_{\text{vec}}$ denotes the vectorization of $\mcal{V}$ and $\mcal{R}_{\text{sq}}$ is the square matricization of $\mcal{R}$ (see Section 1 of Supplement A for formal definitions of these tensor reshapings). As in smoothed functional principal component analysis, it may be satisfactory (or even preferable) to make the choice of $\alpha$ subjectively based on visual inspection (\citealp[Section 9.3.3]{ramsay-silverman-2005}). For a discussion on automatic selection via cross-validation, see Section 3 of Supplement A.

In practice, we do not explicitly choose the rank parameter $\tau$ invoked in Step (b). Since each $\tau$ implies a choice of rank $j_\tau$, and thus some $f(\hat{\theta}_{j_\tau})$, we use the nonincreasing function $j \mapsto f(\hat{\theta}_j)$ to implicitly choose $\tau$. Specifically, we identify the value of $j$ to be the smallest one such that $f(\hat{\theta}_j) < c$ for some threshold $c$, then set $\hat{\mcal{G}}$ to the corresponding minimizer (see Figure \ref{fig:scree-illust}). We would like the threshold $c$ to be small enough that $f$ evaluated at the estimator is low, but not so small that the estimator has large rank. We thus select the rank of the estimator by identifying the elbow of the scree-type plot generated by the function $j \mapsto f(\hat{\theta}_j)$.

\begin{figure}[!h]
\centering
\subfloat[Small $\tau$]{
    \label{fig:scree-illust-1}\includegraphics[width=0.8\linewidth]{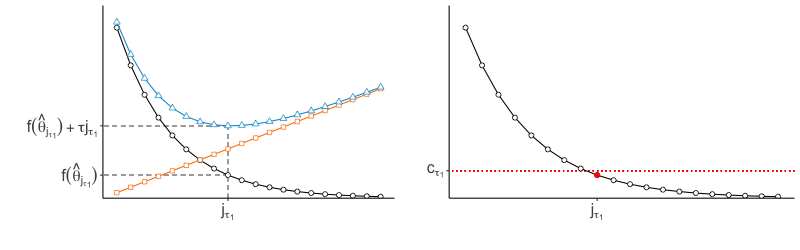}
}\\
\subfloat[Large $\tau$]{
    \label{fig:scree-illust-2}\includegraphics[width=0.8\linewidth]{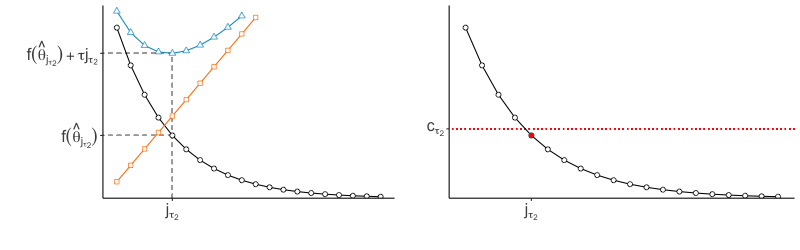}
}
\caption{An illustration of the scree plot approach for selecting the number of factors using (a) a small rank penalty $\tau_1$ and (b) a large rank penalty $\tau_2$. Blue dots represent the non-increasing function $j \mapsto f(\hat{\theta_j})$, orange squares represent the rank penalty $\tau j$, and blue triangle represents their sum (the function to be minimized). The number $c_{\tau}$ is one example of a threshold furnishing the rank $j_{\tau}$.}
\label{fig:scree-illust}
\end{figure}

As shown in Section 3.3 of Supplement A, we can theoretically estimate $\mcal{G}$ using any $\und{\delta}$ for which each $\delta_d$ is less than 0.25. The choice of $\und{\delta}$ thus depends on how much small-scale variation one want to admit into the global component. In this sense, $\und{\delta}$ does not need to be tuned; rather, one needs to consider the purpose of the application in its selection. For instance, in the analysis of Section \ref{sec:da}, we set $\delta_d = 0.1$ for each $d$, which demotes variation isolated to a radius of a few voxels to the local component. However, the optimization described in the next paragraph may be unstable when components of $\und{\delta}$ are near 0.25. For applications requiring wide bands, one may want to test a grid of bandwidths that increments to the desired bandwidth. If estimates change abruptly before reaching the final bandwidth, then algorithmic instability may be to blame.

\begin{table}
    \centering
    \begin{tabular}{c | l | l}
        Parameter & \multicolumn{1}{|c|}{Influence} & \multicolumn{1}{|c}{Selection} \\
        \hline
        \hline
        $\und{\delta}$ & increasing $\delta_d$ removes rough variation from $\hat{\mcal{G}}$ & application-informed, testing for stability \\
        \hline
        $\alpha$ & increasing $\alpha$ yields smoother $\hat{\mcal{L}}_k$ & subjective selection or cross-validation \\
        \hline
        $\tau$ & increasing $\tau$ lowers $\hat{K}$ & implicitly chose via scree-type plots \\
        \hline
        $\und{\kappa}$ & increasing $\kappa_k$ produces sparser $\hat{\mcal{L}}_k$ & cross-validation \\
    \end{tabular}
    \caption{Estimation parameters, along with their influences on the estimators and selection procedures.}
    \label{tab:params}
\end{table}

To solve the optimization problems in (a) we resort to approximate methods. In doing so, we note that any symmetric non-negative definite tensor $\theta \in \mbb{R}^{\mfrak{M}\times\mfrak{M}}$ with rank-$j$ has square matricization $\theta_{\text{sq}} \in \mbb{R}^{M\times M}$, $M = M_1 \dots M_D$, that may be written as $\mb{V}\mb{V}^T$, where $\mb{V} \in \mbb{R}^{M\times j}$. Note also that the columns of $\mb{V}$ are equal to the vectorized scaled eigentensors of $\theta$ (up to a rotation), allowing us to impose smoothness on these eigentensors through the columns of $\mb{V}$. The optimization problems in (a) thus reduce to the factorized matrix completion problems
\begin{align}
\label{eqn:fmp}
    \min_{\mb{V}\in \mbb{R}^{M \times j}} \norm{ \mcal{A}_{\text{sq}} \circ \left( (\hat{\mcal{C}})_{\text{sq}} - \mb{V}\mb{V}^T \right) }_F^2 + \alpha^* \mscr{P}_{\text{sq}}(\mb{V}),
\end{align}
for $j = 1, \dots, K^*$, where $\mscr{P}_{\text{sq}}(\mb{V}) = \mscr{P}(\theta)$ when $\theta = \mb{V}\mb{V}^T$. For our choice of $\mscr{P}$, we have $\mscr{P}_{\text{sq}}(\mb{V}) = \ip{\mb{P}}{\mcal{R}_{\text{sq}}\mb{P}}_F$ where $\ip{\cdot}{\cdot}_F$ denotes the Frobenius inner product. As these problems are not convex in $\mb{V}$, no gradient descent-type algorithm is guaranteed to converge to the problems' infinitely many global optima. In search of ``good'' local optima, we use MATLAB's \texttt{fminunc} function, which implements the Broyden–Fletcher–Goldfarb–Shanno (BFGS; \citealp{broyden-1970,fletcher-1970,goldfarb-1970,shanno-1970}) algorithm or, for larger problems, the limited-memory BFGS (L-BFGS; \citealp{liu-nocedal-1989}) algorithm. Indeed, when given a reasonable initialization, the simulations of Section \ref{sec:sim} show that these routines are stable and quickly converge to suitable local optima. We use the initial value $\mb{V}_0 = \mb{U}_j \mb{\Sigma}_j^{1/2}$, where $\mb{U} \mb{\Sigma} \mb{U}^T$ is the singular value decomposition of $(\hat{\mcal{C}})_{\text{sq}}$, $\mb{U}_j$ is the $M\times j$ matrix obtained by keeping the first $j$ columns of $\mb{U}$, and $\mb{\Sigma}_j$ is the $j\times j$ matrix obtained by keeping the $j$ rows and columns of $\mb{\Sigma}$.

\subsection{Post-Processing: Rotation and Shrinkage}
\label{sec:est-pp}

In Section \ref{sec:est-ile}, we assumed the loading tensors $\mcal{L}_k$ were equal to the scaled eigentensors of $\mcal{G}$. Although mathematically convenient, we did not impose this constraint in the original model as there is no practical justification for orthogonal loadings. Unfortunately, without such an assumption, it is not possible to identify the $\mcal{L}_k$ from $\mcal{C}$. We may still, however, hope to understand the latent factors that drive variation captured in the identifiable global component $\mcal{G}$. For instance, interesting neurological phenomena, like visual processing, may contribute variation to the global component, and an effective factor analytic approach should help us discover these latent processes. However, such insights can be difficult to glean via the scaled eigentensors, which are often dense with many high-magnitude weights. In our two-step post-processing procedure, we transform the illegible eigentensors into loadings possessing simple structure. The resulting parsimony greatly facilitates interpretation of the underlying factors. 

In the first post-processing step, we exploit the model's rotational indeterminacy by rotating the scaled eigentensors of $\hat{\mcal{G}}$ to a maximally-simple structure. We do so using varimax rotation \citep{kaiser-1958} which solves the optimization problem
\begin{align*}
    \hat{\mfrak{L}}^* := \underset{\mcal{M}_{\mb{R}}\hat{\mfrak{L}} \ :\  \mb{R} \in \text{SO}(K)}{\text{arg max}} \sum_{k=1}^K \norm{\left( \mcal{M}_{\mb{R}} \hat{\mfrak{L}} \right)_{\cdot, k}^2}_F^2 - \norm{\left( \mcal{M}_{\mb{R}} \hat{\mfrak{L}} \right)_{\cdot, k}}_F^4,
\end{align*}
where $\hat{\mfrak{L}} \in \mbb{R}^{\mfrak{M}} \times \mbb{N}_K$ is defined by $\hat{\mfrak{L}}_{\und{m},k} = (\hat{\mcal{L}}_k)_{\und{m}}$, $\mcal{M}_{\mb{R}} : \mbb{R}^{\mfrak{M}} \times \mbb{N}_K \to \mbb{R}^{\mfrak{M}} \times \mbb{N}_K$ is an operator given by $(\mcal{M}_{\mb{R}}\mfrak{L})_{\und{m},k} = (\mb{R}\mfrak{L}_{\und{m},\cdot})_k$, and $\text{SO}(K)$ is the set of $K\times K$ rotation matrices. 

Although the rotated loading estimate $\hat{\mfrak{L}}^*$ may be sufficiently interpretable, it sometimes helps to shrink near-zero regions to zero. The second post-processing step accomplishes this systematically by adaptively soft-thresholding $\hat{\mfrak{L}}^*$, 
\begin{equation}
\label{eqn:ast}
    \Tilde{\mfrak{L}}_{\und{m},k} := \text{sgn}(\hat{\mfrak{L}}^*_{\und{m},k}) \max \left\{\abs{\hat{\mfrak{L}}^*_{\und{m},k}} - \kappa_k w \left( \hat{\mfrak{L}}^*_{\und{m},k} \right), 0 \right\},
\end{equation}
where $w(\cdot)$ is a weight function with positive support (e.g., $w(x) = \abs{x}^{-2}$), and the $\kappa_k > 0$ are shrinkage parameters, which may be tuned using the cross-validation procedure described in Section 3 of Supplement A. Beyond enhancing interpretability, we show in Section \ref{sec:sim-int} that shrinkage can also correct an overspecified model by zeroing out ``extra'' loading estimates.

\section{MELODIC Background}
\label{sec:ica}

In Sections \ref{sec:sim} and \ref{sec:da}, we compare our approach to, among others, one based on ICA, the most popular whole-brain connectivity method. In doing so, we rely on the MELODIC function (\citealp{beckmann-smith-2004}) of the FMRIB Software Library (FSL; \citealp{jenkinson-etal-2012}), arguably ICA's most widely used implementation for connectivity analyses. This section details the MELODIC model and its estimation.

FSL's MELODIC function estimates the probabilistic ICA model which assumes that observations for each voxel are generated from a set of statistically independent non-Gaussian sources via a linear mixing process corrupted by additive Gaussian noise: 
\begin{equation*}
    \und{x}_m = \mb{A}\und{s}_m + \und{\mu} + \und{\eta}_m, \ m = 1, \dots, M.
\end{equation*}
Here, $\und{x}_m$ denotes the $n$-dimensional vector of observations at voxel $m$, $\und{s}_m$ denotes the $K$-dimensional vector of non-Gaussian sources, $\mb{A}$ denotes the ($n\times K$)-dimensional \textit{mixing matrix}, $\und{\mu}$ denotes the mean of the observations $\und{x}_m$, and $\und{\eta}_m$ denotes Gaussian noise. The model assumes there are fewer sources than observations (i.e., $K < n$), and that the noise covariance is voxel-dependent (i.e., $\und{\eta}_m \sim N(\und{0}, \sigma^2\mb{\Sigma}_m)$). 

The goal of estimation is to find the $(K\times n)$-dimensional \textit{unmixing matrix} $\mb{W}$ such that $\hat{\und{s}}_m = \mb{W}\und{x}_m$ is a good approximation to $\und{s}_m$ for each voxel. Prior to invoking MELODIC, it is common to spatially smooth the data. We do so using the Gaussian filter of FSL's FSLMATHS utility (\citealp{jenkinson-etal-2012}), tuning the kernel width via cross-validation. Before estimation, MELODIC prepares the multi-session data by (i) temporally concatenating the preprocessed slices to form observation vectors $\und{x}_m$ for each voxel, (ii) temporally whitening each $\und{x}_m$ so that the noise covariance is isotropic at each voxel location, then (iii) normalizing each $\und{x}_m$ to zero mean and unit variance. By default, the utility also reduces the dimension of the concatenated data via MELODIC Incremental Group PCA (MIGP; \citealp{smith-etal-2014}); however, we switch off this feature as Sections \ref{sec:sim} and \ref{sec:da} consider datasets of sufficiently reduced sizes. After this preprocessing, one can show that the maximum likelihood (ML) mixing matrix estimate is 
\begin{equation*}
    \hat{\mb{A}} = \mb{U}_K (\mb{\Lambda}_K - \sigma^2 \mb{I}_K)^{1/2} \mb{Q}^T,
\end{equation*}
where $\mb{U}_K$ and $\mb{\Lambda}_K$ contain the first $K$ left singluar vectors and singular values of $\mb{X}$, the $n\times M$ matrix of preprocessed data, and $\mb{Q}$ is some $K\times K$ rotation matrix. From $\hat{\mb{A}}$, the ML source estimates are obtained via generalized least squares, 
\begin{equation*}
    \hat{\und{s}}_m = (\hat{\mb{A}}^T \hat{\mb{A}})^{-1} \hat{\mb{A}}^T \Tilde{\und{x}}_m,
\end{equation*}
where $\Tilde{\und{x}}_m = (\mb{\Lambda}_K - \sigma^2 \mb{I}_K)^{-1/2} \mb{U}_K^T \und{x}_m$. Thus, to fix $\hat{\mb{A}}$, one chooses the $\mb{Q}$ that maximizes the non-Gaussianity of the source estimates $\hat{\und{s}}_m$.

For a random variable $s$, one measure of non-Gaussianity is \textit{negentropy}, 
\begin{equation*}
    J(s) = H(s_{\text{gauss}}) - H(s),
\end{equation*}
where $H(y) = -\mbb{E}_y[\log y]$ is the entropy of a random variable $y$ and $s_{\text{gauss}}$ is a Gaussian having the same variance (or covariance if $s$ is a random vector) as $s$. Since a Gaussian random variable has the largest entropy among all random variables of equal variance, negentropy is a non-negative function that equals zero if and only if $s$ is Gaussian. As estimation of negentropy is difficult, it is often approximated using the contrast function  
\begin{equation*}
    J(s) \approx \left( \mbb{E}[G(s)]  - \mbb{E}[G(\nu)]\right)^2,
\end{equation*}
where $\nu$ is a standard normal and $G$ is a general non-quadratic function. To estimate the $k$th source $s_{mk}$ of the $m$th voxel MELODIC uses a fixed-point iteration scheme \citep{hyvarinen-oja-2000} to choose the $k$th row $\und{q}_k^T$ of $\mb{Q}$ so that $J(\hat{s}_{mk}) = J(\und{q}_k^T\Tilde{\und{x}}_m)$ is maximized for some domain-specific choice of $G$. By default, MELODIC sets $G$ equal to the \textit{pow3} function. To estimate $K$ sources, the procedure is repeated $K$ times under the constraint that the vectors $\und{q}_k$ are mutually orthogonal.

\section{Simulation Studies}
\label{sec:sim}

To assess the efficacy of our methodology, we conduct three simulation studies. In the first, we compare the accuracy of the post-processed estimator of (\ref{eqn:g-def}) (denoted by FFA) to that of four other estimators. In the second, we assess the relative interpretability of these estimators. In the third, which may be found in Section 4 of Supplement A, we explore how the scree plot approach for selecting the number of factors behaves in different settings. We begin, in Section \ref{sec:sim-data}, by describing the data generating procedure for these studies. 

\subsection{Data Simulation}
\label{sec:sim-data}

To simulate data from an FFM (when $D=2$), we first generate $n$ i.i.d. mean-zero functions $Y_i$ and $n$ i.i.d. mean-zero functions $\epsilon_i$ on an equally spaced $M\times M$ lattice in $[0,1]^2$, summarizing these discretely-observed functions with the order-$2$ tensors (i.e., matrices) $\mcal{Y}_i$ and $\mcal{E}_i$, respectively. We then sum these components to get the samples $\mcal{X}_i = \mcal{Y}_i + \mcal{E}_i$, and compute the empirical covariance tensor, $\hat{\mcal{C}} = n^{-1} \sum_{i=1}^n \mcal{X}_i \otimes \mcal{X}_i$. For simplicity, in all studies, we set $M = M_1 = M_2 = 30$ and $\delta = \delta_1 = \delta_2$. Though these simulated ``brain slices'' are lower resolution than those of the AOMIC data whose axial slices have resolution $\und{M} = (65,77)$, this simplification serves only to reduce compute time and has little impact on estimator performance. 

To simulate the functions $Y_i$, we set $Y_i(\und{s}) = \sum_{k=1}^K f_{i,k} c_k v_k(\und{s})$, where the $v_k$ are (possibly non-orthogonal) functions scaled to have unit norm, the $c_k$ are positive constants, and the factors $f_{i,k}$ are drawn $i.i.d.$ from $N(0,1)$. We decompose the $k$th loading function into its standardization $v_k$ and a scaling parameter $c_k$ to control the difficulty of the estimation problem. Also note that the $(c_k, v_k)$ are not necessarily equal to the eigen-pairs $(\lambda_k, \eta_k)$ of $\mscr{G}$, as we do not require the $v_k$ to be orthogonal. Across the three studies, we make use of three loading schemes, depicted in Figure \ref{fig:load-schemes}. In the first loading scheme (denoted by BL1), each $v_k$ is defined piecewise with a bump function in two regions of its domain, and zero elsewhere. Each $v_k$ in the second loading scheme (denoted by NET) is also defined with piecewise bump functions, but these functions are arranged in the types of complex patterns one might expect to observe in a brain network. In the third loading scheme (denoted by BL2), each $v_k$ is defined by three bump functions, one of which is wider than the other two. 

\begin{figure}[!h]
\centering
\subfloat[BL1]{\label{fig:load-schemes-bl1}\includegraphics[width=0.35\linewidth]{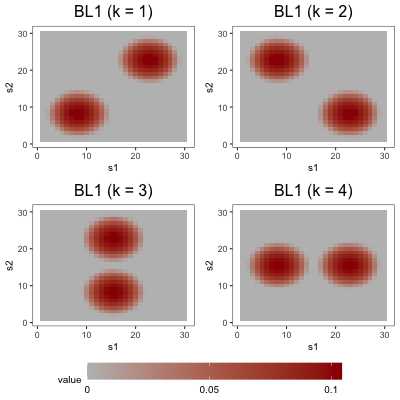}}\qquad
\subfloat[NET]{\label{fig:load-schemes-net}\includegraphics[width=0.35\linewidth]{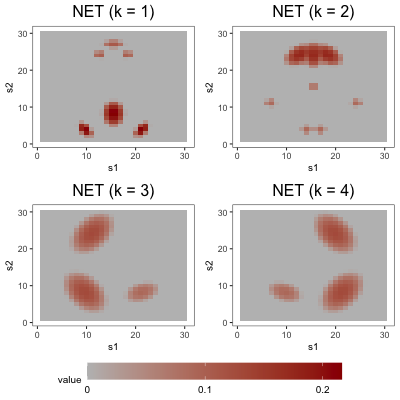}}\\
\subfloat[BL2]{\label{fig:load-schemes-bl2}\includegraphics[width=0.7\linewidth]{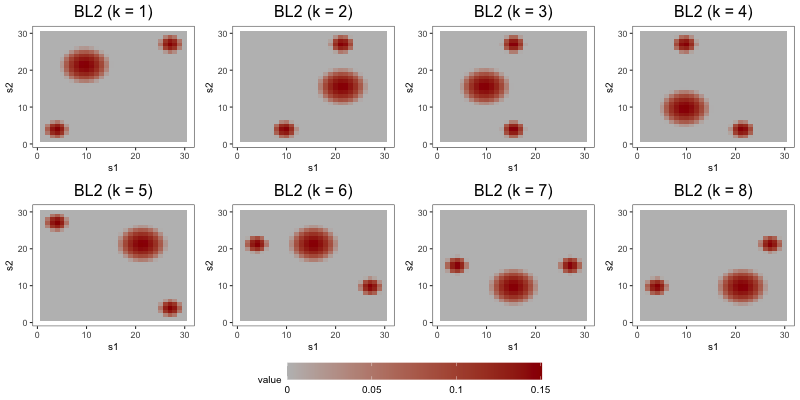}}
\caption{The three loading schemes used across studies: (a) each loading contains a bump function in two regions of its domain; (b) each loading contains bump functions arranged in patterns more complex and varied than those of the BL scheme; (c) each loading contains a bump function in three regions of its domain.}
\label{fig:load-schemes}
\end{figure}

To generate the functions $\epsilon_i$, we set $\epsilon_i(\und{s}) = \sum_{j=1}^J a_{i,j} d_j e_{j}(\und{s})$, where the $e_j$ are (again, possibly non-orthogonal) functions with support on $[(p-1)\delta, p\delta] \times [(q-1)\delta, q\delta] \subset [0,1]^2$, the $d_j$ are positive constants, the $a_{i,j}$ are drawn i.i.d. from $N(0,1)$, and $J \gg K$ is large. As in the loading scheme framework, the $(d_j, e_j)$ may not be equal to the eigen-pairs $(\beta_j, \psi_j)$ of $\mscr{B}$ since the former need not be orthogonal. We consider two error schemes. Each defines $J$ functions $e_j$ with possibly overlapping supports on an equally-spaced grid of $[0,1]^2$. In the first (denoted by BE), each $e_j$ is a two-dimensional bump function. In the second (denoted TRI), each $e_j$ is a two-dimensional triangle function. 

\subsection{Study 1: Accuracy Comparison}
\label{sec:sim-acc}

The aim of the first study is to compare the accuracy of the FFA estimator to that of five alternative estimators: 

\begin{enumerate}
    \item Independent component analysis without smoothing (denoted by ICA): we estimate $\mcal{G}$ by extracting from the data an $M\times M\times K$ tensor $\hat{\mcal{S}}$ of independent components via MELODIC, then computing $\sum_k \hat{\mcal{S}_k} \otimes \hat{\mcal{S}_k}$ where the $M\times M$ matrix $\hat{\mcal{S}_k}$ is the $k$th estimated independent component. We skip variance normalization since the simulated data do not reproduce phenomena like scanner and physiological noise that lead to different levels of variability across different voxels. Omission of this step also ensures that ICA estimates are on the same scale as those derived from other estimators.  
    \item Independent component analysis with smoothing (denoted by ICAS): we estimate $\mcal{G}$ using the same procedure as ICA but smooth the data with a Gaussian filter as described in Section \ref{sec:ica}.
    \item Principal component analysis (denoted by PCA): we estimate $\mcal{G}$ with a truncated eigen-decomposition of $\hat{\mcal{C}}$. Note that this approach corresponds with principal component estimation in multivariate factor analysis (\citealp[Section 9.3]{johnson-wichern-2002}). 
    \item Matrix completion approach, presented by \cite{descary-panaretos-2019} (denoted by DP): we compute the estimator of (\ref{eqn:g-def}) fixing $\alpha = 0$. 
    \item Matrix completion approach with smoothing (denoted by DPS): we compute the estimator of (\ref{eqn:g-def}) using $\alpha \geq 0$, but omit post-processing. 
\end{enumerate}
We study these estimators in four scenarios, which are defined by the Cartesian product of two loading schemes (BL1 and NET) and the two error schemes: S1 for BL1 and BE, S2 for BL1 and TRI, S3 for NET and BE, and S4 for NET and TRI. Within each scenario, we consider 24 configurations, each of which is characterized by a sample size $n \in \{250, 500, 1000\}$, a number of factors $K \in \{ 2,4 \}$, a bandwidth $\delta \in \{ 0.05, 0.1\}$, and a ``regime'' (R1 or R2). In both regimes, the $d_j$ are distributed uniformly in $[0.1, 1]$, while the $c_k$ are distributed uniformly in $[2,3]$, and $[0.8, 1.8]$ for R1 and R2, respectively. Consequently, the lower-signal R2 regime poses a more difficult estimation problem than the R1 regime. Note that in the multi-subject AOMIC analysis of Section \ref{sec:da}, the sample size is the number of time points per subject multiplied by the number of subjects ($480 \times 210 = 10,800$). Accordingly, this simulation study evaluates our model in settings that are likely lower-signal than that of our application. 

Per Section \ref{sec:est}, FFA estimation -- including initial loading estimation and post-processing -- involves two levels of tuning. We consider a smoothing parameter $\alpha$ and just one shrinkage parameter $\kappa$ for all loading functions to lighten computation. Instead of conducting the expensive cross-validation procedures of Section 3 in Supplement A, we simply split the data into a training set (80\%) and a test set (20\%), then choose the $\alpha$ (similarly, $\kappa$) that minimizes the out-of-sample prediction error for the band-deleted empirical covariance. Throughout the study, we set $\delta = 0.1$ for estimation procedures.

For 100 repetitions of each configuration, we compute the six estimators -- FFA, DPS, DP, PCA, ICA, and ICAS -- then calculate normalized errors by evaluating $E(x) = \| \mcal{G} - x \|_F / \| \mcal{G} \|_F$ at each estimator. To assess the comparative accuracy of the FFA estimator, we form five relative errors by dividing the normalized errors for the DPS, DP, PCA, ICA, and ICAS estimators by that for the FFA estimator. Figures \ref{fig:comp-ffa-r1} and \ref{fig:comp-ffa-r2} display results for R1 and R2, respectively.

FFA is in league with or superior to the alternatives in nearly every configuration. DPS, however, consistently competes with FFA, indicating that post-processing leads to only modest accuracy gains (Section \ref{sec:sim-int} considers the more noteworthy interpretability gains afforded by post-processing). The smoothing of FFA and DPS provides helpful de-noising in BL configurations -- particularly in the low-signal regime and/or when $n$ is small -- but is less effective in NET configurations. In the latter, over-smoothing occasionally forces the optimization procedure into bad local minima. This is likely because application of a single smoothing parameter to each loading tensor is more appropriate in the BL scheme, which has uniformly smooth loading tensors, than in the NET scheme, whose first and second loadings are spikier than its third and fourth loadings. 

DP, PCA, ICA, and ICAS typically trail FFA and DPS. It is not surprising that the performances of ICA and PCA are in near lockstep, given that MELODIC projects the temporally whitened data onto the space spanned by the $K$ leading eigenvectors of $\hat{\mcal{C}}$, and that MELODIC's temporal whitening has little effect on these i.i.d. data. The data smoothing of ICAS consistently improves upon ICA, most noticeably in regime R2. In this regime, ICAS performance even matches that of FFA and DPS when $\delta = 0.05$ configurations, though it falls short in $\delta = 0.1$ configurations where data smoothing smears more local variation into the global component. Contamination of the global component brought about by data smoothing is so pronounced that even DP, an estimator that does not explicitly encourage smooth loadings, outpaces ICAS in some of these $\delta = 0.1$ configurations.

\begin{figure}[!h]
    \centering
    \includegraphics[scale=0.5]{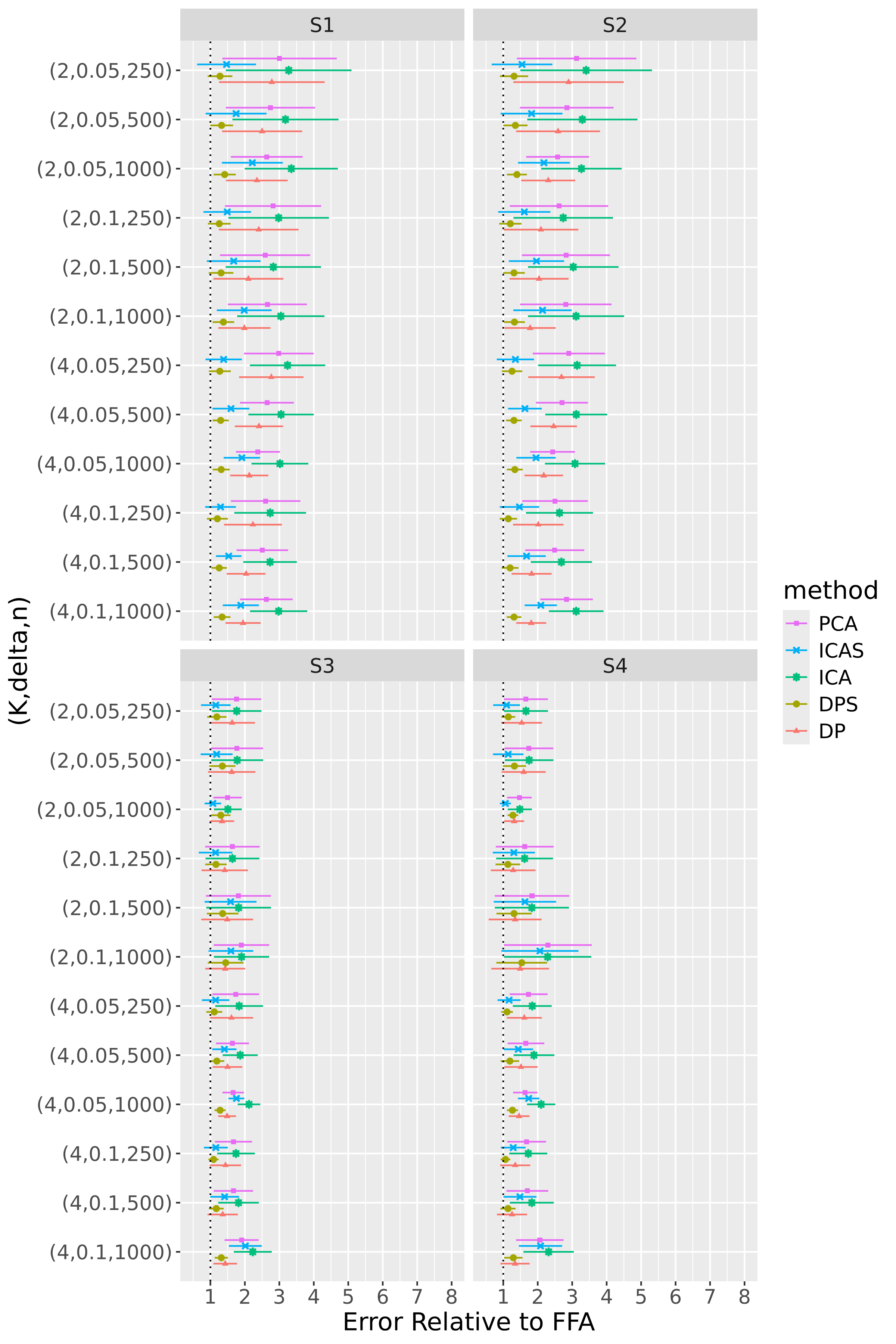}
    \caption{Mean relative errors (with $\pm 2$-standard-deviation error bars) for all Regime 1 configurations using the FFA estimator as baseline.}
    \label{fig:comp-ffa-r1}
\end{figure}

\begin{figure}[!h]
    \centering
    \includegraphics[scale=0.5]{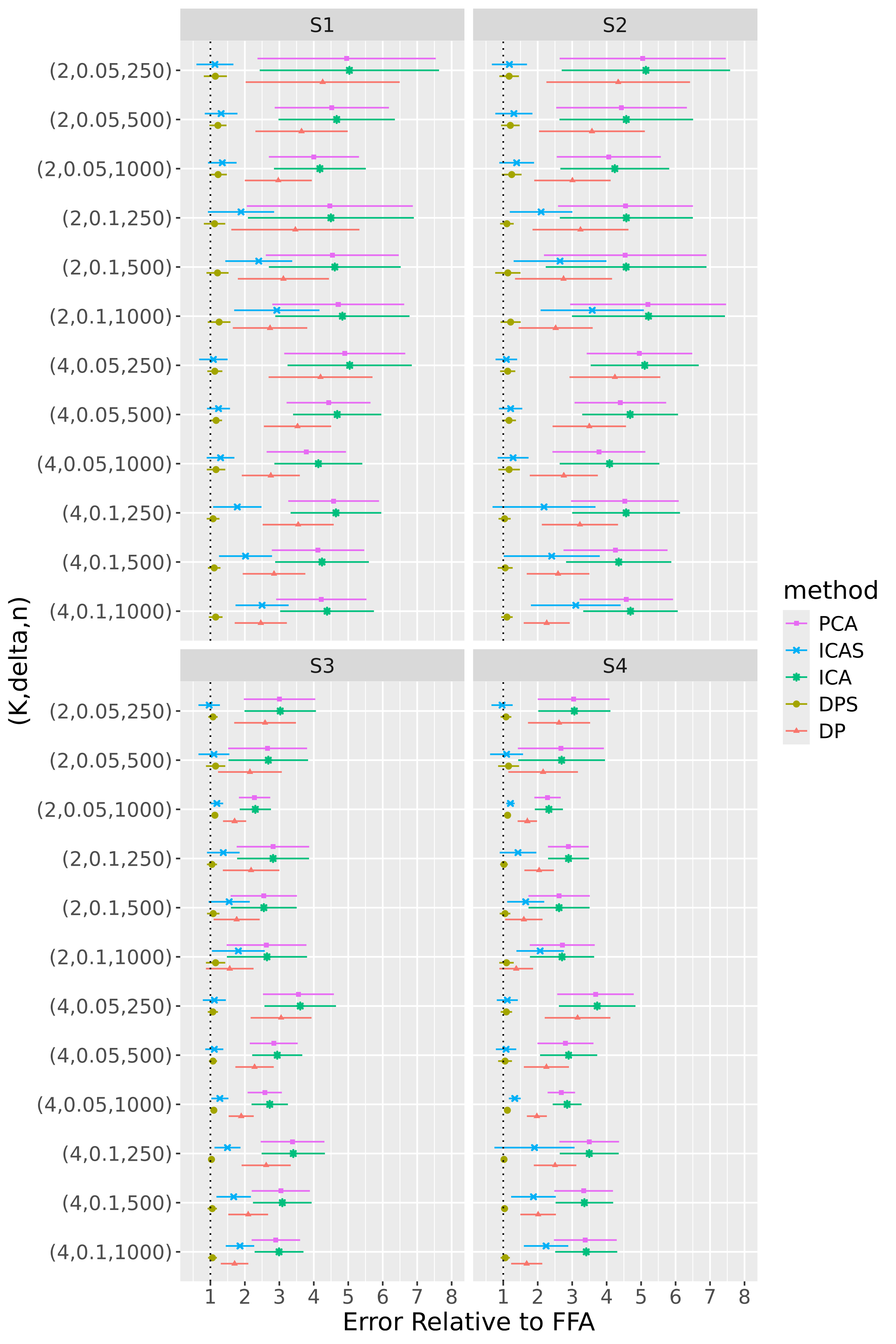}
    \caption{Mean relative errors (with $\pm 2$-standard-deviation error bars) for all Regime 2 configurations using the FFA estimator as baseline.}
    \label{fig:comp-ffa-r2}
\end{figure}

\subsection{Study 2: Interpretability Comparison}
\label{sec:sim-int}

The aim of the second study is to compare the interpretability of FFA to that of DPS and ICAS, the two most competitive estimators from Section \ref{sec:sim-acc}. To do so, we focus on output from a single repetition of a configuration with loading scheme BL2, error scheme BE, regime R2, $n$ set to 1000, $K$ set to 8, and $\delta$ set to 0.1. Consideration of only one repetition alleviates computational strains, allowing us to tune component-specific shrinkage parameters for the FFA estimator. 

In the first part of this study, we correctly specify model order by estimating the three models with 8 components (see Figure \ref{fig:sim-K8}). DPS fails to capture the simple structure of the true loadings. ICAS, on the other hand, largely accomplishes this, identifying 6 of the true loadings. However, it lacks the sparsity of FFA, which clearly capture 7 to 8 of the true loadings.  

\begin{figure}[!h]
\centering
\subfloat[FFA]{\label{fig:sim-K8-ffa}\includegraphics[width=0.6\linewidth]{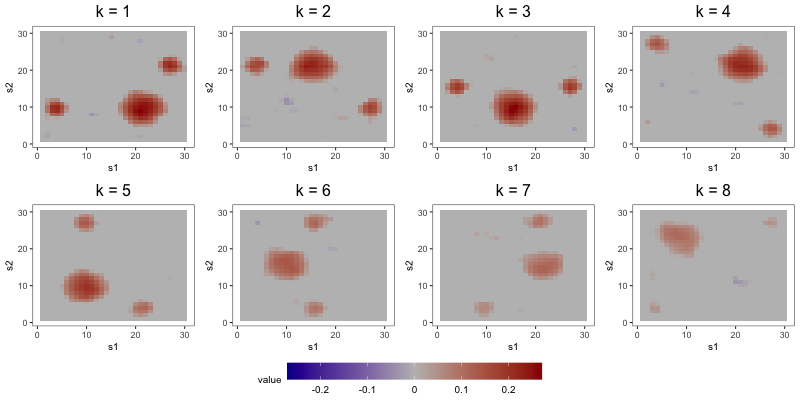}}\\
\subfloat[ICAS]{\label{fig:sim-K8-ICAS}\includegraphics[width=0.6\linewidth]{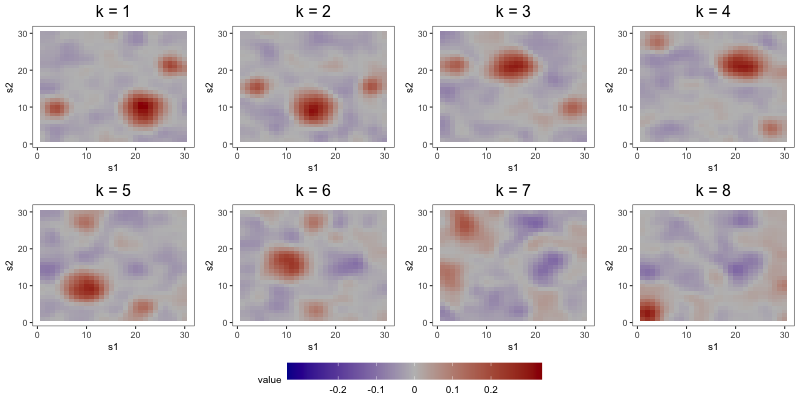}}\\
\subfloat[DPS]{\label{fig:sim-K8-dps}\includegraphics[width=0.6\linewidth]{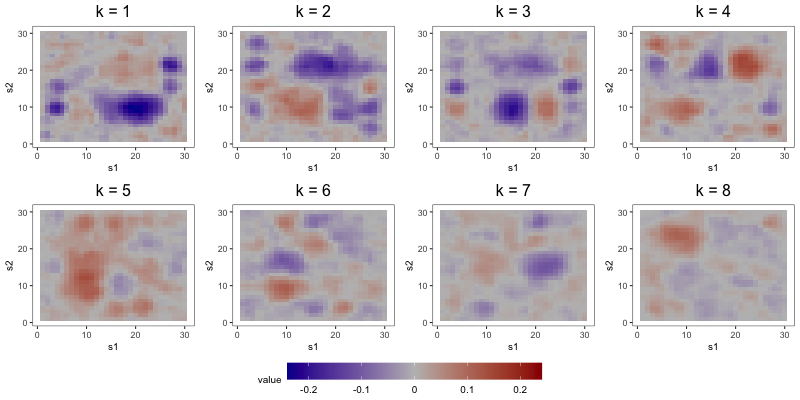}}
\caption{Estimates for order-8 FFA, DPS, and ICAS models when loading scheme is BL2, error scheme is BE, regime is R2, $n = 1000$, $K = 8$, and $\delta = 0.1$.}
\label{fig:sim-K8}
\end{figure}

The most noteworthy difference, however, between FFA and ICAS is that ICAS results are less robust to dimension overspecification. That is, if we estimate two IC models, one with the true number of ICs and another with more, the latter will splinter structures found in single ICs of the former across multiple ICs. Such fragmentation does not occur in overspecified FFMs. To see this, consider output from order-25 FFA and ICAS models fit to these data (see Figures \ref{fig:sim-ffa-K25} and \ref{fig:sim-ica-K25}). The higher-order FFM preserves the 7 to 8 loadings of the lower-order FFM and even shrinks many of the other loadings to zero, effectively correcting overspecification. However, of the 6 true loadings captured in the 8-IC model, only two are preserved in the 25-IC model ($k=1, 2$ in Figure \ref{fig:sim-ica-K25}). Other structures in the 8-IC model are scattered across the remaining components of the higher-order model. This is because, in its search for maximally independent spatial maps, ICA fragments broad correlation structures across multiple maps \citep{friston-1998}. This splintering presents challenges to interpretation, particularly when the specified dimension far exceeds the true dimension. In contrast, varimax rotation in FFA does not fragment correlation structure. Moreover, shrinkage seems to have partially ``corrected'' the overspecification, zeroing out many of the extra loadings. This makes FFA an attractive approach when interpretation is the foremost goal and dimension overspecification is possible. 

\begin{figure}[!h]
    \centering
    \includegraphics[width=0.8\linewidth]{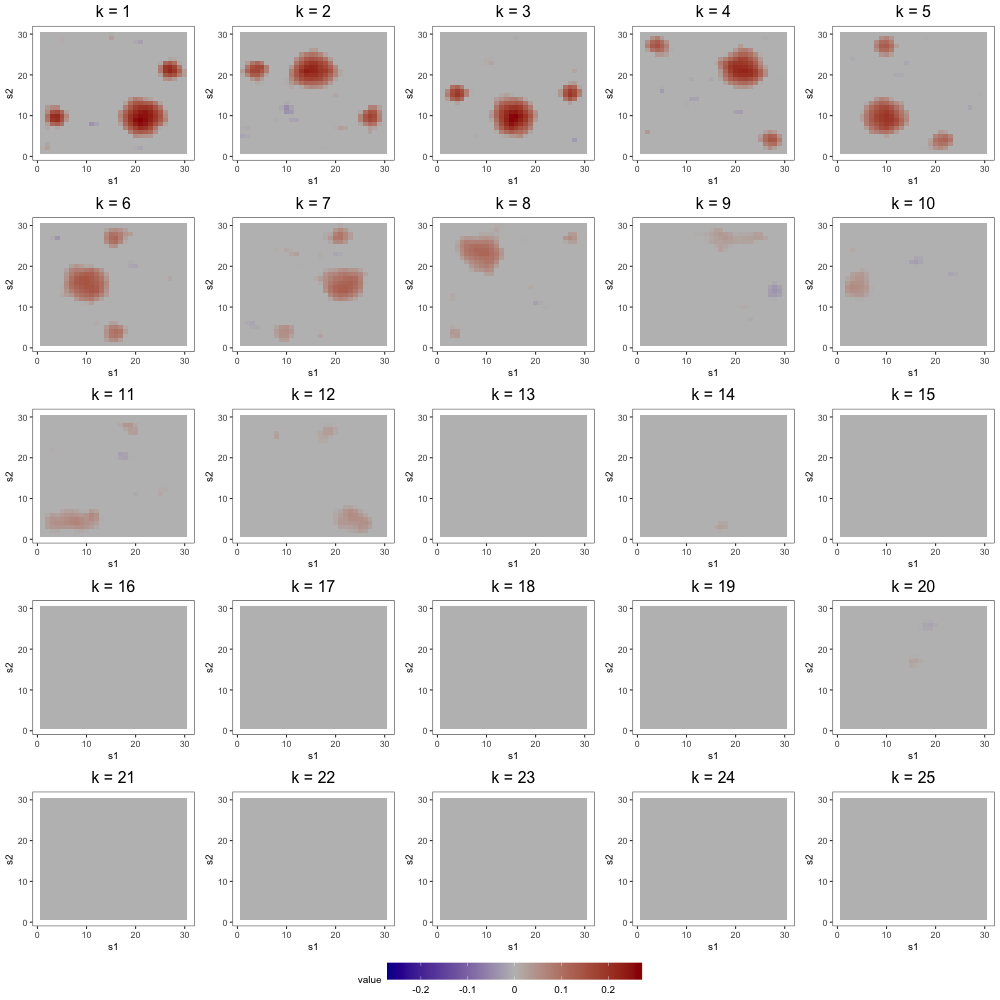}
    \caption{Estimates for order-25 FFA model when loading scheme is BL2, error scheme is BE, regime is R2, $n = 1000$, $K = 8$, and $\delta = 0.1$.}
    \label{fig:sim-ffa-K25}
\end{figure}

\begin{figure}[!h]
    \centering
    \includegraphics[width=0.8\linewidth]{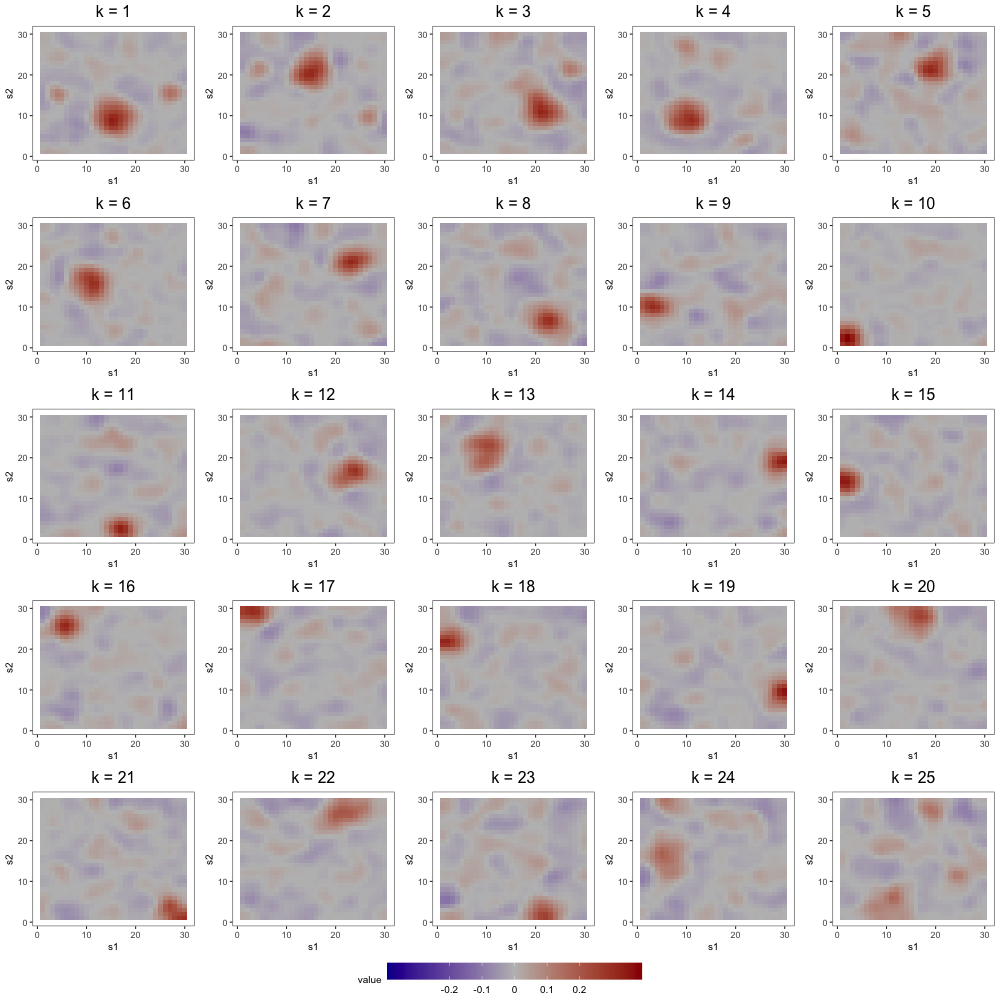}
    \caption{Estimates for order-25 ICA model when loading scheme is BL2, error scheme is BE, regime is R2, $n = 1000$, $K = 8$, and $\delta = 0.1$.}
    \label{fig:sim-ica-K25}
\end{figure}

\section{AOMIC Data Analysis}
\label{sec:da}

We now consider the AOMIC-PIOP1 dataset which includes six-minute resting-state scans for 210 healthy university students. Prior to each subject's resting-state session, a structural image was acquired at a resolution of 1 mm isotropic while the functional image was acquired at a spatial resolution of 3 mm isotropic and a temporal resolution of 0.75 seconds (resulting in a single-subject data matrix with dimension $65 \times 77 \times 60 \times 480$). Including all $\sim$ 300,000 voxels in this analysis would necessitate estimation and in-memory storage of a prohibitively large spatial covariance. We avoid this by focusing on one axial slice (at $z=30$). Preprocessing of anatomical and functional MRI were performed using a standard Fmriprep pipeline \citep{esteban-etal-2019}, as detailed in \cite{snoek-et-al-2021}.

Our strategy for discovering the connectivity patterns present in these resting-state data was to prewhiten the voxel time courses in each subject's preprocessed scan (which are, in general, nonstationary and autocorrelated), then treat the multi-subject collection of slices as an i.i.d. sample $\mcal{X}_1, \dots, \mcal{X}_{n}$, $n=100800$, which we used to fit the FFM. Although fMRI data is inherently nonstationary and autocorrelated, many connectivity studies that correlate voxel time series skip this essential prewhitening step, risking the discovery of spurious correlations \citep{christova-etal-2011, afyouni-etal-2019}. Of course, prewhitening only addresses temporal correlation within voxel time courses, not between time courses. There are several ways to prewhiten voxel time courses, including those based on Fourier \citep{laird-etal-2004} and wavelet \citep{bullmore-etal-2001} decompositions. Taking another common approach, we performed prewhitening by fitting a non-seasonal AutoRegressive Integrative Moving Average (ARIMA) model to each voxel time course using the R function \texttt{forecast::auto.arima} \citep{hyndman-khandakar-2008}, then extracting the residuals. Figure 5 from Section 6.1 of Supplement A demonstrates the effects of prewhitening on voxel time courses. Taking after \cite{christova-etal-2011}, we form samples $\mcal{X}_1, \dots, \mcal{X}_{n}$ from the scaled ARIMA residuals. 

We begin analysis of these samples by computing the empirical covariance $\hat{\mcal{C}}$ in equation (\ref{eqn:emp-cov}). Based on visual inspection and the fact that the sample size is quite large, we opt not to smooth (i.e., we set $\alpha = 0$), then plot the non-increasing function $g(j) = \|\mcal{A} \circ (\hat{\mcal{C}} - \hat{\theta}_j)\|_F^2$ for $j = 1, \dots, 100$ (see Figure 6 of Section 6.1 in Supplement A). Based on this plot, we set $\hat{K} = 12$ and use the corresponding $\hat{\mcal{L}}_1, \dots, \hat{\mcal{L}}_{12}$ as the initial loading estimates. To choose the shrinkage parameter $\und{\kappa} = (\kappa_1, \dots, \kappa_{12})$ used to compute the post-processed loadings $\Tilde{\mcal{L}}_1, \dots, \Tilde{\mcal{L}}_{12}$, we define a manageable search grid then follow the tuning procedure described in Section 3 of Supplement A. Figure \ref{fig:rs-ffa} displays the rich covariance structure captured by these post-processed loadings. In addition to anatomical segmentation (e.g., grey matter segmentation in loadings 1, 2, and 6; white matter segmentation in loadings 3 and 9; and ventricle segmentation in loadings 5 and 10), these 12 loadings contain known resting-state networks. The prominent bilateral structure in the first loading captures parts of the auditory resting-state network, including co-activation in Heschl's gyrus and the posterior insular (see $7_{20}$ of Figure 1 in \citealp{smith-etal-2009}). Next, frontal activation in the second loading includes areas associated with resting-state executive control, like the anterior cingulate (see $8_{20}$ of Figure 1 in \citealp{smith-etal-2014}). Lastly, the sixth loading contains activation in the occipital lobe, an area associated with visual processing (see $2_{20}$ of Figure 1 in \citealp{smith-etal-2009}).

To provide a methodological comparison, we also performed group ICA by temporally concatenating the preprocessed slices then running ICA on the resulting data to produce spatial maps common to all subjects \citep{calhoun-etal-2001}. As in Section \ref{sec:sim}, we do so using FSL's MELODIC function. For this analysis, we fit 12-, 25-, and 50-independent component (IC) models. Transformed components (for interpretability) for the 12-IC model are shown in Figure \ref{fig:rs-components}, and those for the 25- and 50-IC models are shown in Figures 8, 9, and 10 of Section 5.2 in Supplement A. To facilitate comparison with FFA, we also estimate the FFM with 25 and 50 factors (see Figures 5, 6, and 7 in Section 5.2 of Supplement A) even though the plots of Figure 6 in Section 6.1 of Supplement A suggest a smaller number of factors. 

We first compare the 12-component models (see Figure \ref{fig:rs-components}). Several prominent structures are captured by both FFA and ICA: dorsal correlations appear in the 6th loading, and in the 4th and 9th ICs; frontal correlations appear in the 2nd loading and the 7th IC; bilateral correlations appear in the 1st loading and in the 8th IC; similar ventricle correlations appear in the 5th loading and in the 1st IC; and brain outlines appear in the 4th and 7th loadings, and in the 11th IC. Despite the many common features, there exist some discrepancies (e.g., the spikes of activation in the 10th and 12th IC do not display prominently in any single loading).

\begin{figure}[!h]
\centering
\subfloat[FFA]{\label{fig:rs-ffa}\includegraphics[width=0.45\linewidth]{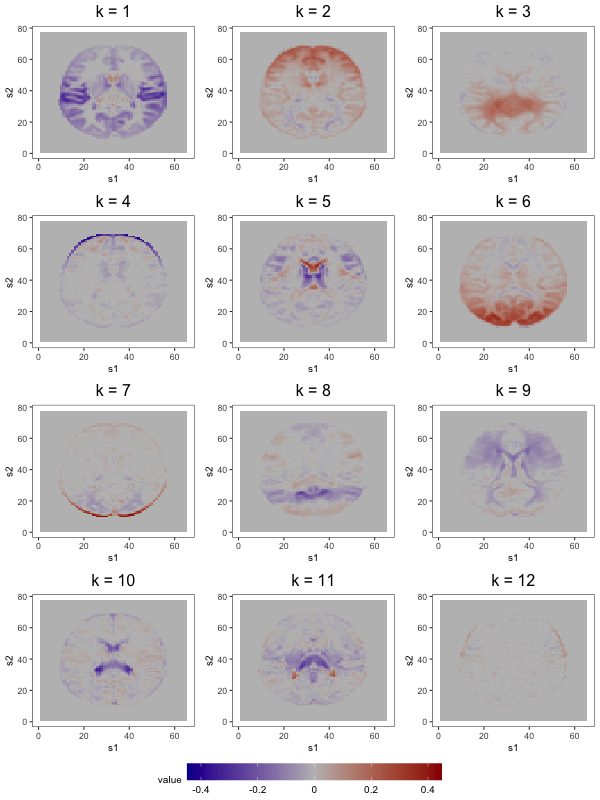}}\qquad
\subfloat[ICA]{\label{fig:rs-ica}\includegraphics[width=0.45\linewidth]{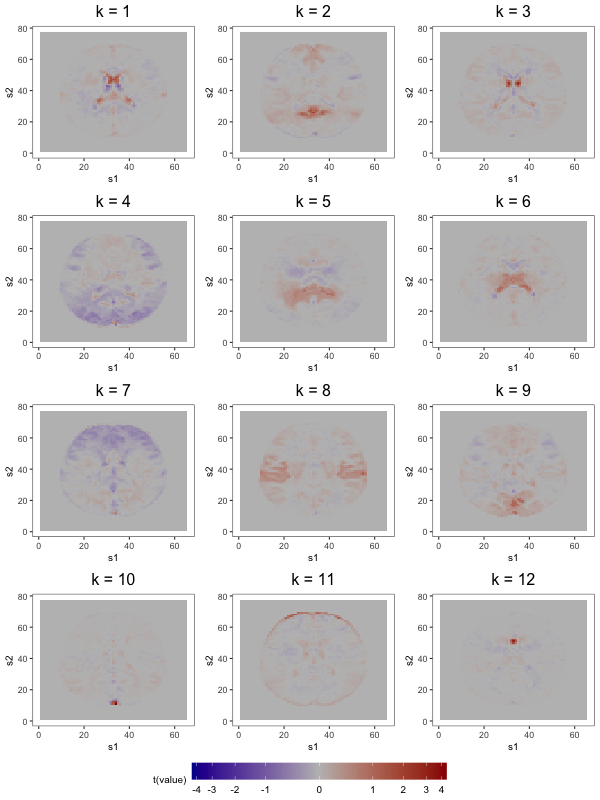}}
\caption{(a) Loadings of the functional 12-factor model estimated from ARIMA residuals. (b) Transformed components ($t(x) = \text{sgn}(x) \log (\abs{x} + 1)$) of the 12-component IC model estimated from the preprocessed scans. ICs are presented on the log-scale to bring forward correlations that are washed out by the spikes of activation in ICs 3, 10, and 12.}
\label{fig:rs-components}
\end{figure}

However, as observed in Section \ref{sec:sim-int}, the most practically meaningful difference is that MELODIC lacks FFA's robustness to dimension misspecification. Consider, for instance, the similar ventricle correlations in loading 5 and IC 1 of the order-12 plots in Figures \ref{fig:splint-ffa} and \ref{fig:splint-ica}. In the FFMs with 25 and 50 factors, this structure is preserved. However, this structure, captured in a single IC of the 12-IC model, is splintered across 2 ICs in both the 25- and 50-IC models, with the components of the latter containing less of the original structure than those of the former. Similar fragmentation occurs for other structures found in the 12-IC model.


\begin{figure}[!h]
\centering
\subfloat[FFA]{\label{fig:splint-ffa}\includegraphics[width=0.27\linewidth]{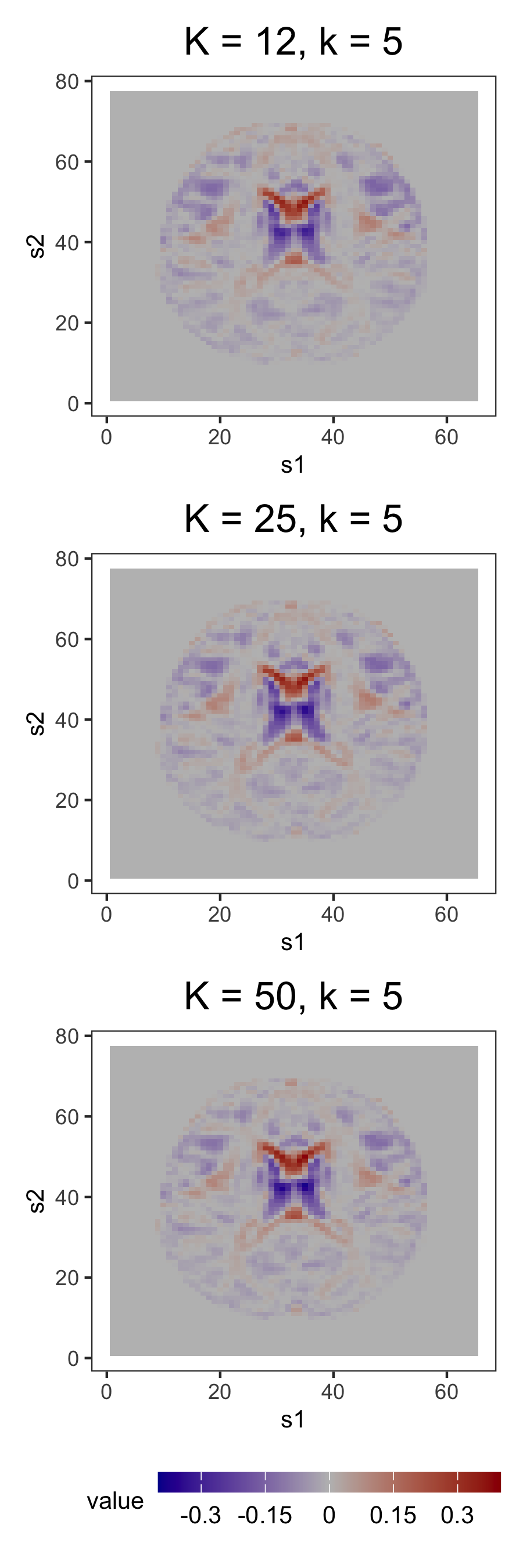}}\qquad
\subfloat[ICA]{\label{fig:splint-ica}\includegraphics[width=0.45\linewidth]{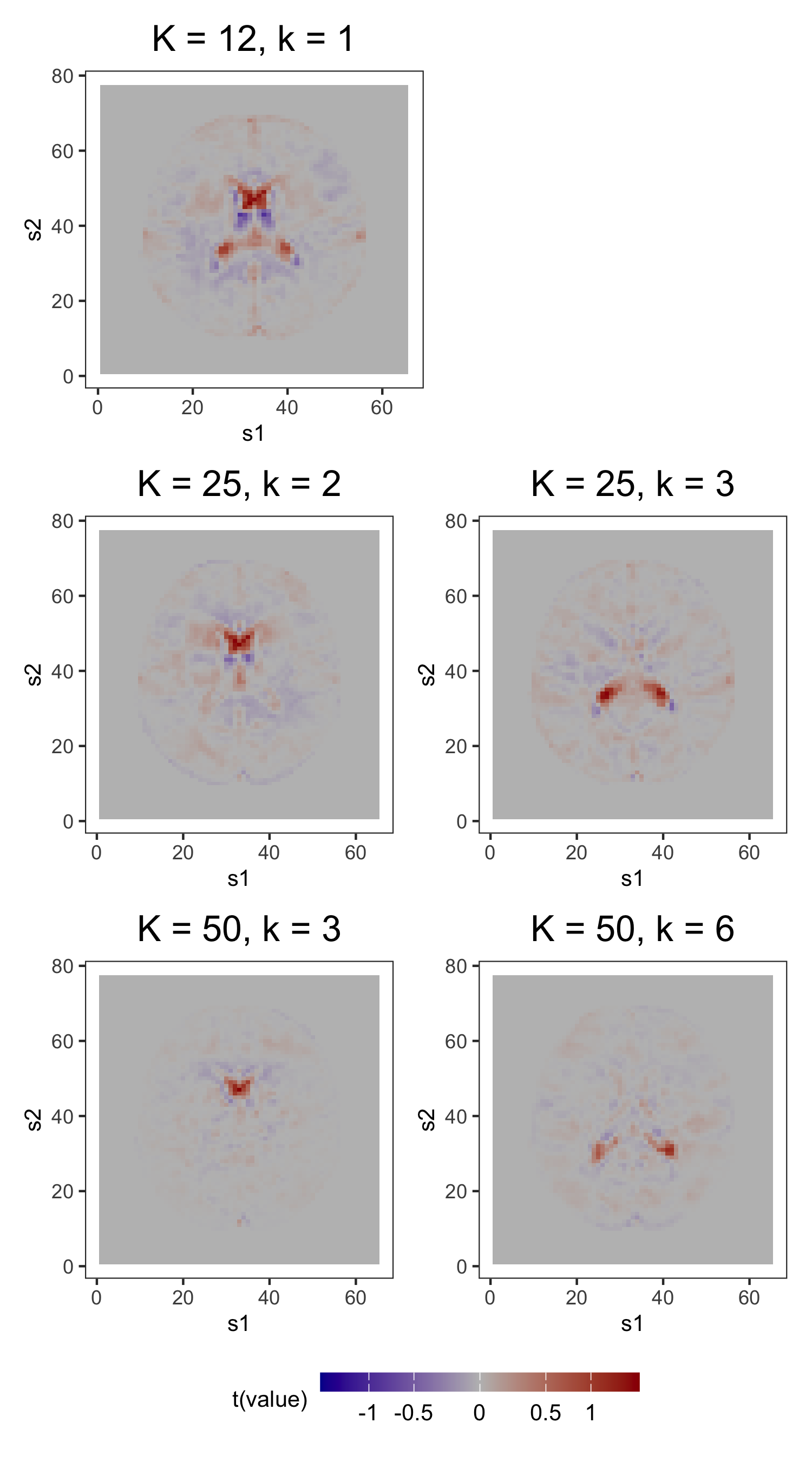}}
\caption{In their respective 12-component models, FFA and ICA capture correlation structures within the ventricles. In higher-order models, FFA preserves this structure while ICA fragments it across multiple ICs.}
\label{fig:splint}
\end{figure}

One way to combat ICA's instability is through multi-scale approaches. For example, \cite{iraji-etal-2023} estimate common spatial maps by first generating 100 half-splits of their multi-subject dataset, then fitting IC models of order 25, 50, 75, 100, 125, 150, 175, and 200 to each data split. This produces 100 sets of 900 components from which they select the 900 with the highest average spatial similarity (calculated by Pearson correlation) across the 100 sets. They then identify those ICs among the 900 that are most distinct from each other (spatial similarity < 0.8). In Section 6.2 of Supplement A, we adapt this procedure to the AOMIC data, showing that the final subset of estimated ICs does, in fact, preserve low-order structure. However, this requires more computing resources than a single-model ICA approaches. Moreover, this output lacks some of the interpretability afforded by that from a single-model approach: estimated ICs from this multi-scale algorithm are cherry-picked from many IC models, consequently collapsing the straightforward independence-based interpretation arising from the assumptions of a single IC model.

\section{Discussion}
\label{sec:discussion}

This work recognizes the central contribution of \cite{descary-panaretos-2019} -- a framework that decouples the global and local variation of a functional observation -- as the proper scaffolding upon which to build a factor analytic approach for functional data. Our work extends their framework to multidimensional functional data, enhances the estimator of the global covariance via a roughness penalty, then appends a post-processing procedure that improves interpretability. The result is a factor analytic approach tailored to the study of functional connectivity in fMRI. 

Three characteristics of our methodology make it uniquely suited to the study of functional connectivity. First, it outputs spatial maps possessing simple structure, which facilitates interpretation of the latent factors that drive correlations in the data. It is a hallmark of factor analytic methods and has, in fields like education and psychology, made multivariate factor analysis a popular alternative to less legible techniques, like PCA. We demonstrate the utility of this structure in the neuroimaging context. Second, our approach is not preceded by data smoothing like other functional connectivity methods. Given that connectivity studies aim to characterize correlations between distant brain regions, model estimation should exclude local variation, yet data smoothing forbids this. It smears local variation into the global signal, thereby corrupting subsequent attempts to estimate this signal. Third, our approach is robust to dimension overspecification. This model feature, not found in popular ICA-based methods, is critical in fMRI where it is difficult to correctly specify the dimension of the signal. 

The limitations of our methodology present several opportunities for future work. One weakness of the approach is that its application does not consider all spatial information available in an fMRI scan. Although the model remains valid for 3-dimensional brain volumes, inclusion of a third spatial dimension renders the existing estimation method computationally impractical. Additionally, the model ignores the temporal information present in fMRI data. We attempted to mitigate the impact of these temporal dynamics in our analysis through prewhitening, but such dependencies ought to be acknowledged by the model. Future work on functional factor modeling of large-scale functional or spatiotemporal data should address these concerns.


\pagebreak
\begin{center}
\textbf{\LARGE Supplement A: Functional Factor Modeling of Brain Connectivity}
\end{center}
\setcounter{section}{0}
\setcounter{equation}{0}
\setcounter{figure}{0}
\setcounter{table}{0}
\makeatletter
\vspace{1em}

\section{Notation and Definitions}
\label{nd}

\subsection{Tensor Reshaping}
\label{nd-tr}

In the discrete observation paradigm, it is useful to define some \textit{tensor reshapings}, or bijections between the set of indices of one tensor and that of another tensor. Denote the index set of $\mcal{X} \in \mbb{R}^{\mfrak{M}}$ by $\mbb{N}_{\und{M}} = \mbb{N}_{M_1} \times \dots \times \mbb{N}_{M_D}$. Define a bijection $\mu: \mbb{N}_{\und{M}} \to \mbb{N}_{M_1 \dots M_D}$ so that the \textit{vectorization} of $\mcal{X} \in \mbb{R}^{\mfrak{M}}$ is the vector 
\begin{align*}
    \mcal{X}_{\text{vec}} := [v_{\mu^{-1}(m)}]_{m=1}^{M_1\dots M_D}. 
\end{align*}
Next, define a symmetric bijection $\gamma: \mbb{N}_{\und{M}} \times \mbb{N}_{\und{M}} \to \mbb{N}_{M_1 \dots M_D} \times \mbb{N}_{M_1 \dots M_D}$ so that the \textit{square matricization} of $\mcal{C} \in \mbb{R}^{\mfrak{M}\times\mfrak{M}}$ is the square matrix  
\begin{align*}
    \mcal{C}_{\text{sq}} := [w_{\gamma^{-1}(i, j)}]_{i,j = 1}^{M_1 \dots M_D}. 
\end{align*}
Although there are many suitable selections for $\mu$ and $\gamma$, we choose these bijections to be consistent with the array reshaping operator in MATLAB. 

\subsection{Properties of Tensors}
\label{nd-pot}

We use the tensor reshapings of Section \ref{nd-tr} to define several properties of tensors. We say $\mcal{C} \in \mbb{R}^{\mfrak{M}\times \mfrak{M}}$ is \textit{symmetric} if $\mcal{C}_{\text{sq}}$ is symmetric, \textit{nonnegative definite} if $\mcal{C}_{\text{sq}}$ is nonnegative definite, and \textit{rank-$r$} if $\mcal{C}_{\text{sq}}$ is rank-$r$. Through square matricization and vectorization, it can be shown that any symmetric nonnegative definite tensor $\mcal{C} \in \mbb{R}^{\mfrak{M}\times\mfrak{M}}$ with rank-$r$ admits the eigendecomposition
\begin{align*}
    \mcal{C} = \sum_{k=1}^r c_k \mcal{A}_k \otimes \mcal{A}_k,
\end{align*}
where the $c_k$ are the nonnegative \textit{eigenvalues} of $\mcal{C}$ and the orthonormal $\mcal{A}_k$ are the \textit{eigentensors} of $\mcal{C}$.

\section{Rotational Indeterminacy}
\label{ri}

The FFM is identifiable up to an orthogonal rotation. To see this, let $\mfrak{L} \in [0,1]^D \times \mbb{N}_K$, defined by $\mfrak{L}_{\und{s},k} = l_{k}(\und{s})$, be a set of loading functions indexed by $k$. For a $K \times K$ rotation matrix $\mb{R}$, define the rotation operator $\mcal{M}_{\mb{R}}: [0,1]^D \times \mbb{N}_K \to [0,1]^D \times \mbb{N}_K$ by $(\mcal{M}_{\mb{R}} \mfrak{L})_{\und{s},k} = (\mb{R}\mfrak{L}_{\und{s},\cdot})_k$. That is, $\mcal{M}_{\mb{R}}$ rotates the loading functions by rotating the $K$-dimensional vector $(l_1(\und{s}), \dots, l_K(\und{s}))$ by $\mb{R}$ at each fixed $\und{s} \in [0,1]^D$. If $\und{f} = (f_1, \dots, f_K)$ is a vector of factors and $\mb{R}$ is some rotation matrix, then the following equivalences hold:
\begin{align*}
    Y(\und{s}) = \sum_{k=1}^K \mfrak{L}_{\und{s},k} \und{f}_k = \sum_{k=1}^K (\mcal{M}_{\mb{R}}\mfrak{L})_{\und{s},k} (\mb{R}\und{f})_k
\end{align*}
and
\begin{align*}
    \mscr{G} = \sum_{k=1}^K \mfrak{L}_{\cdot, k} \otimes \mfrak{L}_{\cdot, k} = \sum_{k=1}^K (\mcal{M}_{\mb{R}}\mfrak{L})_{\cdot, k} \otimes (\mcal{M}_{\mb{R}}\mfrak{L})_{\cdot, k}.
\end{align*}
Thus, when the global component $Y$ and its covariance $\mscr{G}$ are known, it is impossible to distinguish the loading functions of $\mfrak{L}$ from those of $\mcal{M}_{\mb{R}}\mfrak{L}$.

\section{Theoretical Support}
\label{ts}

This section includes statements and proofs of results described in the manuscript, as well as those for auxiliary results that enrich our theory. 

\subsection{Identifiability in the Completely-Observed Paradigm}
\label{ts-id-co}

We begin with Theorem \ref{thm:id}, which establishes conditions needed to identify the global covariance $\mscr{G}$ and the local covariance $\mscr{B}$ from the covariance $\mscr{C}$ in the complete-observation paradigm. The result is a straightforward extension of Theorem 1 from \cite{descary-panaretos-2019} to multidimensional functional data.  

\begin{theorem}[Identifiability]
\label{thm:id}
    Let $\mscr{G}_1,\mscr{G}_2: L^2([0,1]^D) \to L^2([0,1]^D)$ be trace-class covariance operators of rank $K_1,K_2 < \infty$, respectively. Let $\mscr{B}_1,\mscr{B}_2: L^2([0,1]^D) \to L^2([0,1]^D)$ be banded trace-class covariance operators whose  bandwidths $\und{\delta}_1$ and  $\und{\delta}_2$, respectively, have components less than 1. If the eigenfunctions of $\mscr{G}_1$ and $\mscr{G}_2$ are real analytic, then we have the equivalence
    \begin{equation*}
        \mscr{G}_1 + \mscr{B}_1 = \mscr{G}_2 + \mscr{B}_2 \hspace{1em} \Longleftrightarrow \hspace{1em} \mscr{G}_1 = \mscr{G}_2  \hspace{1em} \text{and} \hspace{1em} \mscr{B}_1 = \mscr{B}_2.
    \end{equation*}
\end{theorem}

By definition, a function $f$ defined on open $U \subset \mbb{R}$ is real analytic at $x_0$ if $f$ is equal to its Taylor expansion about $x_0$ in some neighborhood of $x_0$. A function $f$ is real analytic on open $V \subset U$ if it is real analytic at each $x_0$ in $V$. In the proof of the above result, we use an alternative characterization of real analyticity that generalizes to multivariable functions. The power of analyticity in our setting comes from the so-called analytic continuation property: if a function is analytic on an open interval $U \subset \mbb{R}^D$, but is known only on open $V \subset U$, then the function extends uniquely to $U$. In the proof of Theorem \ref{thm:id}, after showing that the kernel $g$ is real analytic on $[0,1]^{2D}$, we use this observation to uniquely extend $g$ to the on-band subset of its domain (contaminated by $b$) based on off-band values (uncontaminated by $b$).

\begin{proof}[Proof of Theorem \ref{thm:id}]
    Let $g_1,g_2 \in L^2([0,1]^{2D})$ denote the kernels of $\mscr{G}_1$ and $\mscr{G}_2$, respectively. We begin by showing that $g_1$ (similarly, $g_2$) is analytic on $[0,1]^{2D}$. By Mercer's Theorem \citep{hsing-eubank-2015}, $g_1$ admits the decomposition
    \begin{align*}
        g_1(\und{s}_1, \und{s}_2) = \sum_{k=1}^K \lambda_k \eta_k(\und{s}_1)\eta_k(\und{s}_2),
    \end{align*}
    where we have assumed that each $\eta_k$ is analytic on $[0,1]$. Equivalently, each $\eta_k$ has some real analytic extension $\Tilde{\eta}_k$ to some open $U_k \supset [0,1]^D$. Let $U = \cap_k U_k$ and choose open $V$ and compact $R$ so that $U \supset R \supset V \supset [0,1]^D$. Since each $\Tilde{\eta}_k$ is real analytic on $U$ and $K \subset U$ is compact, there exists a positive constant $C$ such that for every multi-index $\alpha \in \mbb{N}^D$, 
    \begin{align*}
        \sup_{\und{s} \in R} \left| \frac{\partial^{\alpha} \eta_k}{\partial \und{s}^{\alpha}} (\und{s}) \right| \leq C^{\abs{\alpha}} \alpha!, \hspace{1em} k = 1, \dots, K.
    \end{align*}
    Define $\lambda^* = \max_k \lambda_k$, $\eta^* = \max_k \sup_{\und{s}\in R} \abs{\Tilde{\eta}_k (\und{s})}$, and $\Tilde{g}_1(\und{s}_1, \und{s}_2) = \sum_{k=1}^K \lambda_k \Tilde{\eta}_k (\und{s}_1)\Tilde{\eta}_k (\und{s}_2)$. Then for every multi-index $\alpha$,
    \begin{align*}
        \sup_{(\und{s}_1, \und{s}_2) \in V \times V} \abs{ \frac{\partial^{\alpha}}{\partial \und{s}_1^{\alpha}} \Tilde{g}_1(\und{s}_1, \und{s}_2) } & \leq \lambda^* \eta^* \sum_{k=1}^K \sup_{\und{s} \in V} \abs{ \frac{\partial^{\alpha}}{\partial \und{s}_1^{\alpha}} \Tilde{\eta}_k (\und{s}_1) } \\
        & \leq \lambda^* \eta^* \sum_{k=1}^K \sup_{\und{s} \in R} \abs{ \frac{\partial^{\alpha}}{\partial \und{s}_1^{\alpha}} \Tilde{\eta}_k (\und{s}_1) } \\
        & \leq \lambda^* \eta^* \sum_{k=1}^K C^{\abs{\alpha}} \alpha! \\
        & \leq C_0^{\abs{\alpha}} \alpha!,
    \end{align*}
    for some sufficiently large $C_0$. By symmetry, we also have 
    \begin{align*}
        \sup_{(\und{s}_1, \und{s}_2) \in V \times V} \abs{ \frac{\partial^{\alpha}}{\partial \und{s}_2^{\alpha}} \Tilde{g}_1(\und{s}_1, \und{s}_2) } \leq  C_0^{\abs{\alpha}} \alpha!.
    \end{align*}
    Thus, by Theorem 4.3.3 of \citep{krantz-park-2002}, $\Tilde{g}_1$ is real analytic on $V \times V$, making $g_1$ (similarly $g_2$) real analytic on $[0,1]^{2D}$.

    By Theorem 6.3.3 of \citep{krantz-park-2002}, the zero set of either kernel is at most $(2D - 1)$-dimensional, provided the kernels are not uniformly zero. Since the result follows trivially if $\mscr{G}_1$ and $\mscr{G}_2$ are the zero operator, we assume that their kernels are not uniformly zero. Thus, if we can show that the two kernels coincide on an open (i.e., $2D$-dimensional) subset $W$ of $[0,1]^{2D}$, then by Corollary 1.2.6 of \citep{krantz-park-2002}, they will coincide on $(0,1)^{2D}$, and on $[0,1]^{2D}$ by continuity. 

    Let $\delta < 1$ be the maximum component across $\und{\delta}_1$ and $\und{\delta}_2$, and define $W = (\delta, 1)^D \times (0, 1 - \delta)^D$. Since $\mscr{G}_1 + \mscr{B}_1 = \mscr{G}_1 + \mscr{B}_1$, but $\mscr{B}_1 = \mscr{B}_2 = 0$ on $W$, it must be that $\mscr{G}_1 = \mscr{G}_2$ on $W$. This completes the proof. 
\end{proof} 

Despite this seemingly restrictive property, the assumption of analyticity is not overly prohibitive. There are many examples of analytic functions, including polynomials, the exponential and logarithmic functions, and trigonometric functions. Moreover, the class of real analytic functions is closed under addition, subtraction, multiplication, division (given a nonvanishing denominator), differentiation, and integration. In fact, Theorem \ref{thm:dens} shows that any finite-rank covariance operator is well-approximated by finite-rank covariance operators having real analytic eigenfunctions. 

\begin{theorem}[Density]
\label{thm:dens}
    Let $Z$ be an $L^2([0,1]^D)$-valued random function with a trace-class covariance $\mscr{H}$ of rank $K < \infty$. Then, for any $\epsilon > 0$, there exists a random function $Y$ whose covariance $\mscr{G}$ has analytic eigenfunctions and rank $q \leq K$, such that
    \begin{equation*}
        \mbb{E} \norm{Z - Y}_{L^2}^2 < \epsilon \hspace{1em} \text{and} \hspace{1em} \norm{\mscr{H} - \mscr{G}}_{*} < \epsilon
    \end{equation*}
    for $\norm{\cdot}_{*}$ the nuclear norm. If additionally $\mscr{H}$ has $C^1$ eigenfunctions on $[0,1]^D$ then we have the stronger result that for any $\epsilon > 0$, there exists a random function $Y$ whose covariance $\mscr{G}$ has analytic eigenfunctions and rank $q \leq K$, such that
    \begin{equation*}
        \sup_{\und{s} \in [0,1]^D} \mbb{E} \abs{Z(\und{s}) - Y(\und{s})}^2 < \epsilon \hspace{1em} \text{and} \hspace{1em} \sup_{\und{s}_1,\und{s}_2 \in [0,1]^D} \abs{h(\und{s}_1, \und{s}_2) - g(\und{s}_1, \und{s}_2)} < \epsilon,
    \end{equation*}
    where $h$ and $g$ are the kernels of $\mscr{H}$ and $\mscr{G}$, respectively. 
\end{theorem}
\begin{proof}{}
    Proposition 1 of \cite{descary-panaretos-2019} establishes this result for $D = 1$. The proof uses Fourier series approximations of the eigenfunctions of $\mscr{H}$ to establish each result. To extend these results to $D>1$, we need only apply the same reasoning using a basis of $L^2([0,1]^D)$ that is composed of functions analytic on $[0,1]^D$, and that yields approximations dense in $C^1([0,1]^D)$ under the sup-norm. One can show that these criteria are satisfied by the Fourier tensor product basis,
    \begin{align*}
        \{ e_{n_1} \otimes \dots \otimes e_{n_D} : n_1, \dots, n_D \in \mbb{N}\},
    \end{align*}
    where $e_{n_d}$ is the $n_d$-th function of the Fourier basis of $L^2([0,1])$.
\end{proof}

\subsection{Identifiability in the Discretely-Observed Paradigm}
\label{ts-id-do}

Next, we turn to theory for the discrete-observation setting. We begin by describing an observation scheme that generalizes the evenly-spaced grid considered in the manuscript. Let $\{I_{\und{m},\und{M}}\}_{\und{m}\in \mbb{N}_{\und{M}}}$ be the partition of $[0,1]^D$ into cells of size $M_1^{-1} \times \dots \times M_D^{-1}$, where $\und{M}$ and $\mbb{N}_{\und{M}}$ are defined as in the manuscript. Suppose we observe each of the $n$ samples of $X$ at the $M = M_1 \dots M_D$ discrete points in $\{ \und{s}_{\und{m}} \}_{\und{m} \in \mbb{N}_{\und{M}}}$ where $\und{s}_{\und{m}} \in I_{\und{m},\und{M}}$, and denote the set of all possible discrete point collections by $\mcal{S}_{\und{M}}$. From here, we can follow the steps taken in the manuscript to define the \textit{\und{M}-resolution functional factor model} with $K$ factors as
\begin{align*}
    \mcal{X}^{\und{M}} = \sum_{k=1}^K \mcal{L}^{\und{M}}_k f_k + \mcal{E}^{\und{M}}.
\end{align*}
We denote the $\und{M}$-resolution versions of $(\mscr{C}, \mscr{G}, \mscr{B})$ by $(\mscr{C}^{\und{M}}, \mscr{G}^{\und{M}}, \mscr{B}^{\und{M}})$. These operators have kernels
\begin{align*}
    c^{\und{M}}(\und{x}, \und{y}) & = \sum_{\und{i},\und{j} \in \mbb{N}_{\und{M}}} c(\und{s}_{\und{i}}, \und{s}_{\und{j}}) \mathds{1}\{ (\und{x}, \und{y}) \in I_{\und{i},\und{M}} \times I_{\und{j},\und{M}}\}, \\
    g^{\und{M}}(\und{x}, \und{y}) & = \sum_{\und{i},\und{j} \in \mbb{N}_{\und{M}}} g(\und{s}_{\und{i}}, \und{s}_{\und{j}}) \mathds{1}\{ (\und{x}, \und{y}) \in I_{\und{i},\und{M}} \times I_{\und{j},\und{M}}\}, \\
    b^{\und{M}}(\und{x}, \und{y}) & = \sum_{\und{i},\und{j} \in \mbb{N}_{\und{M}}} b(\und{s}_{\und{i}}, \und{s}_{\und{j}}) \mathds{1}\{ (\und{x}, \und{y}) \in I_{\und{i},\und{M}} \times I_{\und{j},\und{M}}\},
\end{align*}
which can be summarized using the tensors of dimension $M_1 \times \dots \times M_D \times M_1 \times \dots \times M_D$,
\begin{align*}
    \mcal{C}^{\und{M}} (\und{i},\und{j})  = c^{\und{M}}(\und{s}_{\und{i}},  \und{s}_{\und{j}}), \hspace{1em}
    \mcal{G}^{\und{M}} (\und{i},\und{j}) = g^{\und{M}}(\und{s}_{\und{i}}, \und{s}_{\und{j}}), \hspace{1em}
    \mcal{B}^{\und{M}} (\und{i},\und{j}) = b^{\und{M}}(\und{s}_{\und{i}}, \und{s}_{\und{j}}).
\end{align*}
As in the manuscript, the covariance decomposes as
\begin{equation*}
    \mcal{C}^{\und{M}} = \underbrace{\sum_{k=1}^K \mcal{L}^{\und{M}}_k \otimes \mcal{L}^{\und{M}}_k}_{\mcal{G}^{\und{M}}} + \mcal{B}^{\und{M}}.
\end{equation*}

Theorem \ref{thm:id-disc} presents conditions needed to identify $\mcal{G}$ and $\mcal{B}$ from $\mcal{C}$. Chief among them is the requirement from the complete-observation setting that $\mscr{G}$ be finite-rank with real analytic eigenfunctions. Theorem \ref{thm:nonzero-minors} is an intermediate result that links the functional assumption of analyticity to discretely-observed covariances. 

\begin{theorem}
\label{thm:nonzero-minors}
    Let $\mscr{G}$ have kernel $g(\und{s}_1, \und{s}_2) = \sum_{k=1}^K \lambda_k \eta_k(\und{s}_1) \eta_k(\und{s}_2)$ with $K < \infty$ and real analytic eigenfunctions $\{ \eta_1, \dots, \eta_K\}$. If $M = M_1 \dots M_D > K$, then all order-$K$ minors of $\mcal{G}^{\und{M}}_{\text{sq}}$ are nonzero, almost everywhere on $\mcal{S}_{\und{M}}$.
\end{theorem}
\begin{proof}{}
    First, notice that from $g(\und{s}_1, \und{s}_2) = \sum_{k=1}^K \lambda_k \eta_k (\und{s}_1) \eta_k (\und{s}_2)$, we have that
    \begin{align*}
        (\mcal{G})_{(\und{i},\und{j})} = (\mcal{G}_\text{sq})_{(\gamma(\und{i}), \gamma(\und{j}))} = \sum_{k=1}^K \lambda_k \eta_k (\und{s}_{\und{i}}) \eta_k (\und{s}_{\und{j}}), \hspace{1em} \und{i}, \und{j} \in \mbb{N}_{\und{M}},
    \end{align*}
    where $\gamma: \mbb{N}_{\und{M}} \to \mbb{N}_{M}$ is the index set bijection that defines square matricization of $\mcal{G}$. Thus, the the rank-$K$ matrix $\mcal{G}_{\text{sq}} \in \mbb{R}^{M \times M}$ can be written as $\mb{U} \mb{\Sigma} \mb{U}^T$, where $(\mb{U})_{(\gamma (\und{i}), k)} = \eta_k (\und{s}_{\und{i}})$ and $\mb{\Sigma} = \text{diag}(\lambda_1, \dots, \lambda_K)$. Note that any $K \times K$ submatrix of $\mcal{G}_{\text{sq}}$ can be written as
    \begin{align*}
        \mb{U}_F \mb{\Sigma} \mb{U}_{F'}^T,
    \end{align*}
    where $\mb{U}_F$ and $\mb{U}_{F'}$ are $K \times K$ matrices obtained by deleting rows of $\mb{U}$ whose indices are not included in the cardinality-$K$ subsets $F,F' \subset \mbb{N}_{M}$, respectively. The condition that any order-$K$ minor of $\mcal{G}_{\text{sq}}$ be nonzero is equivalent to the condition that
    \begin{align*}
        \det\left[ \mb{U}_F \mb{\Sigma} \mb{U}_{F'}^T \right] = \det\left[ \mb{U}_F \right] \det\left[\mb{\Sigma}\right] \det\left[ \mb{U}_{F'}^T \right] \neq 0,
    \end{align*}
    for any subsets $F,F' \subset \mbb{N}_{M}$ of cardinality $K$. By construction, $\det[\mb{\Sigma}] \neq 0$, so the minor condition is then equivalent to requiring that $\det[\mb{U}_F] \neq 0$ for any subset $F \subset \mbb{N}_{M}$ of cardinality $K$. 

    We will show that this is indeed the case almost everywhere on $\mcal{S}_{\und{M}}$. Let $\mu$ denote the Lebesgue measure on $\mcal{S}_{\und{M}}$ and let $F = \mbb{N}_K$, without loss of generality. Using the Leibniz formula, we have that $\det [\mb{U}_F]$ can be written as the function
    \begin{align*}
        D_{\det} (\und{s}_{\gamma^{-1}(1)}, \dots, \und{s}_{\gamma^{-1}(K)}) & = \sum_{\sigma \in S_K} \epsilon(\sigma) \prod_{k=1}^K (\mb{U}_F)_{(\sigma(k),k)} \\
        & = \sum_{\sigma \in S_K} \epsilon(\sigma) \prod_{k=1}^K \eta_{k}(\und{s}_{\gamma^{-1}(\sigma(k))}),
    \end{align*}
    where $S_K$ is the symmetric group on $K$ elements and $\epsilon (\sigma)$ is the signature of the permutation $\sigma$. Note that the function $D_{\det}$ is real analytic on $[0,1]^{DK}$, by virtue of each $\eta_k$ being real analytic on $[0,1]^D$.

    We will now proceed by contradiction. Assume that
    \begin{align*}
        \mu \left\{ \{ \und{x}_{\und{j}} \}_{\und{j} \in \mbb{N}_{\und{M}}} \in \mcal{S}_{\und{M}} : D_{\det}(\mb{x}_{\gamma^{-1} (1)}, \dots , \mb{x}_{\gamma^{-1} (K)}) = 0 \right\} > 0.
    \end{align*}
    Since $\mu$ is the Lebesque measure, it follows that the Hausdorff dimension of the set $A = \{ (\und{x}_{1}, \dots,  \und{x}_{K}) : D_{\det} (\und{x}_{1}, \dots,  \und{x}_{K}) = 0 \}$ is equal to $DK$. However, since $D_{\det}$ is analytic, Theorem 6.3.3 of \citep{krantz-park-2002} implies the dichotomy: either $D_{\det}$ is constant everywhere on $(0,1)^{DK}$, or the set $A$ is at most of dimension $DK - 1$. Since the zero set of $D_{\det}$ has positive measure and is thus of Hausdorff dimension $DK$, it must be that $D_{\det}$ is everywhere zero on $(0,1)^{DK}$:
    \begin{align*}
        D_{\det} (\und{x}_{1}, \dots,  \und{x}_{K}) = \sum_{\sigma \in S_K} \epsilon(\sigma) \prod_{k=1}^K \eta_k (\und{x}_{\sigma(k)}) = 0, \hspace{1em} \forall (\und{x}_1, \dots \und{x}_K) \in (0,1)^{DK}.
    \end{align*}

    Now fix $(\und{x}_{1}, \dots,  \und{x}_{K-1})$ and apply $D_{\det}$ (viewed as a function of $\und{x}_K$ only) the continuous linear functional $T_{\eta_k}(f) = \ip{f}{\eta_k}$. Then for all $(\und{x}_{1}, \dots,  \und{x}_{K-1}) \in (0,1)^{D(K-1)}$:
    \begin{align*}
        0 & = \ip{D_{\det}}{\eta_k} \\
        & = \ip{\sum_{\sigma \in S_K} \epsilon(\sigma) \prod_{k=1}^K \eta_k (\und{x}_{\sigma(k)})}{\eta_k} \\
        & = \sum_{\sigma \in S_K} \left[ \epsilon(\sigma) \prod_{k:\sigma(k)\neq K} \eta_k(\und{x}_{\sigma(k)}) \right] \ip{\eta_{\sigma^{-1}(K)}}{\eta_K} \\
        & = \sum_{\sigma \in S_K} \epsilon(\sigma) \prod_{k=1}^{K-1} \eta_k(\und{x}_{\sigma(k)})
    \end{align*}
    Applying $T_{\eta_k}$ iteratively to $D_{\det}$ while keeping $(\und{x}_{1}, \dots,  \und{x}_{k-1})$ fixed then leads to
    \begin{align*}
        \eta_1(\und{s}) = 0, \hspace{1em} \forall \und{s} \in (0,1)^D.
    \end{align*}
    This last equality contradicts the fact that $\eta_1$ has norm one, and allows us to conclude that 
    \begin{align*}
        \mu \left\{ \{ \und{x}_{\und{j}} \}_{\und{j} \in \mbb{N}_{\und{M}}} \in \mcal{S}_{\und{M}} : D_{\det}(\mb{x}_{\gamma^{-1} (1)}, \dots , \mb{x}_{\gamma^{-1} (K)}) = 0 \right\} = 0.
    \end{align*}
\end{proof}

We now state Theorem \ref{thm:id-disc}, then invoke Theorem \ref{thm:nonzero-minors} in its proof.

\begin{theorem}[Identifiability at Finite Resolution]
\label{thm:id-disc}
    Let $\mscr{G}_1$ and $\mscr{G}_2$ be covariance operators of rank $K_1, K_2 < \infty$, respectively, and assume without loss of generality that $K_1 \geq K_2$. Let $\mscr{B}_1$ and $\mscr{B}_2$ be banded continuous covariance operators whose bandwidths $\und{\delta}_1$ and $\und{\delta}_2$, respectively, have components less than $1/2$. If the eigenfunctions of $\mscr{G}_1$ and $\mscr{G}_2$ are real analytic, and 
    \begin{align}
    \label{eqn:max-K-supp}
        K_1 \leq K^* =  \prod_{d=1}^D  \left\lfloor \left( \frac{1}{2} - \max\{ \delta_{1,d}, \delta_{2,d}\} \right) M_d - 1\right\rfloor,
    \end{align}
    then we have the equivalence
    \begin{align*}
        \mcal{G}^{\und{M}}_1 + \mcal{B}^{\und{M}}_1 = \mcal{G}^{\und{M}}_2 + \mcal{B}^{\und{M}}_2 \hspace{1em} \Longleftrightarrow \hspace{1em} \mcal{G}^{\und{M}}_1 = \mcal{G}^{\und{M}}_2 \hspace{1em} \text{and} \hspace{1em} \mcal{B}^{\und{M}}_1 = \mcal{B}^{\und{M}}_2,
    \end{align*}
    almost everywhere on $\mcal{S}_{\und{M}}$ with respect to the Lebesgue measure. 
\end{theorem}
\begin{proof}[Proof of Theorem 2 from the Manuscript]
    Given our conditions, the tensors $\mcal{B}_i \in \mbb{R}^{\mfrak{M} \times \mfrak{M}}$ have bandwidths $2\lceil \delta_{i,d} M_d \rceil$ for $i \in \{ 1, 2\}$ and $d = 1, \dots, D$. Let $\delta_d = \max\{ \delta_{1,d}, \delta_{2,d} \}$ for $d = 1, \dots, D$, and assume without loss of generality that $K_1 \leq K_2$. Let $\Omega$ be the set of indices on which both $\mcal{B}_1$ and $\mcal{B}_2$ vanish:
    \begin{align*}
        \Omega = \left\{ (\und{i}, \und{j}) \in \mbb{N}_{\und{M}} \times \mbb{N}_{\und{M}} : \abs{i_d - j_d} > \lceil \delta_d M_d \rceil \ \text{for some} \ d = 1, \dots, D \right\}.
    \end{align*}
    From $\mcal{G}_1 + \mcal{B}_1 = \mcal{G}_2 + \mcal{B}_2$, we obtain that $(\mcal{G}_1)_{(\und{i},\und{j})} = (\mcal{G}_2)_{(\und{i},\und{j})}$ for all $(\und{i}, \und{j}) \in \Omega$. Let $a_d = \left\lfloor \left( \frac{1}{2} - \max\{ \delta_{1,d}, \delta_{2,d}\} \right) M_d - 1\right\rfloor$ and define
    \begin{align*}
        \Omega_{\mcal{A}} = \{ (\und{i}, \und{j}) \in \mbb{N}_{\und{M}} \times \mbb{N}_{\und{M}} : 1 \leq i_d \leq a_d + 1, M_d - a_d \leq j_d \leq M_d \},
    \end{align*}
    then equation (1) of the manuscript implies that $\Omega_{\mcal{A}} \subset \Omega$, which in turn implies that the tensors $\mcal{G}_1$ and $\mcal{G}_2$ contain a common subtensor of dimension $a_1 \times \dots \times a_D \times a_1 \times \dots \times a_D$. Consequently, $\mcal{G}_{1\text{sq}}$ and $\mcal{G}_{2\text{sq}}$ contain a common submatrix of dimension $a_1 \dots a_D \times a_1 \dots a_D$, and, by equation (1) of the manuscript, a common submatrix $\mb{A}$ of dimension $K_1 \times K_1$.

    Assume that all order-$K_1$ minors of $\mcal{G}_{1\text{sq}}$ are nonzero. Then the determinant of $\mb{A}$ is nonzero. This implies that $K_2 = \text{rank} (\mcal{G}_{2\text{sq}}) \geq (\mcal{G}_{1\text{sq}}) = K_1$, which means $\mcal{G}_1$ and $\mcal{G}_2$ are two rank-$K_1$ tensors equal on $\Omega$. 

    Now, let $\mcal{G}^*$ be a tensor equal to $\mcal{G}_1$ on $\Omega$, but unknown at those indices not belonging to $\Omega$. We will now show that there exists a unique rank-$K_1$ completion of $\mcal{G}^*$. Due to equation (1) of the manuscript, it is possible to find a submatrix of $\mcal{G}_{\text{sq}}^*$ with dimension $(K_1 + 1) \times (K_1 + 1)$ containing exactly one unobserved value, denoted $x^*$. Using the fact that the determinant of any square submatrix of dimension larger than $K_1$ is zero, we obtain a linear equation of the form $ax^* + b = 0$ where $a$ is equal to to the determinant of a submatrix of $\mcal{G}_\text{sq}^*$ of dimension $K_1 \times K_1$. Since we have assumed that any minor of order $r_1$ is nonzero, we have that $a \neq 0$ and the previous equation has a unique solution. It is then possible to impute the value of $x^*$. Again, due to equation (1) of the manuscript, it is possible to apply this procedure iteratively until all missing entries are imputed, thereby allowing us to uniquely complete the tensor $\mcal{G}^*$ into a rank-$K_1$ tensor. 

    In summary, we have demonstrated that when all order-$K_1$ minors of $\mcal{G}_{1\text{sq}}$ are nonzero, it holds that $\mcal{G}^* = \mcal{G}_1 = \mcal{G}_2$, and hence $\mcal{B}_1 = \mcal{B}_2$. Theorem \ref{thm:nonzero-minors} ensures that $\mcal{G}_{1\text{sq}}$ indeed has nonvanishing minors of order $K_1$ almost everywhere on $\mcal{S}_{\und{M}}$. So we conclude that $\mcal{G}_1 = \mcal{G}_2$ and $\mcal{B}_1 = \mcal{B}_2$ almost everywhere on $\mcal{S}_{\und{M}}$. 
\end{proof}

\subsection{Estimation}
\label{ts-est}

We now present theory for the estimators of $\mcal{G}$, $K$, and $\mcal{L}_k$ discussed in Section 3 of the manuscript. To motivate the estimator of $\mcal{G}$, we define an optimization problem for which $\mcal{G}$ is the unique solution. Let $\Theta_{\und{M}}$ be the set of $\mfrak{M}\times\mfrak{M}$ symmetric non-negative definite tensors with trace norm bounded by that of $\hat{\mcal{C}}_n$, whose own trace norm may be scaled to one. Next, define the functional $S^{\und{M}}: \Theta_{\und{M}} \to \mbb{R}$ by
\begin{align*}
    S^{\und{M}} (\theta; \tau) & = \underbrace{M^{-2}\| \mcal{A}^{\und{M}} \circ (\mcal{C}^{\und{M}} - \theta) \|_F^2}_{R^{\und{M}}(\theta)} + \tau \text{rank}(\theta),
\end{align*}
where $M = M_1 \dots M_D$, $\mcal{A}^{\und{M}} \in \mbb{R}^{\mfrak{M}\times\mfrak{M}}$ is defined by $\mcal{A}^{\und{M}}(\und{i}, \und{j}) = \mathds{1} \{ \abs{i_d - j_d} > \lceil M_d / 4 \rceil$ for $d = 1, \dots, D\}$, and $\tau > 0$ is a rank-penalizing parameter. The subsequent lemma shows how the functional $S^{\und{M}}$ is used to formulate two versions of the same optimization problem, both of which are solved uniquely by $\mcal{G}$. 

\begin{lemma}
\label{lem:opt-probs}
    Let $\mscr{G}: L^2([0,1]^D) \to L^2([0,1]^D)$ be a rank $K < \infty$ covariance operator with analytic eigenfunctions and kernel $g$, and $\mscr{B}: L^2([0,1]^D) \to L^2([0,1]^D)$ be a trace-class covariance operator with $\und{\delta}$-banded kernel $b$. For $\{ \und{s}_{\und{m}} \}_{\und{m} \in \mbb{N}_{\und{M}}} \in \mcal{S}_{\und{M}}$, let 
    \begin{align*}
        \mcal{G}^{\und{M}} = \{ g(\und{s}_{\und{m}_1}, \und{s}_{\und{m}_2}) \}_{\und{m}_1, \und{m}_2} \hspace{1em} \text{and} \hspace{1em} \mcal{B}^{\und{M}} = \{ b(\und{s}_{\und{m}_1}, \und{s}_{\und{m}_2}) \}_{\und{m}_1, \und{m}_2},
    \end{align*}
    and $\mcal{C}^{\und{M}} = \mcal{G}^{\und{M}} + \mcal{B}^{\und{M}}$. Assume 
    \begin{align}
    \label{cond:est}
        \max_d \delta_d < 1/4 \hspace{1em} \text{and} \hspace{1em} K \leq K^* = \prod_{d=1}^D \left\lfloor \frac{M_d}{4} - 1 \right\rfloor.
    \end{align}
    Then for almost all grids in $\mcal{S}_{\und{M}}$:
    \begin{enumerate}
        \item The tensor $\mcal{G}^{\und{M}}$ is the unique solution to the optimization problem
        \begin{align*}
            \min_{\theta \in \mbb{R}^{\mfrak{M}\times\mfrak{M}}} \text{rank}(\theta) \hspace{1em} \text{subject to} \hspace{1em} R^{\und{M}}(\theta) = 0.
        \end{align*}
        \item Equivalently, in penalized form, 
        \begin{align*}
            \mcal{G}^{\und{M}} = \arg \min_{\theta \in \mbb{R}^{\mfrak{M}\times\mfrak{M}}} S^{\und{M}} (\theta; \tau),
        \end{align*}
        for sufficiently small $\tau > 0$. 
    \end{enumerate}
\end{lemma}
\begin{proof}
    We first show that $\mcal{G}^{\und{M}}$ is the unique solution to the constrained optimization problem. Note that a solution $\theta$ to the constrained problem must coincide with $\mcal{G}^{\und{M}}$ on 
    \begin{align*}
        \Omega = \{ (\und{i},\und{j}) : \abs{i_d - j_d} > \lceil M_d / 4 \rceil, d = 1, \dots, D \},
    \end{align*}
    since $\norm{\mcal{A}^{\und{M}} \circ (\mcal{C}^{\und{M}} - \theta)}_F = 0$. Then by condition (2) of the manuscript and Theorem \ref{thm:nonzero-minors}, the square matricization of $\theta$ has a non-zero minor of order $K$, implying that $\text{rank}(\theta) \geq K$. Since $\mcal{G}^{\und{M}}$ satisfies the constraint and has rank $K$, we have that $\text{rank}(\theta) = K$. Using the iterative procedure from the proof of Theorem 2 of the manuscript, we can show that $\theta$ equals $\mcal{G}^{\und{M}}$ everywhere. 

    Next, we show that $\mcal{G}^{\und{M}}$ is the unique minimizer of $S(\cdot;\tau)$ for sufficiently small $\tau > 0$. First, consider some $\theta \in \mbb{R}^{\mfrak{M} \times \mfrak{M}}$ with $\text{rank}(\theta) \geq K$. Since $\mcal{G}^{\und{M}}$ uniquely solves the constrained problem, we have that for all $\tau > 0$,
    \begin{align*}
        \norm{\mcal{A}^{\und{M}} \circ (\mcal{C}^{\und{M}} - \mcal{G}^{\und{M}})}_F^2 + \tau \text{rank}(\mcal{G}^{\und{M}}) <         \norm{\mcal{A}^{\und{M}} \circ (\mcal{C}^{\und{M}} - \theta)}_F^2 + \tau \text{rank}(\theta).
    \end{align*}
    Thus, $\theta$ does not minimize $S(\cdot;\tau)$. Next, consider some $\theta \in \mbb{R}^{\mfrak{M} \times \mfrak{M}}$ with $\text{rank}(\theta) \leq K - 1$. Let $\mu = \min_{\theta' \in \mbb{R}^{\mfrak{M} \times \mfrak{M}}, \text{rank}(\theta') \leq K - 1} \norm{\mcal{A}^{\und{M}} \circ (\mcal{C}^{\und{M}} - \theta')}_F^2 > 0$, where the inequality holds since a minimizer of $\norm{\mcal{A}^{\und{M}} \circ (\mcal{C}^{\und{M}} - \theta')}_F^2$ has rank bounded below by $K$. Let $\tau^* = \frac{\mu}{K - 1}$. Then for all $\tau < \tau^*$, 
    \begin{align*}
        \norm{\mcal{A}^{\und{M}} \circ (\mcal{C}^{\und{M}} - \mcal{G}^{\und{M}})}_F^2 + \tau \text{rank}(\mcal{G}^{\und{M}}) = 
        \tau K < \mu + \tau \leq \norm{\mcal{A}^{\und{M}} \circ (\mcal{C}^{\und{M}} - \theta)}_F^2 + \tau \text{rank}(\theta).
    \end{align*}
    Thus, $\theta$ does not minimize $S(\cdot ; \tau)$. 
\end{proof}

In Lemma \ref{lem:opt-probs}, we present constrained and penalized formulations of the same optimization problem. The constrained problem informs practical implementation of the estimator, while the penalized version motivates the formal estimator of $\mcal{G}$. To make this estimation problem tractable, we assume that $\mcal{G}$ belongs to the subset $\Theta_{\und{M}}^* \subset \Theta_{\und{M}}$ defined by
\begin{align*}
    \Theta_{\und{M}}^* = \{ \theta \in \Theta_{\und{M}} : \mscr{P} (\theta) < B, B < \infty \},
\end{align*}
where $\mscr{P}: \mbb{R}^{\mfrak{M} \times \mfrak{M}} \to \mbb{R}$ is some roughness-penalizing functional. On this space, define the functional $S_n^{\und{M}}: \Theta_{\und{M}}^* \to \mbb{R}$ by
\begin{align*}
    S^{\und{M}}_n (\theta; \tau, \alpha_n) = \underbrace{M^{-2}\| \mcal{A}^{\und{M}} \circ (\hat{\mcal{C}}_n^{\und{M}} - \theta) \|_F^2}_{R_n^{\und{M}}(\theta)} + \tau \text{rank}(\theta) + \alpha_n \mscr{P}(\theta),
\end{align*}
where $\alpha_n > 0$ is a roughness-penalizing parameter. This functional, which may be viewed as a roughness-penalizing empirical analogue of $S^{\und{M}}$, is used to define the estimator for $\mcal{G}$.  

\begin{definition}[Estimator of $\mcal{G}^{\und{M}}$]
\label{def:est-L}
    Suppose $(X_1, \dots, X_n)$ are an i.i.d. sample from a functional $K$-factor model. Let $\{ \und{s}_{\und{m}} \}_{\und{m} \in \mbb{N}_{\und{M}}} \in \mcal{S}_{\und{M}}$ and assume we observe
    \begin{align*}
        \mcal{X}_{i,\und{m}} = X_i(\und{s}_{\und{m}}), \hspace{1em} i = 1, \dots, n, \hspace{1em} \und{m} \in \mbb{N}_{\und{M}}.
    \end{align*}
    Let $\hat{\mcal{C}}_n^{\und{M}} \in \mbb{R}^{\mfrak{M} \times \mfrak{M}}$ be the empirical covariance tensor of the tensors 
    \begin{align*}
        \left\{ \{ \mcal{X}_{i,\und{m}} \}_{\und{m}\in\mbb{N}_{\und{M}}} \right\}_{i=1}^n.
    \end{align*}
    We define the estimator $\hat{\mcal{G}}^{\und{M}}_n \in \Theta_{\und{M}}^{*}$ of $\mcal{G}^{\und{M}}$ to be an approximate minimum of $S_n^{\und{M}}(\theta; \tau, \alpha_n)$, where $\tau > 0$ is sufficiently small, and $\alpha_n \to 0$ as $n \to \infty$. By approximate minimum, we mean that $\hat{\mcal{G}}^{\und{M}}_n$ satisfies 
    \begin{align*}
        S^{\und{M}}_n (\hat{\mcal{G}}^{\und{M}}_n) \leq \inf_{\theta} S^{\und{M}}_n(\theta) + o_P(1).
    \end{align*}
\end{definition}

The rank and scaled eigentensors of the estimator in Definition \ref{def:est-L} produce estimators for the number of factors $K$ and the loading tensors $\mcal{L}_k$, respectively. 

\begin{definition}[Estimators of $K$ and $\mcal{L}_k^{\und{M}}$]
\label{def:est-loadings}
    Consider the setting and estimator $\hat{\mcal{G}}^{\und{M}}_n$ presented in Definition \ref{def:est-L}. We define estimators for the number of factors $K$ and the $k$th loading function $\mcal{L}_k^{\und{M}}$ as
    \begin{align*}
        \hat{K}_n = \text{rank}(\hat{\mcal{G}}^{\und{M}}_n) \hspace{1em} \text{and} \hspace{1em} \hat{\mcal{L}}_{k,n}^{\und{M}} = \hat{\lambda}_{k,n}^{1/2} \hat{\mcal{H}}_{k,n}^{\und{M}},
    \end{align*}
    where $\hat{\lambda}_{k,n}$ and $\hat{\mcal{H}}_{k,n}^{\und{M}}$ are the $k$th eigenvalue and eigentensor of $\hat{\mcal{G}}^{\und{M}}_n$, respectively. 
\end{definition}

In Theorem \ref{thm:consistency}, we show that $\hat{\mcal{G}}$, $\hat{K}$, and $\hat{\mcal{L}}_k$, $k = 1, \dots, \hat{K}$ are consistent for their population analogues.

\begin{theorem}[Consistency]
\label{thm:consistency}
    Suppose $\mcal{L}_k$ is equal to the scaled eigentensor of $\mcal{G}$ for $k = 1, \dots, K$. Then for sufficiently small $\tau > 0$ and $\alpha_n \to 0$, we have
    \begin{enumerate}
        \item[(i)] $\norm{\hat{\mcal{G}}_n - \mcal{G}}_F^2 \overset{P}{\to} 0$,
        \item[(ii)] $\abs{\text{rank}(\hat{L}_n) - K} \overset{P}{\to} 0$, and
        \item[(iii)] $\norm{\hat{\mcal{L}}_{n,k} - \mcal{L}_k}_F^2 \overset{P}{\to} 0, \ k = 1, \dots, K$.
    \end{enumerate}
\end{theorem}
\begin{proof}{}
    In proving (i), we assume $D=1$ so that all tensors in the problem are $M \times M$ matrices. For $D>1$, simply square matricize each tensor then proceed as follows. By Corollary 3.2.3 from \cite{vaart-wellner-1996}, the result follows after verifying that:
    \begin{enumerate}
        \item[(a)] $\sup_{\theta \in \Theta_M^*} \abs{S_n^M(\theta) - S^M(\theta)} \overset{a.s.}{\to} 0$,
        \item[(b)] $S^M(\theta)$ is lower semi-continuous on $\Theta_M^*$,
        \item[(c)] $S^M(\theta)$ has a unique minimum at $\mcal{G}$, and
        \item[(d)] $\{ \hat{\mcal{G}}_n\}$ is uniformly tight. 
    \end{enumerate}
Note that we already established (c) in the Lemma from the manuscript. After scaling $\hat{\mcal{C}}_n^M$ to have unit trace norm, (d) holds as well. To verify (b), note that $R^M(\theta)$ is continuous and $\text{rank}(\theta)$ is lower semi-continuous. Being the sum of a continuous and lower semi-continuous function $S^M(\theta)$ is itself lower semi-continuous. 

To check (a), note that for $\theta \in \Theta_M^*$,
\begin{align*}
    | S_n^M(\theta; & \tau, \alpha_n) - S^M(\theta; \tau) | \\
    & = \abs{R_n^M(\theta) - R^M(\theta) + \alpha_n \mscr{P}(\theta)} \\
    & = \abs{ \norm{ \mcal{A}^M \circ (\hat{\mcal{C}}_n^M - \theta) }_F^2  - \norm{ \mcal{A}^M \circ (\mcal{C}^M - \theta) }_F^2} + \alpha_n \mscr{P}(\theta) \\
    & \leq \abs{ \norm{ \mcal{A}^M \circ (\hat{\mcal{C}}_n^M - \theta) }_F  - \norm{ \mcal{A}^M \circ (\mcal{C}^M - \theta) }_F } \\
    & \hspace{3em} \times \left( \norm{ \mcal{A}^M \circ (\hat{\mcal{C}}_n^M - \theta) }_F  + \norm{ \mcal{A}^M \circ (\mcal{C}^M - \theta) }_F \right) \\
    & \hspace{3em} + \alpha_n \mscr{P}(\theta) \\
    & \leq \norm{ \mcal{A}^M \circ (\hat{\mcal{C}}_n^M - \mcal{C}^M)}_F \left( 2 + \norm{\hat{\mcal{C}}_n^M}_F + \norm{\mcal{C}^M}_F\right) + \alpha_n B \\
    & \overset{a.s.}{\to} 0,
\end{align*}
where convergence follows from $\hat{\mcal{C}}_n^M \to \mcal{C}^M$ almost surely, and $\alpha_n \to 0$. This completes the proof of (i). 

To prove (ii), we again assume $D=1$ and extend the proof to $D>1$ via square matricization. Proceeding by contradiction, suppose that (ii) does not hold. Then for some subsequence $\{n'\} \subset \{n\}$, there exist $\epsilon > 0$ and $\delta > 0$ such that $\mbb{P} \left( \abs{\text{rank}(\hat{\mcal{G}}_{n'}) - K} > \epsilon \right) > \delta$ for all $n'$. Consequently, there exists either
\begin{enumerate}
    \item[(e)] a subset $\{j\} \subset \{ n' \}$ such that $\mbb{P}\left( \text{rank}(\hat{\mcal{G}}_{j}) < K \right) > \delta/2$ for all $j$, or 
    \item[(f)] a subset $\{i\} \subset \{ n' \}$ such that $\mbb{P}\left( \text{rank}(\hat{\mcal{G}}_{i}) > K \right) > \delta/2$ for all $i$.
\end{enumerate}
Focusing on (e), the convergence of $\hat{\mcal{G}}_n$ to $\mcal{G}$ in probability means there exists a subsequence $\{j'\} \subset \{j\}$ such that $\hat{\mcal{G}}_{j'} \to \mcal{G}$ almost surely. This means that on a set with probability at least $\delta / 2$, $\mcal{\hat{L}}_{j'} \to \mcal{G}$ and $\text{rank}(\hat{\mcal{G}}_{j'}) < r$ for all $j'$. This is not possible since the set of matrices with rank at most $r - 1$ is closed.  

We thus focus on (f). Again, since $\hat{\mcal{G}}_n$ converges to $\mcal{G}$ in probability, there exists a subsequence $\{i'\} \subset \{i\}$ for which $\hat{\mcal{G}}_{i'} \to \mcal{G}$ almost surely. This means that on a set with probability at least $\delta / 2$, $\hat{\mcal{G}}_{i'} \to \mcal{G}$ and $\text{rank}(\hat{\mcal{G}}_{i'}) > r$ for all $i'$. Working on this set, for all $i'$, 
\begin{align}
    R_{i'}^M(\hat{\mcal{G}}_{i'}) & + \tau (K + 1) + \alpha_{i'} \mscr{P}(\hat{\mcal{G}}_{i'}) \nonumber \\
    & \leq R_{i'}^M(\hat{\mcal{G}}_{i'}) + \tau \text{rank}(\hat{\mcal{G}}_{i'}) + \alpha_{i'} \mscr{P}(\hat{\mcal{G}}_{i'}) \nonumber \\
    & \leq \inf_{\theta \in \Theta_M^* : \text{rank}(\theta) = K} \left\{ R_{i'}^M(\theta) + \tau \text{rank}(\theta) + \alpha_{i'} \mscr{P}(\theta) \right\} \label{pf:cons-Li-minimizer} \\
    & \leq \inf_{\theta \in \Theta_M^* : \text{rank}(\theta) = K} \left\{ R_{i'}^M(\theta) + \tau K + \alpha_{i'} \mscr{P}(\theta) \right\} \nonumber \\
    & \leq \inf_{\theta \in \Theta_M^* : \text{rank}(\theta) = K} R_{i'}^M(\theta) + \tau K + \alpha_{i'} B, \label{pf:cons-B}
\end{align}
where line (\ref{pf:cons-Li-minimizer}) follows because $\hat{\mcal{G}}_{i'}$ is a minimizer of $S_{i'}^M$. From the proof of (i), we have that 
\begin{align}
    \sup_{\theta \in \Theta_M^*} \abs{R_{i'}^M(\theta) - \alpha_n \mscr{P}(\theta) - R^M(\theta)} \overset{a.s.}{\to} 0, \label{pf:cons-unif-conv}
\end{align}
which implies
\begin{align}
    R_{i'}^M(\hat{\mcal{G}}_{i'}) - 
    \alpha_{i'}\mscr{P}(\hat{\mcal{G}}_{i'}) - R^M(\hat{\mcal{G}}_{i'}) \overset{a.s.}{\to} 0. \label{pf:cons-unif-cons-imp}
\end{align}
Moreover, since $R^M$ is continuous and $\hat{\mcal{G}}_{i'} \to \mcal{G}$ almost surely, by continuous mapping we have
\begin{align}
    R^M(\hat{\mcal{G}}_{i'}) \overset{a.s.}{\to} R^M(\mcal{G}) = 0. \label{pf:cons-via-cont-map}
\end{align}
Combining the results of lines (\ref{pf:cons-unif-cons-imp}) and (\ref{pf:cons-via-cont-map}) gives 
\begin{align}
    R_{i'}^M(\hat{\mcal{G}}_{i'}) + \alpha_{i'}\mscr{P}(\hat{\mcal{G}}_{i'}) \overset{a.s.}{\to} 0. \label{pf:cons-lhs-conv}
\end{align}
Now, note that on the set $\{ \theta \in \Theta_M^* : \text{rank}(\theta) = K\}$ the $R_{i'}^M$ are equi-Lipschitz continuous. From this and the uniform convergence in line (\ref{pf:cons-unif-conv}), we have
\begin{align}
    \inf_{\theta \in \Theta_M^* : \text{rank}(\theta) = K} R_{i'}^M(\theta) \to \inf_{\theta \in \Theta_M^* : \text{rank}(\theta) = K} R^M(\theta) = 0. \label{pf:cons-rhs-conv}
\end{align}
By the inequality in line (\ref{pf:cons-B}) and the results in lines (\ref{pf:cons-lhs-conv}) and (\ref{pf:cons-rhs-conv}), we derive the contradiction that $\tau \leq 0$. This proves (ii).

Finally, (iii) follows immediately from (i), thus completing the proof. 
\end{proof}

\section{Tuning Procedures}
\label{tp}

Recall from Section 3.1.1 of the manuscript that we must choose values for the smoothing parameter $\alpha$ before estimating $\mcal{G}$. We find that subjective selection is often preferable, but there are some settings (e.g., a simulation study) that require automatic selection. One may do so using V-fold cross-validation, as follows:
\begin{enumerate}
    \item[(1a)] Partition the sample $\mcal{X}_1, \dots, \mcal{X}_n$ into $V$ folds of equal (or near equal) size.
    \item[(1b)] For a given $\alpha$, let $\Tilde{\mcal{G}}^{(v)}(\alpha)$ be the solution to the optimization problem
    \begin{align*}
        \min_{0 \preceq \theta \in \mbb{R}^{\mfrak{M}\times\mfrak{M}}} & M^{-2} \norm{ \mcal{A} \circ \left(\hat{\mcal{C}}_n^{(v)} - \theta \right) }_F^2 + M^{-1} \alpha \sum_{k=1}^{\text{rank}(\theta)} \Tilde{\mscr{P}} \mscr{E}_k (\theta) \\
        & \text{subject to} \hspace{1em} \text{rank}(\theta) \leq K^*,
    \end{align*}
    where $\hat{\mcal{C}}_n^{(v)}$ denotes the empirical covariance of all samples except those in the $v$th fold, and $K^*$ is the maximal rank of $\mcal{G}$. That is, $\Tilde{\mcal{G}}^{(v)}(\alpha)$ is the maximal-rank version of the estimator in Definition 1 of the manuscript for the sample excepting the $v$th fold.
    \item[(1c)] Compute the cross-validation scores
    \begin{align*}
        \text{CV}(\alpha) = \frac{1}{VM^2} \sum_{v=1}^V \norm{ \mcal{A} \circ \left( \hat{\mcal{C}}_n^v -  \Tilde{\mcal{G}}^{(v)}(\alpha)\right) }_F^2,
    \end{align*}
    where $\hat{\mcal{C}}_n^v$ is the empirical covariance of samples in the $v$th fold. 
    \item[(1d)] Minimize $\text{CV}(\alpha)$ to provide the choice of smoothing parameters.
\end{enumerate}
One might wonder if we could use the above procedure to tune a \textit{vector} of smoothing parameters $\und{\alpha}$ that applies a different level of smoothing to each eigentensor. We do not recommend this for the following reason. Recall that we solve the problem in (1b) by re-framing it as a factorized matrix completion problem like equation (4) of the manuscript that optimizes over $\mb{V} \in \mbb{R}^{M\times K^*}$, $M = M_1...M_D$. For fixed $\mb{V}\mb{V}^T$, the error term in equation (4) of the manuscript is constant for any $\mb{V}$ while the penalty term varies as we rotate $\mb{V}$. This means that rough loading tensors can avoid penalization by ``hiding'' from larger $\alpha_k$ in some rotation of $\mb{V}$, making it difficult to target a specific loading tensor with a particular roughness penalty.

Next, recall from Section 3.2 of the manuscript that some tuning procedure is needed to systematically choose the shrinkage parameter $\und{\kappa}$. One option is to choose $\und{\kappa}$ using $V$-fold cross-validation, as follows:
\begin{enumerate}
    \item[(2a)] Apply a varimax rotation to the initial estimate $\hat{\mfrak{L}}$ to obtain $\hat{\mfrak{L}}^*$.
    \item[(2b)] Partition the sample $\mcal{X}_1, \dots, \mcal{X}_n$ into $V$ folds of equal (or near equal) size. 
    \item[(2c)] Using the tuned smoothing parameter $\alpha^*$, compute $\hat{\mfrak{L}}^{(v)}(\alpha^*)$ as in step (b) of the estimation procedure for $\mcal{G}$ (see Section 3.1.1 of the manuscript), setting $K^* = \hat{K}$. We suppress $\alpha^*$ from here forward. 
    \item[(2d)] Rotate each $\hat{\mfrak{L}}^{(v)}$ towards the target $\hat{\mfrak{L}}$ (using \texttt{GPArotation::targetT} from \cite{bernaards-etal-2015}) to get $(\hat{\mfrak{L}}^{(v)})^*$. 
    \item[(2e)] For a given $\und{\kappa}$, adaptively soft-threshold $(\hat{\mfrak{L}}^{(v)})^*$ as in equation (5) of the manuscript to get $\Tilde{\mfrak{L}}^{(v)}(\und{\kappa})$. 
    \item[(2f)] Compute the cross-validation scores
    \begin{align*}
        \text{CV}(\und{\kappa}) = \frac{1}{VM^2} \sum_{v=1}^V \norm{ \mcal{A} \circ \left( \hat{\mcal{C}}^{v} - \Tilde{\mcal{G}}^{(v)}(\und{\kappa})\right) }_F^2,
    \end{align*}
    where $\Tilde{\mcal{G}}^{(v)}(\und{\kappa}) = \sum_{k=1}^K \Tilde{\mfrak{L}}_{\cdot, k}^{(v)}(\und{\kappa}) \otimes \Tilde{\mfrak{L}}_{\cdot, k}^{(v)}(\und{\kappa})$, and $\hat{\mcal{C}}^v$ is the empirical covariance of samples in the $v$th fold. 
    \item[(2g)] Minimize $\text{CV}(\und{\kappa})$ to provide the choice of shrinkage parameter. 
\end{enumerate}

\section{Rank Simulation Study}
\label{sec:ss}

To study the behavior of the scree plot approach, we simulated a single instance of each configuration, then plotted the quantities $f(\hat{\theta}_j)$ for $j = 1, \dots, 8$. Figures \ref{fig:scree-plots-1}, \ref{fig:scree-plots-2}, \ref{fig:scree-plots-3}, and \ref{fig:scree-plots-4} display the resulting scree plots. As expected, R2 presents a harder problem than R1, while rank selection generally becomes easier as $n$ grows. Post-hoc analysis of some selection failures -- for instance, when the scree plot appears non-decreasing or under-selects the rank -- reveals that over-smoothing can force optimization into bad local optima, and that such optima may sometimes be avoided with a more conservative choice of the smoothing parameter. In particular, for several R2 configurations, the scree plot approach fails more frequently for the NET loading scheme than for the BL scheme. As discussed in Section 4.1 of the manuscript, this is likely due to the application of a single smoothing parameter to the loading tensors of NET which exhibit varying degrees of smoothness. Despite these occasional failures, our simulation study supports the scree plot approach as a viable rank selection method.

\begin{figure}[!h]
\centering
\subfloat[Scenario S1]{\label{fig:scree-plots-1-s1}\includegraphics[width=0.35\linewidth]{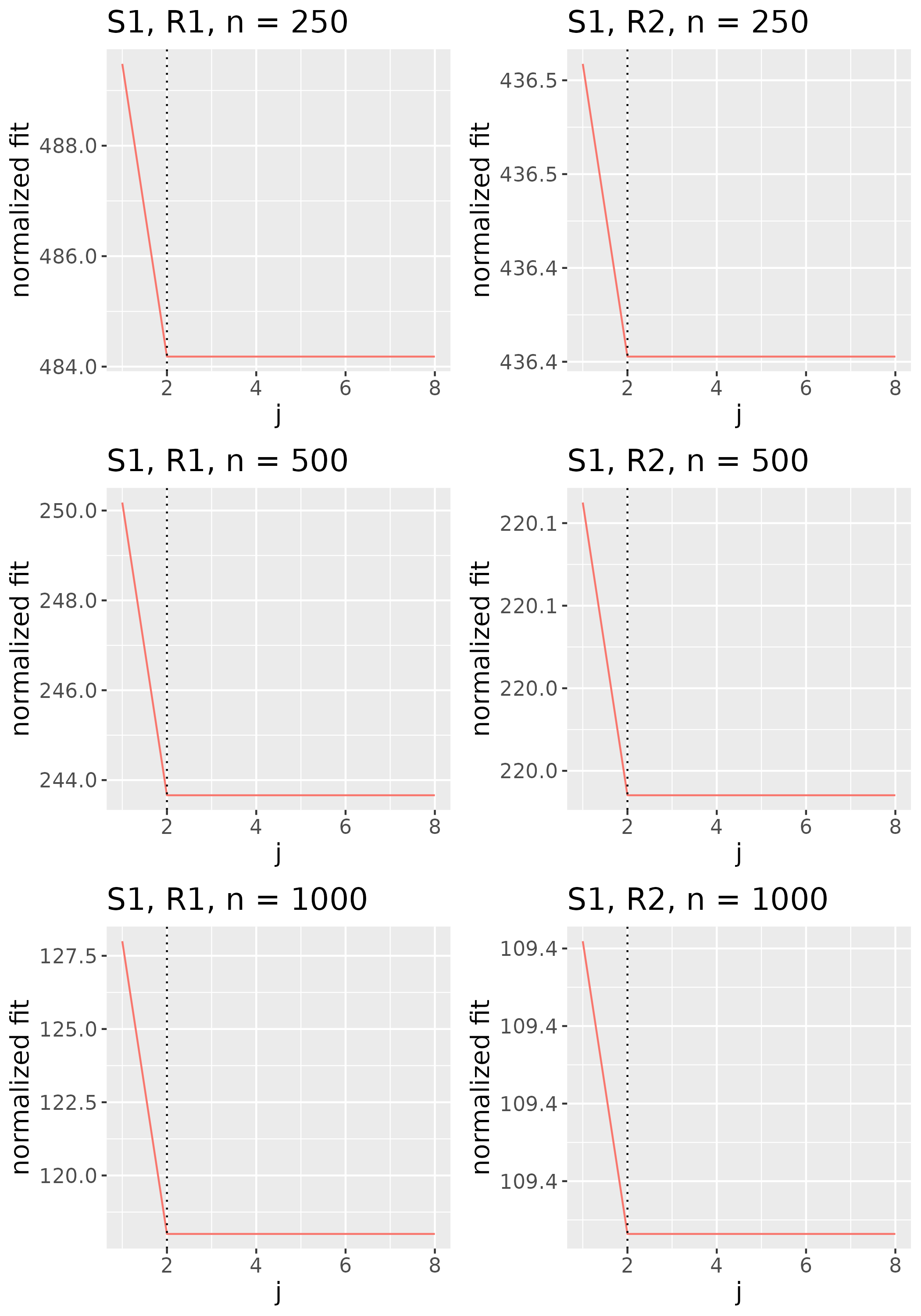}}\qquad
\subfloat[Scenario S2]{\label{scree-plots-1-s2}\includegraphics[width=0.35\linewidth]{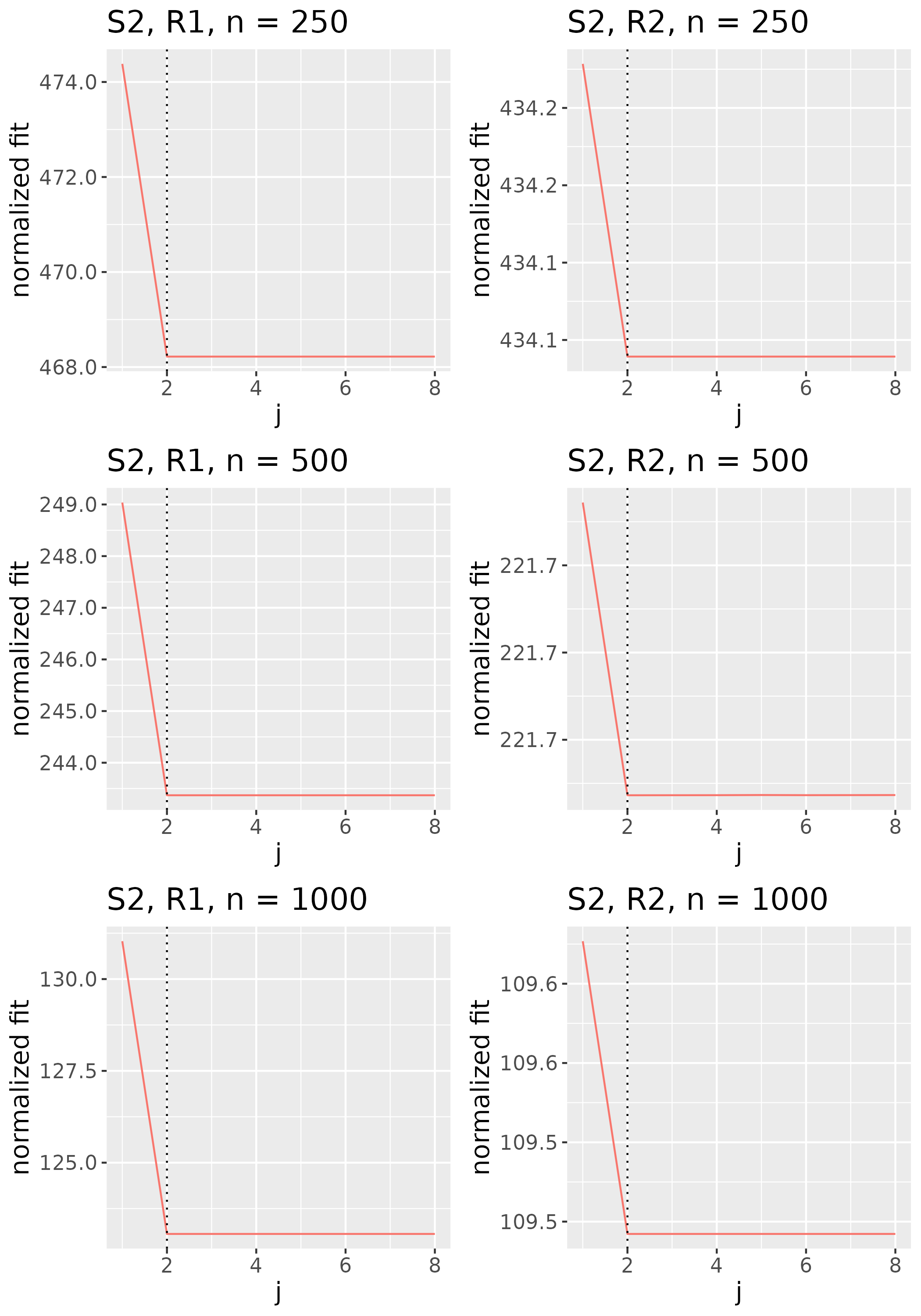}}\\
\subfloat[Scenario S3]{\label{scree-plots-1-s3}\includegraphics[width=0.35\linewidth]{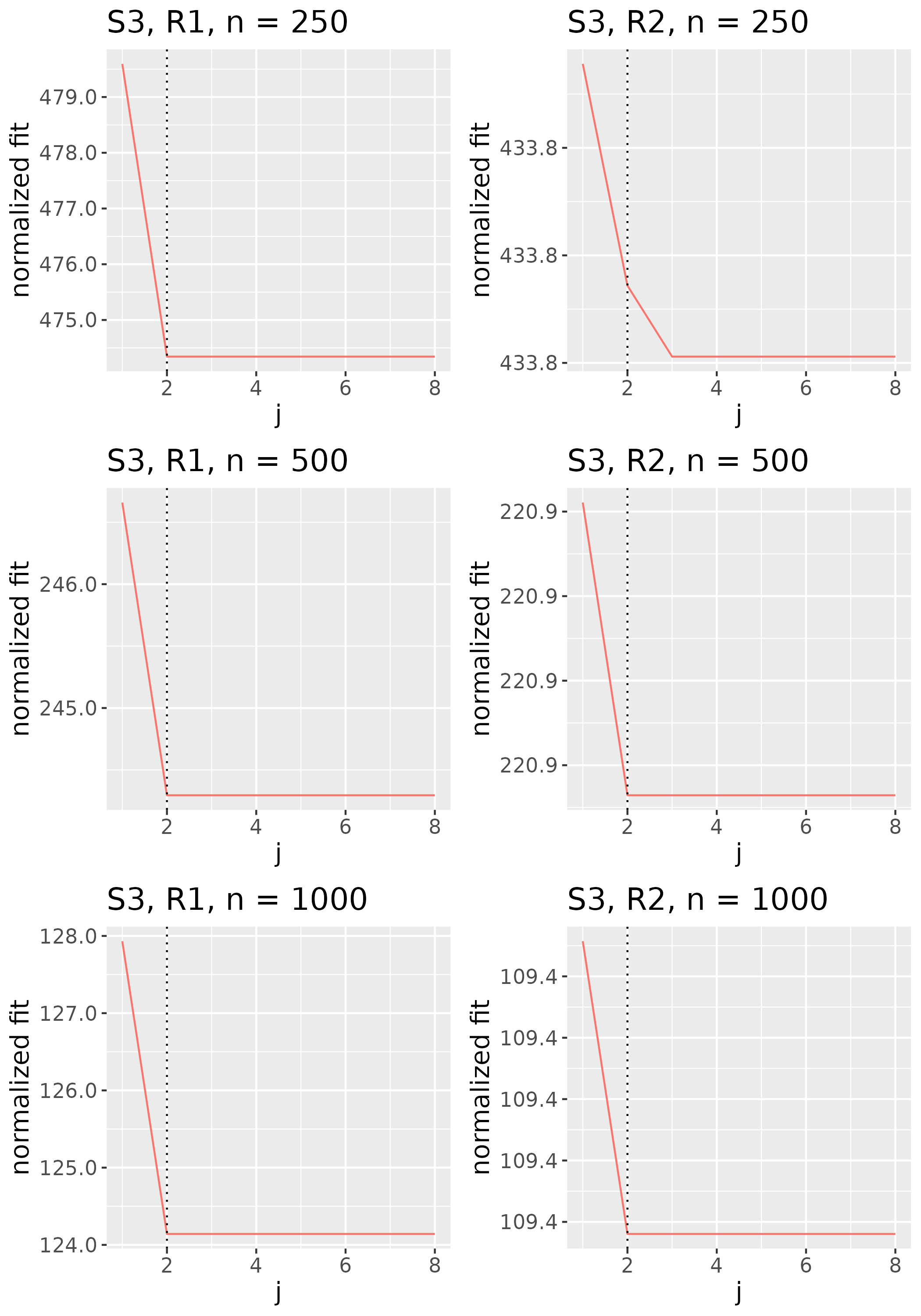}}\qquad
\subfloat[Scenario S4]{\label{scree-plots-1-s1}\includegraphics[width=0.35\linewidth]{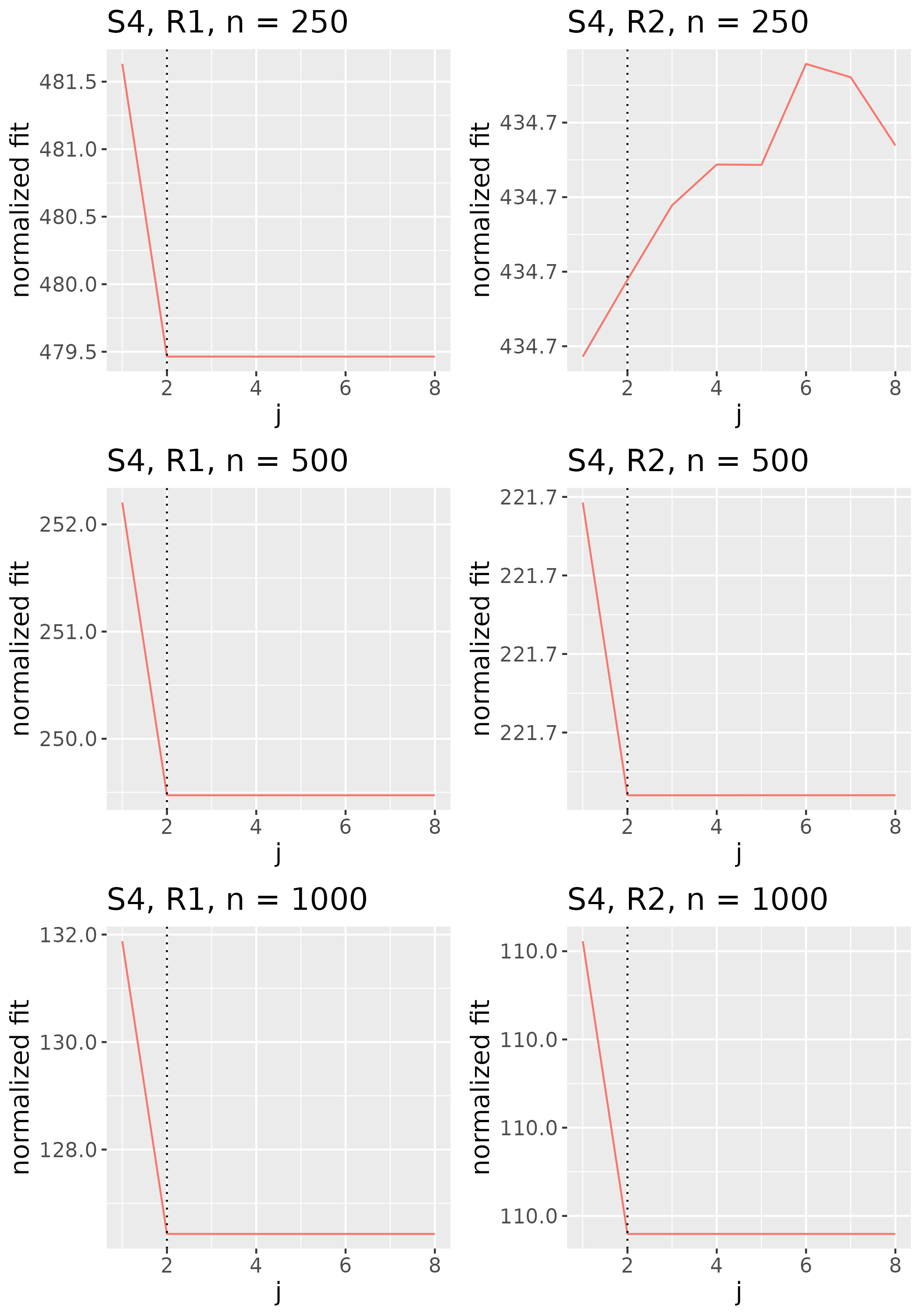}}
\caption{Plots of the function $j \mapsto f(\hat{\theta}_j)$ for $K = 2$ and $\delta = 0.05$. Each plot is for some combination of scenario, regime, and sample size $n$.}
\label{fig:scree-plots-1}
\end{figure}

\begin{figure}[!h]
\centering
\subfloat[Scenario S1]{\label{fig:scree-plots-2-s1}\includegraphics[width=0.35\linewidth]{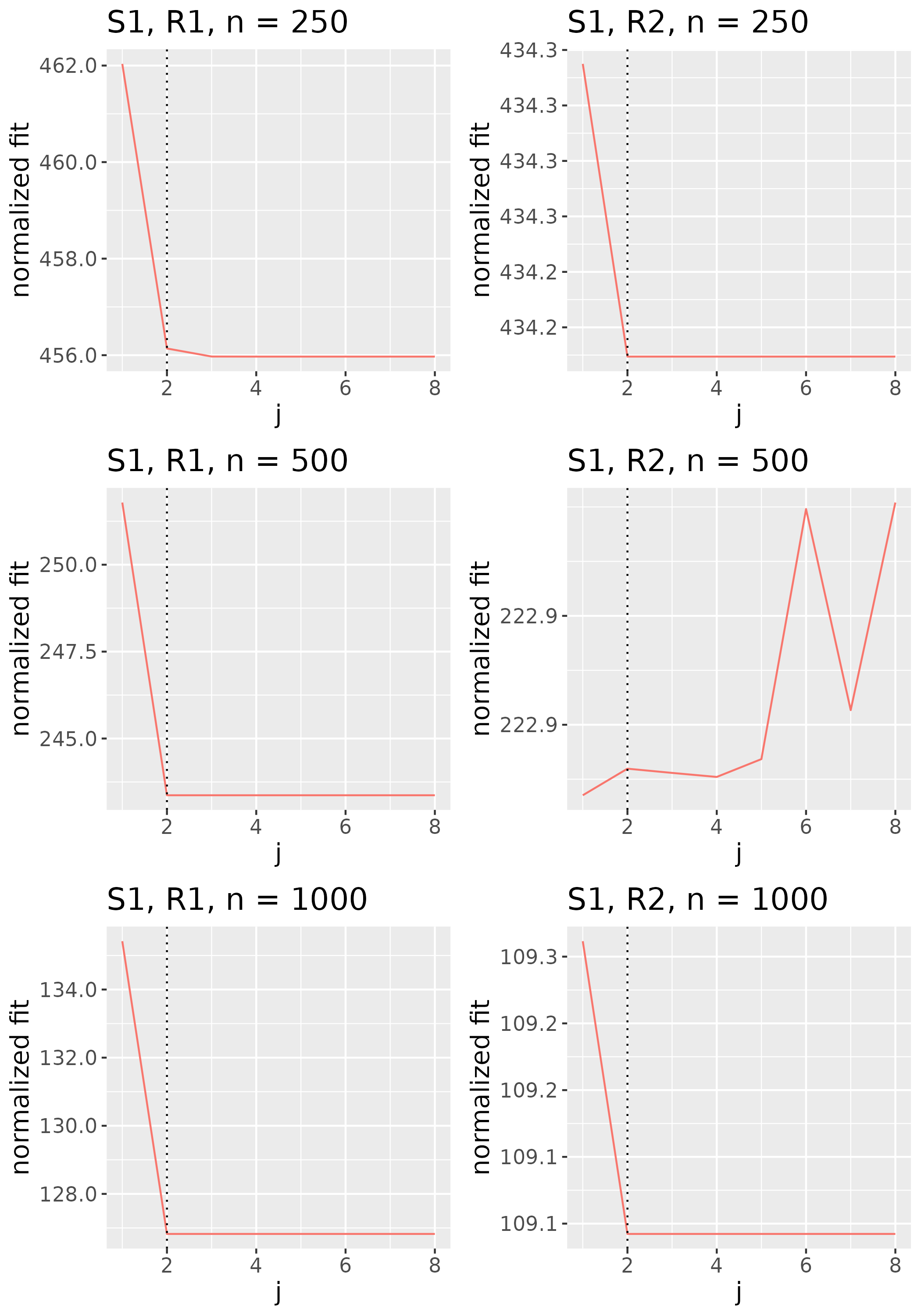}}\qquad
\subfloat[Scenario S2]{\label{scree-plots-2-s2}\includegraphics[width=0.35\linewidth]{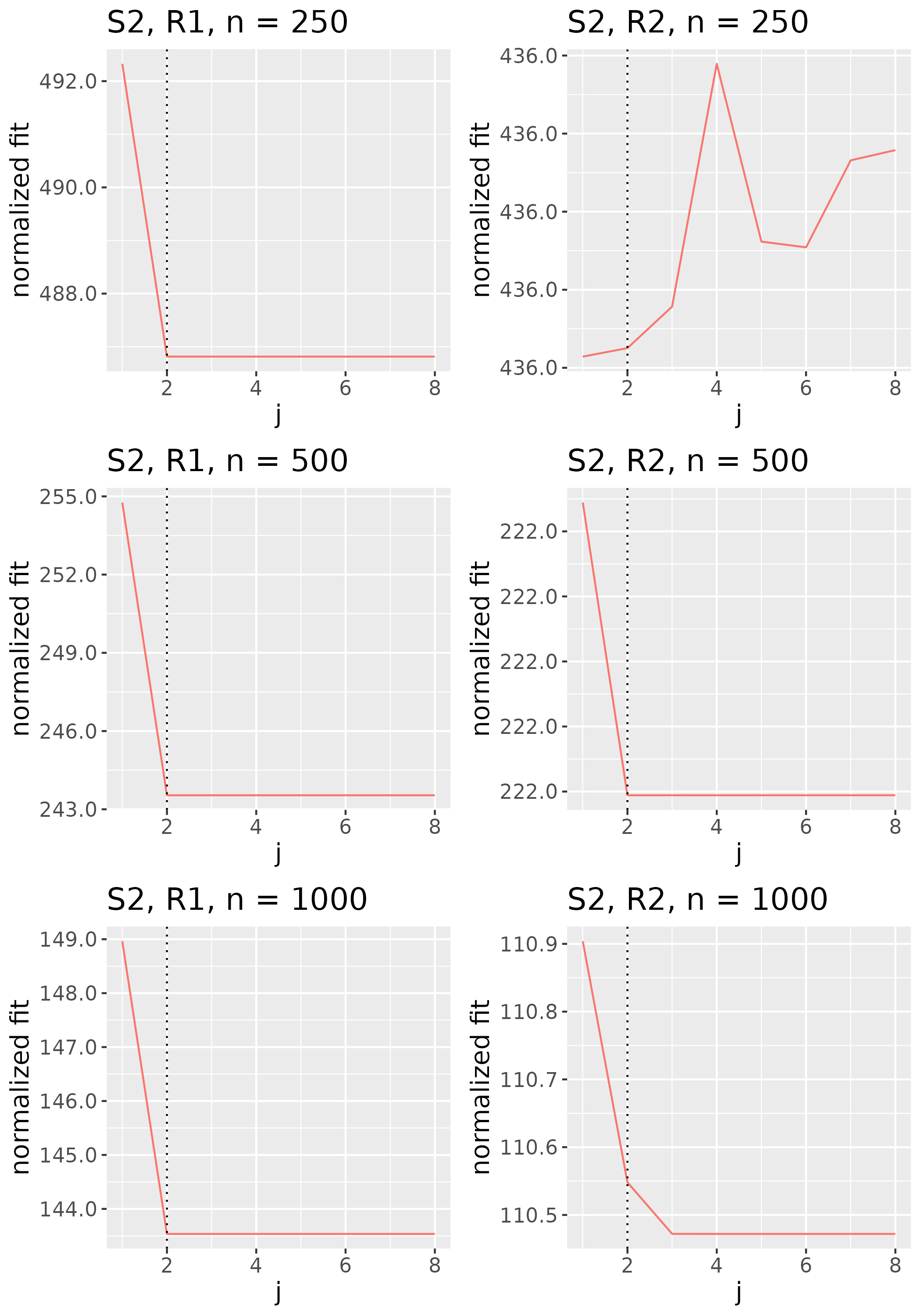}}\\
\subfloat[Scenario S3]{\label{scree-plots-2-s3}\includegraphics[width=0.35\linewidth]{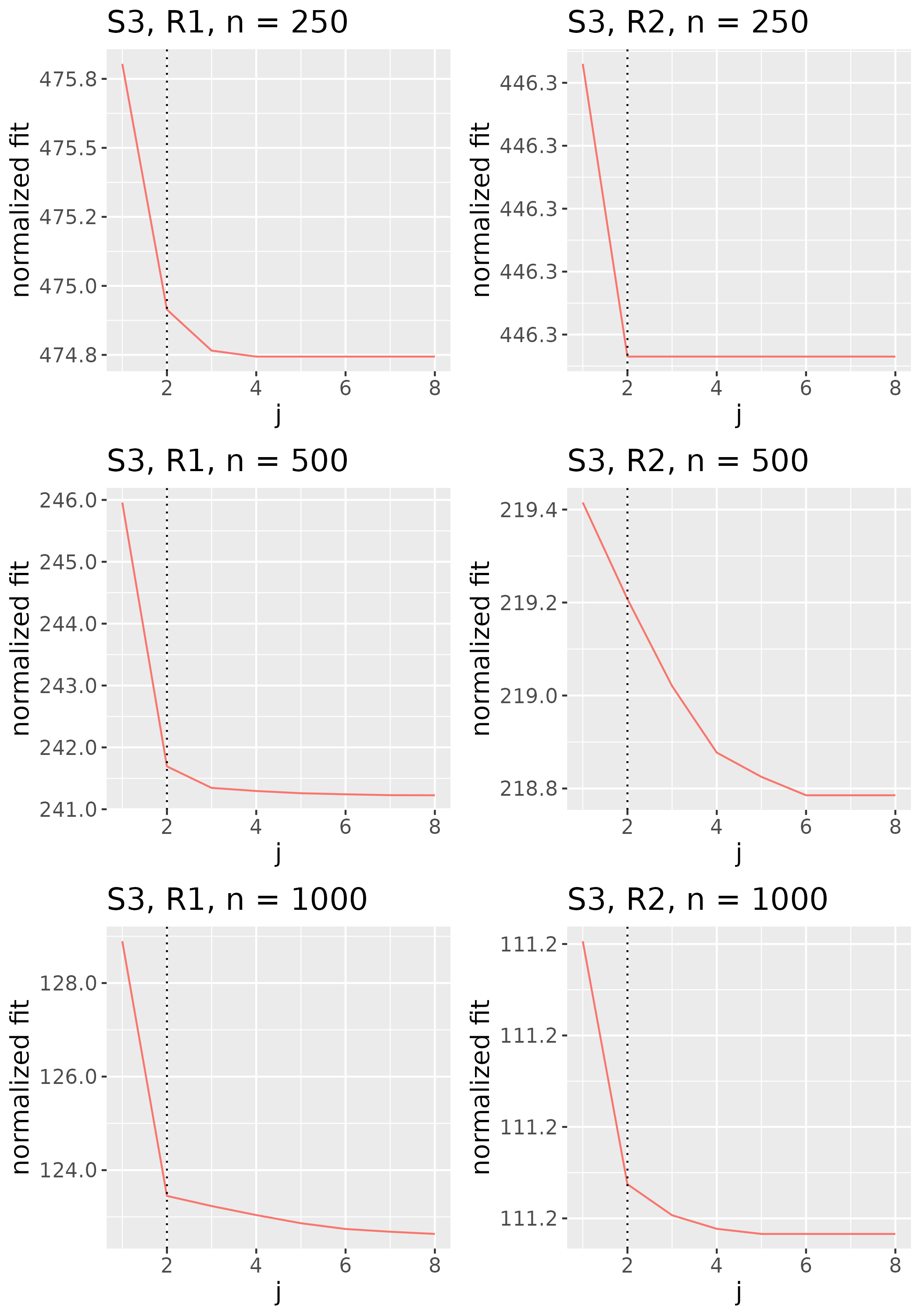}}\qquad%
\subfloat[Scenario S4]{\label{scree-plots-2-s1}\includegraphics[width=0.35\linewidth]{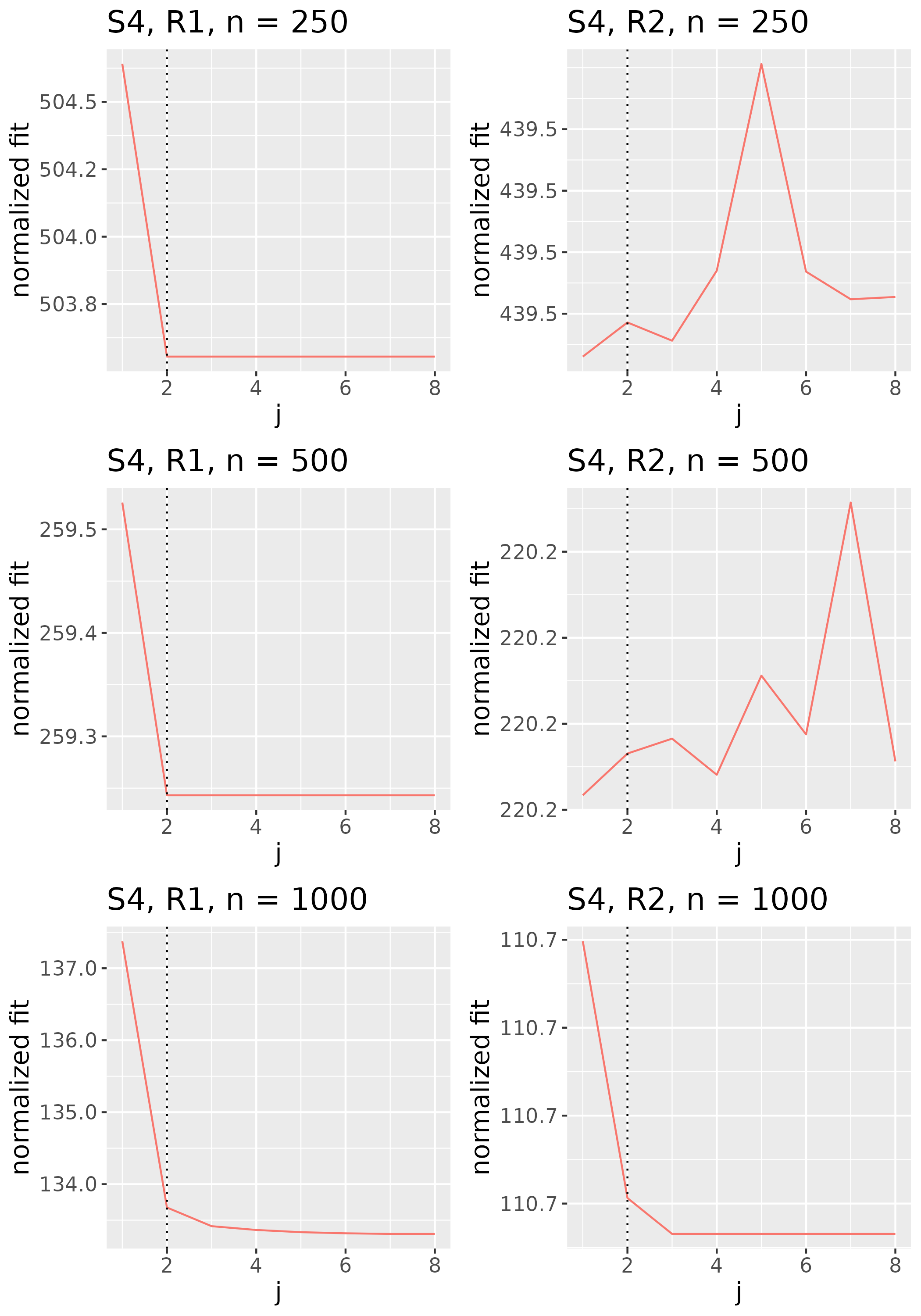}}%
\caption{Plots of the function $j \mapsto f(\hat{\theta}_j)$ for $K = 2$ and $\delta = 0.1$. Each plot is for some combination of scenario, regime, and sample size $n$.}
\label{fig:scree-plots-2}
\end{figure}

\begin{figure}[!h]
\centering
\subfloat[Scenario S1]{\label{fig:scree-plots-3-s1}\includegraphics[width=0.35\linewidth]{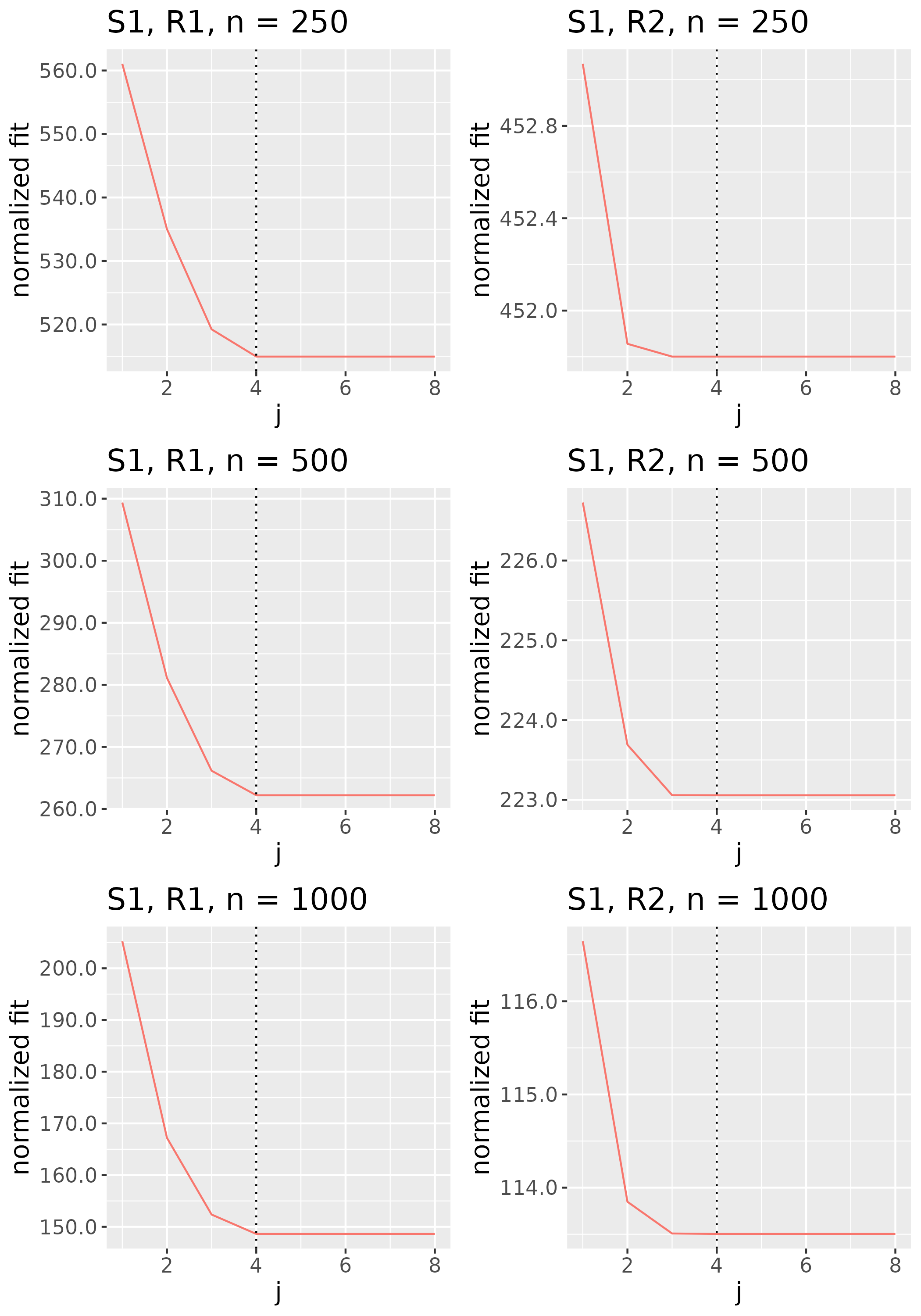}}\qquad
\subfloat[Scenario S2]{\label{scree-plots-3-s2}\includegraphics[width=0.35\linewidth]{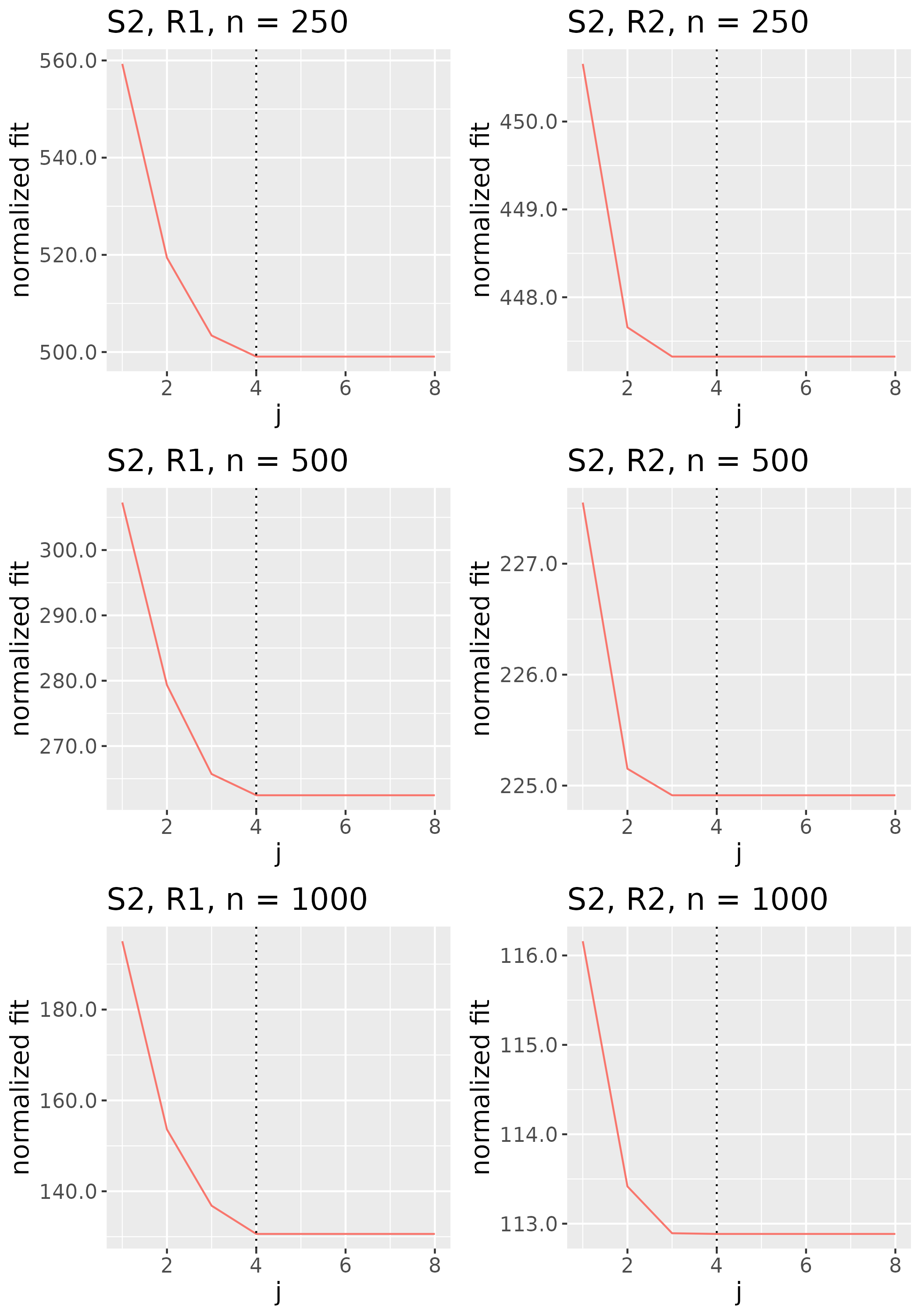}}\\
\subfloat[Scenario S3]{\label{scree-plots-3-s3}\includegraphics[width=0.35\linewidth]{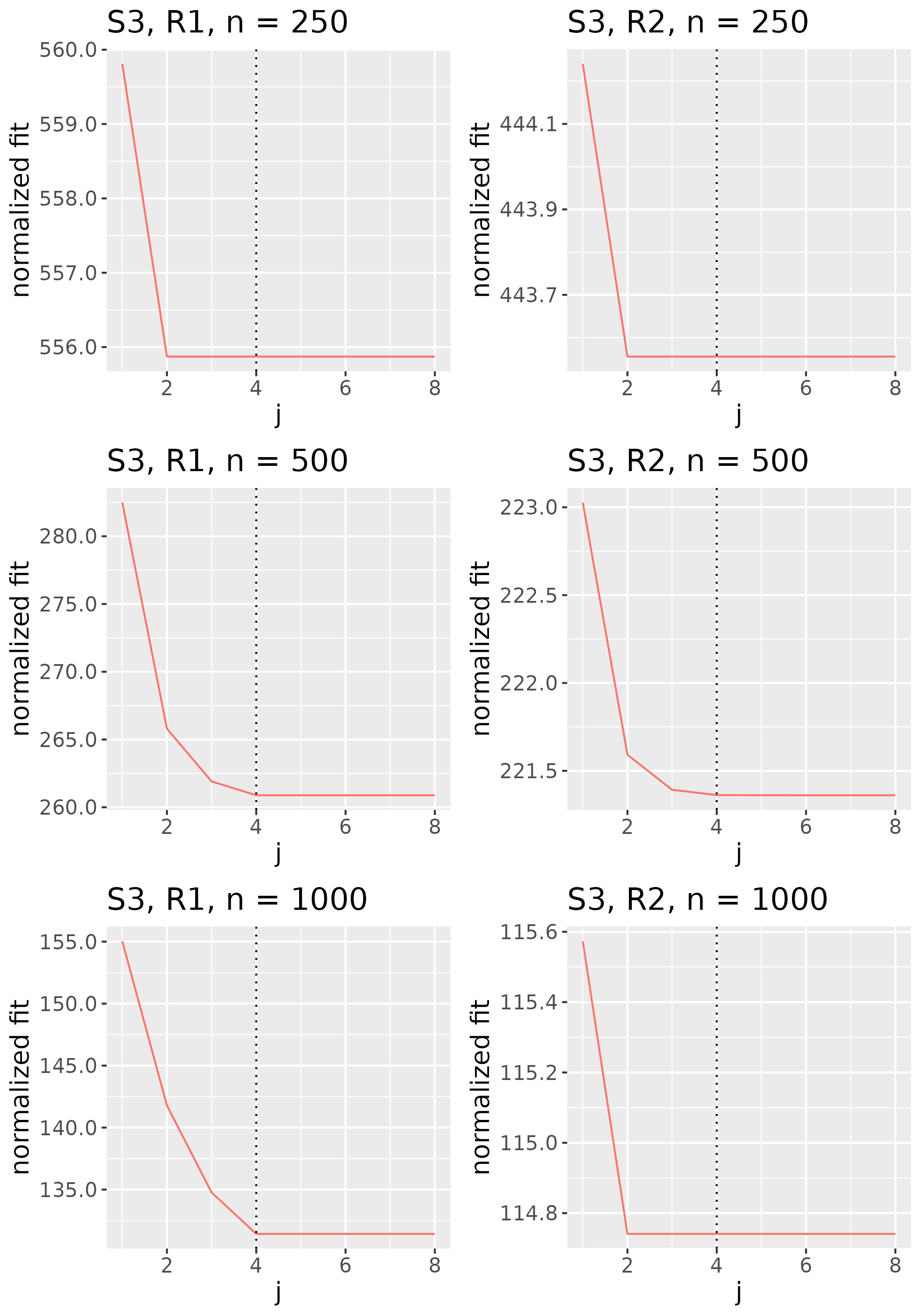}}\qquad%
\subfloat[Scenario S4]{\label{scree-plots-3-s1}\includegraphics[width=0.35\linewidth]{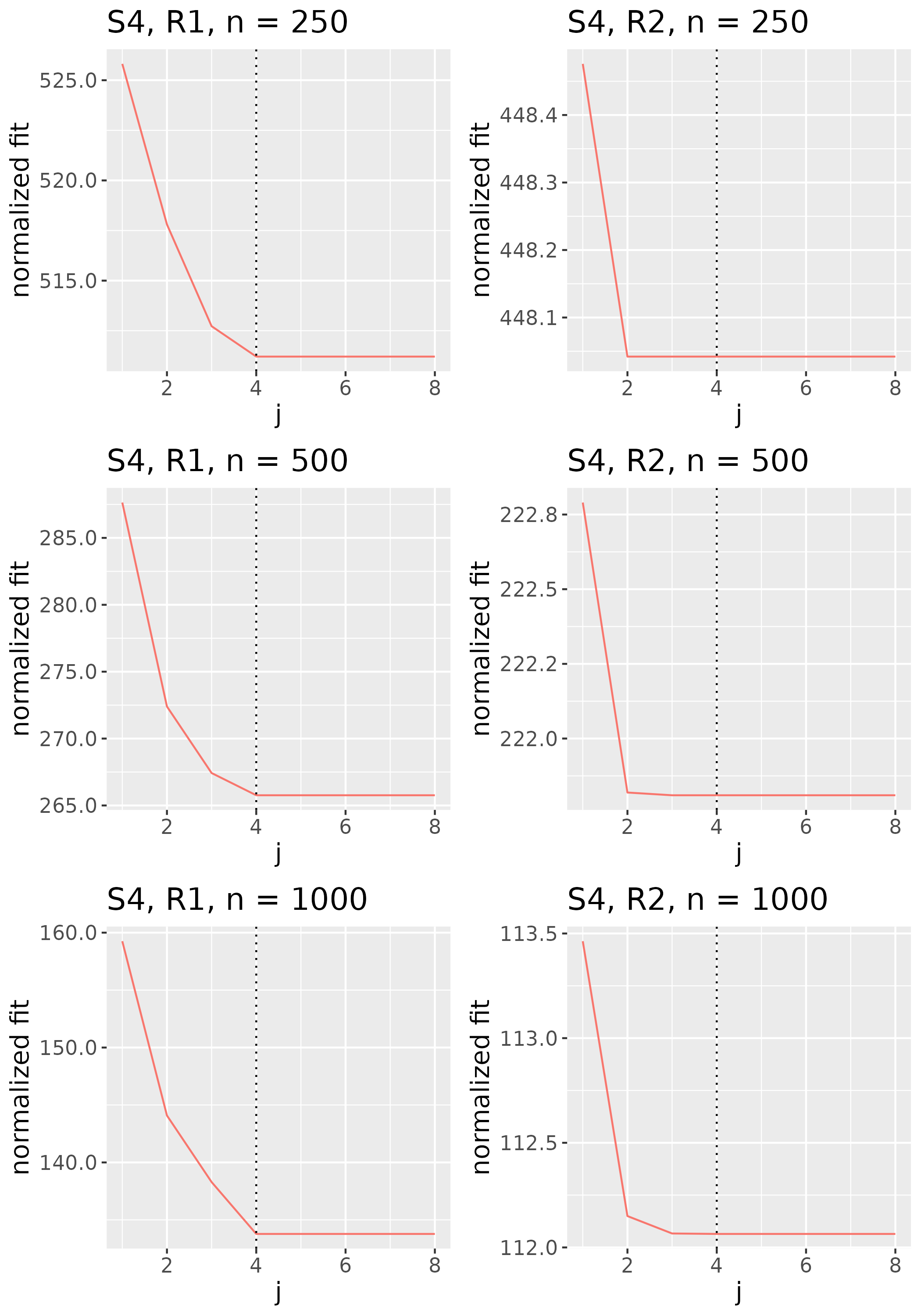}}%
\caption{Plots of the function $j \mapsto f(\hat{\theta}_j)$ for $K = 4$ and $\delta = 0.05$. Each plot is for some combination of scenario, regime, and sample size $n$.}
\label{fig:scree-plots-3}
\end{figure}

\begin{figure}[!h]
\centering
\subfloat[Scenario S1]{\label{fig:scree-plots-4-s1}\includegraphics[width=0.35\linewidth]{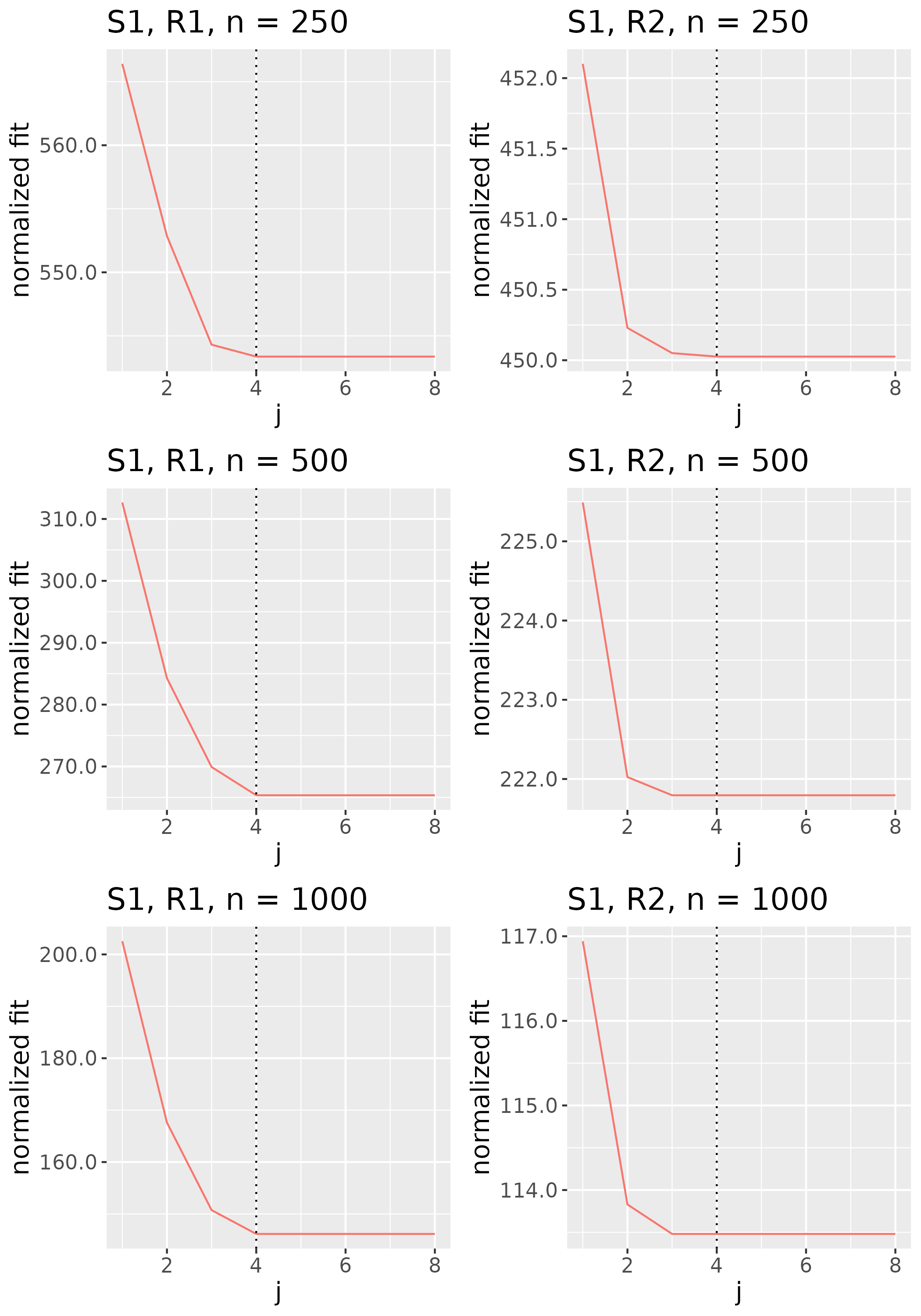}}\qquad
\subfloat[Scenario S2]{\label{scree-plots-4-s2}\includegraphics[width=0.35\linewidth]{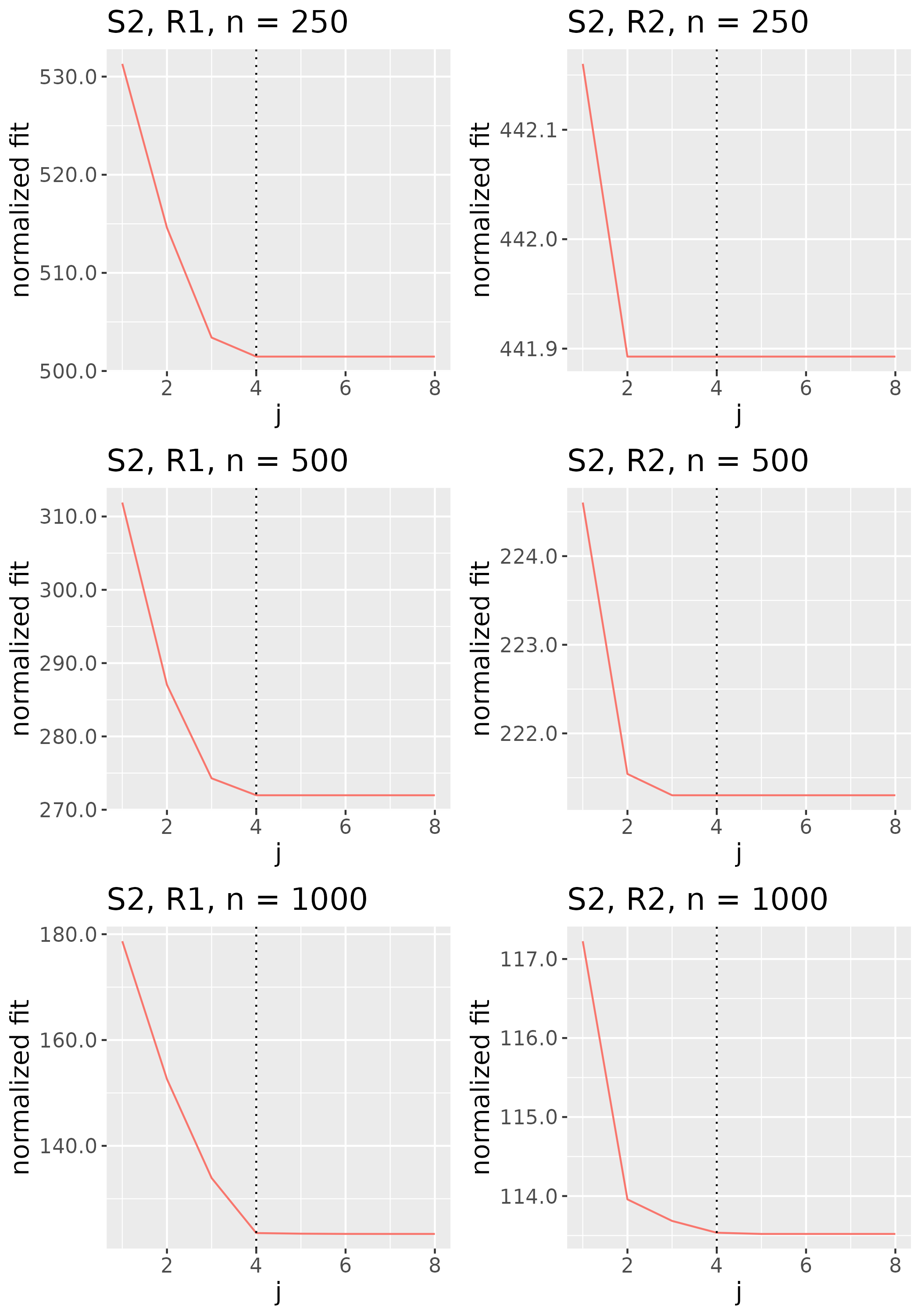}}\\
\subfloat[Scenario S3]{\label{scree-plots-4-s3}\includegraphics[width=0.35\linewidth]{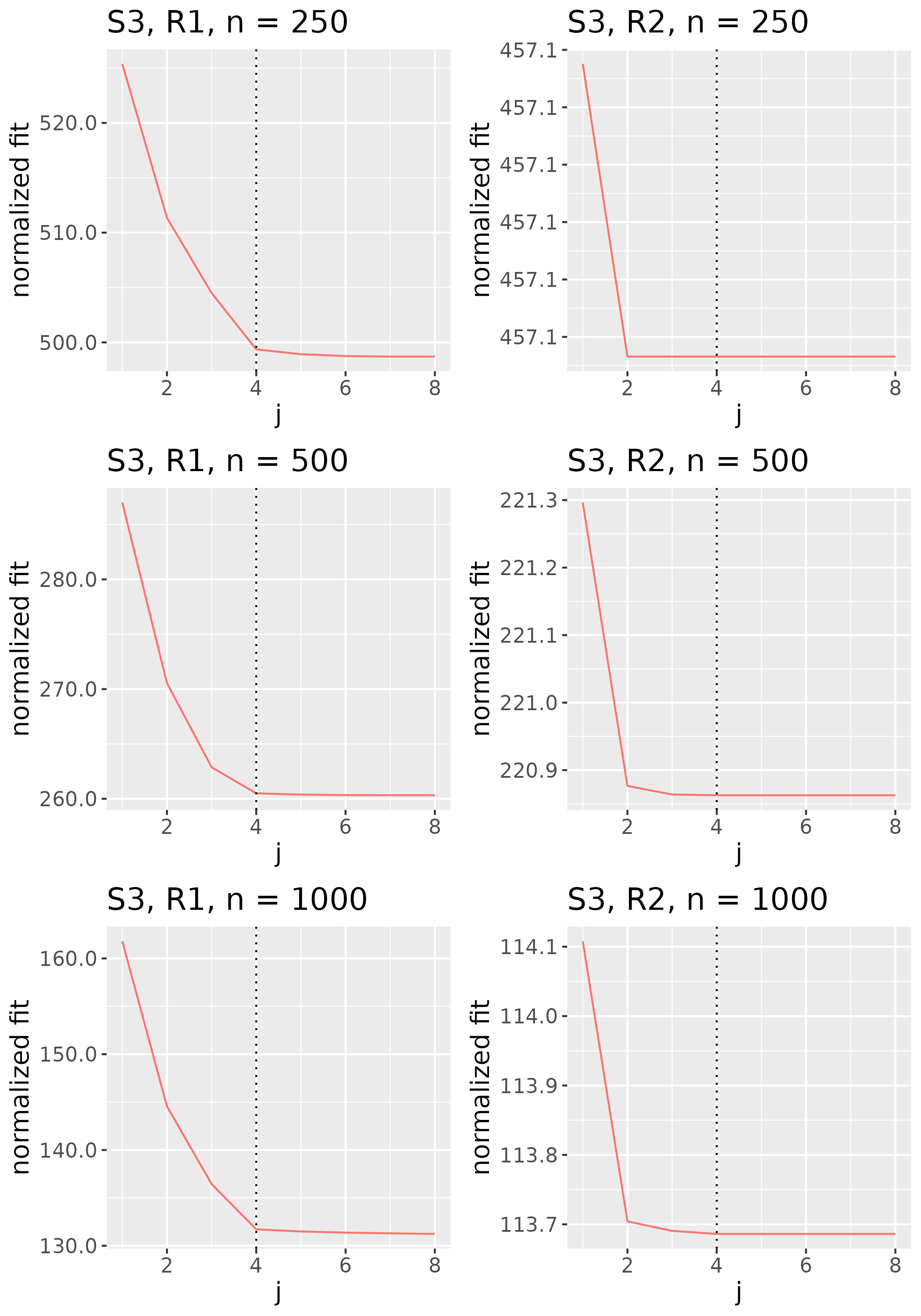}}\qquad%
\subfloat[Scenario S4]{\label{scree-plots-4-s1}\includegraphics[width=0.35\linewidth]{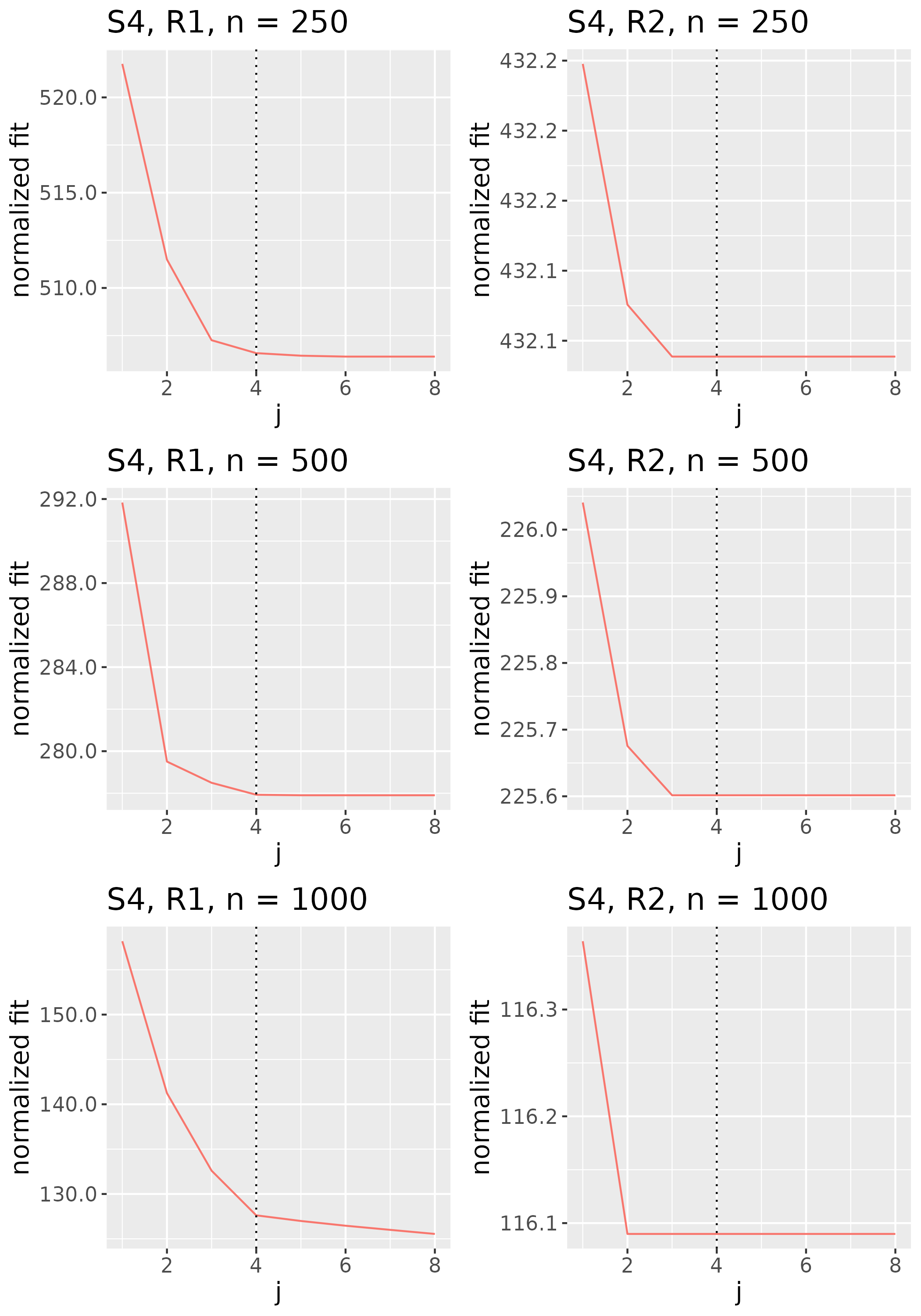}}%
\caption{Plots of the function $j \mapsto f(\hat{\theta}_j)$ for $K = 4$ and $\delta = 0.1$. Each plot is for some combination of scenario, regime, and sample size $n$.}
\label{fig:scree-plots-4}
\end{figure}

\section{AOMIC Data Analysis}
\label{sec:da-supp}

\subsection{Output from Higher-Order Models}
\label{sec:da-supp-out}

This section contains supplementary figures from the analysis of resting-state data from the Amsterdam Open MRI Collection (AOMIC) PIOP1 dataset in Section 5 of the manuscript. 

Figure \ref{fig:rs-whitening} demonstrates the effect of prewhitening on voxel time courses. Figure \ref{fig:rs-white-raw} shows a time course and corresponding autocorrelation function (acf) for three example voxels from one subject's preprocessed data, while Figure \ref{fig:rs-white-resid} displays each time course's ARIMA residuals scaled to unit variance with accompanying acf. 

\begin{figure}[!h]
\centering
\subfloat[Preprocessed BOLD.]{\label{fig:rs-white-raw}\includegraphics[width=0.35\linewidth]{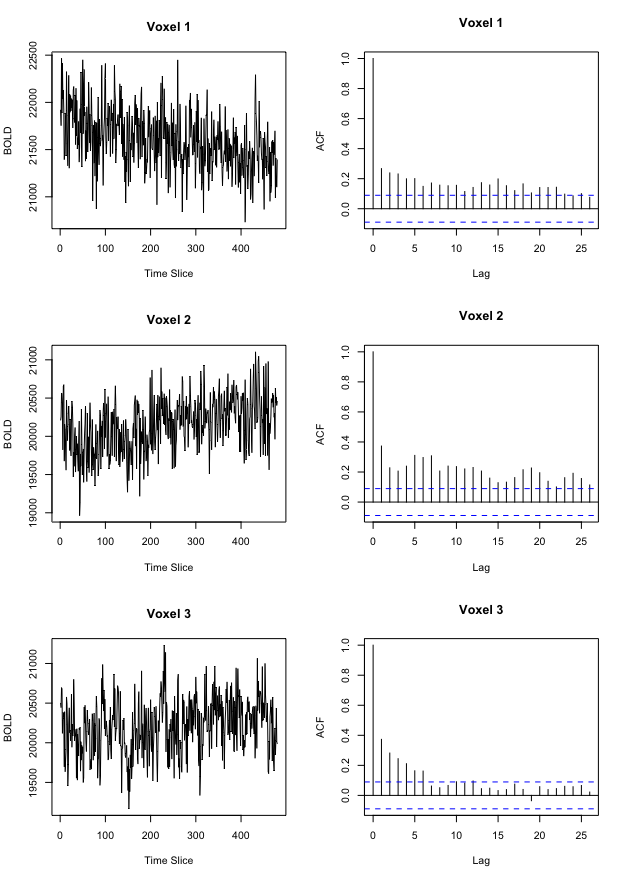}}\qquad
\subfloat[ARIMA residual.]{\label{fig:rs-white-resid}\includegraphics[width=0.35\linewidth]{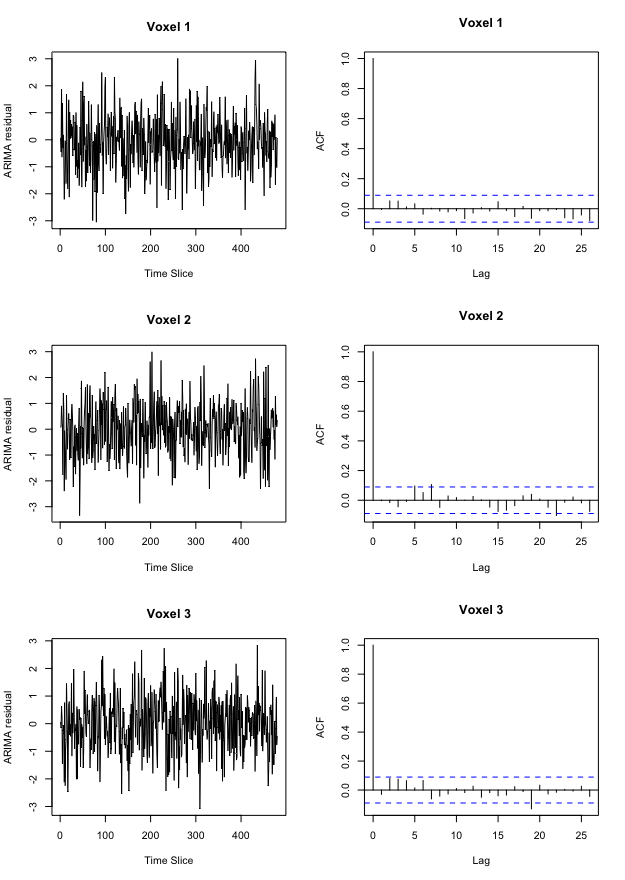}}
\caption{(a) The preprocessed BOLD time course and autocorrelation function (acf) for three example voxels whose time series exhibit varying degrees of nonstationarity and temporal correlation. (b) The scaled ARIMA residuals and acf for the three voxel time courses depicted in (a). In the acf plot of a white noise series, we expect 95\% of the nonzero-lag spikes to lie within the plot's blue dotted lines.}
\label{fig:rs-whitening}
\end{figure}

The left-hand plot of Figure \ref{fig:rs-scree} displays the non-increasing function $g(j) = \|\mcal{A} \circ (\hat{\mcal{C}}_n - \hat{\theta}_j)\|_F^2$ for $j = 1, \dots, 100$. To aid in identification of this plot's elbow, we also plot the ratios $r(i) = g(i) / g(i+1)$ for $i = 1, \dots, 99$ in Figure \ref{fig:rs-scree}. The estimated number of factors $\hat{K}$ should be the $j$ at which $g(j)$ levels off or, equivalently, the $i$ at which $r(i)$ becomes a constant close to 1. As discussed in Section 6 of the manuscript, we set $\hat{K} = 12$.

\begin{figure}[!h]
    \centering
    \includegraphics[scale=0.2]{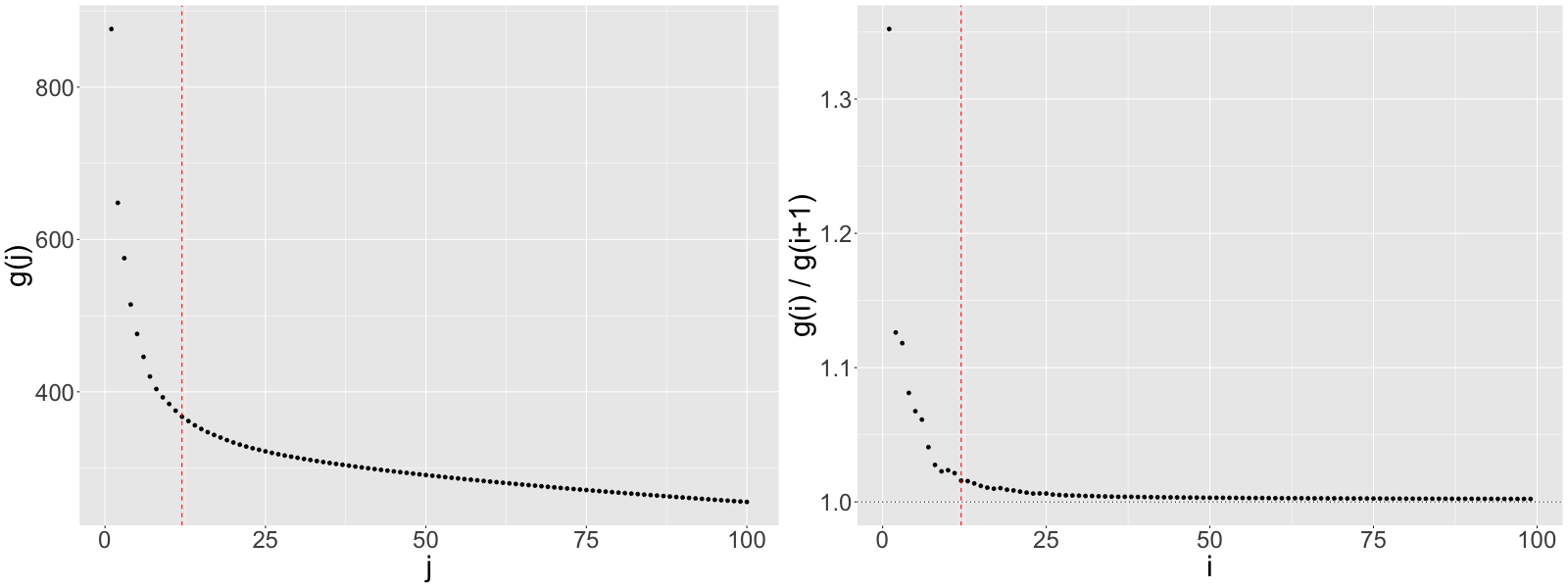}
    \caption{The scree-type plot (left) and ratio plot (right) used to select the number of factors in the functional factor analysis the AOMIC resting-state data.}
    \label{fig:rs-scree}
\end{figure}

Figure \ref{fig:rs-ffa-25} displays loading estimates from the 25-factor FFM while Figures \ref{fig:rs-ffa-50-1} and \ref{fig:rs-ffa-50-2} show those from the 50-factor FFM. Figure \ref{fig:rs-ica-25} displays independent component (IC) estimates from the 25-IC model while Figures \ref{fig:rs-ica-50-1} and \ref{fig:rs-ica-50-2} show those from the 50-IC model.

\begin{figure}[!h]
    \centering
    \includegraphics[width=0.8\linewidth]{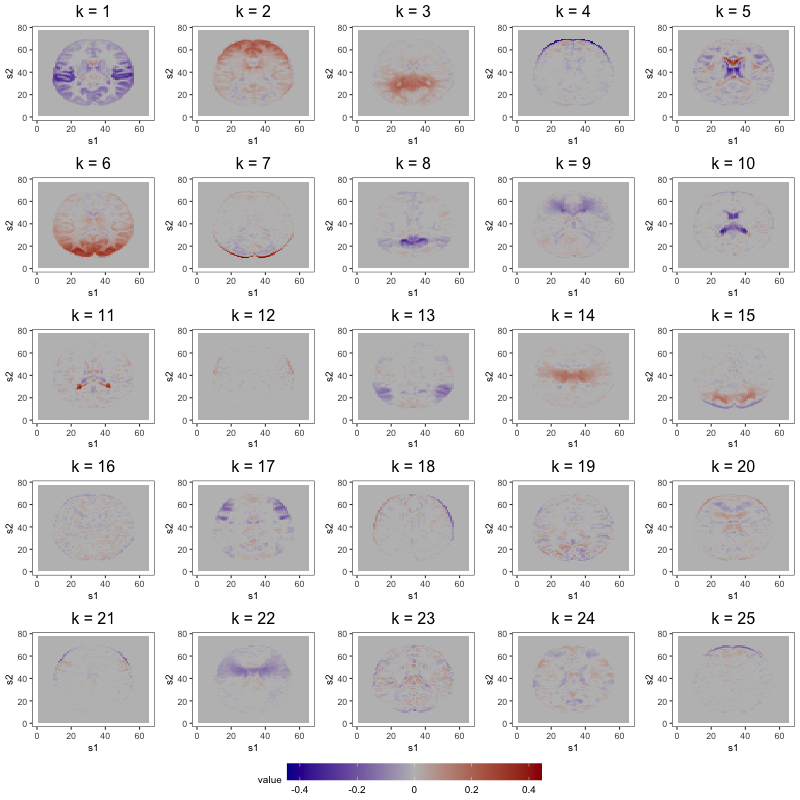}
    \caption{Loadings of the functional 25-factor model estimated from ARIMA residuals.}
    \label{fig:rs-ffa-25}
\end{figure}

\begin{figure}[!h]
    \centering
    \includegraphics[width=0.8\linewidth]{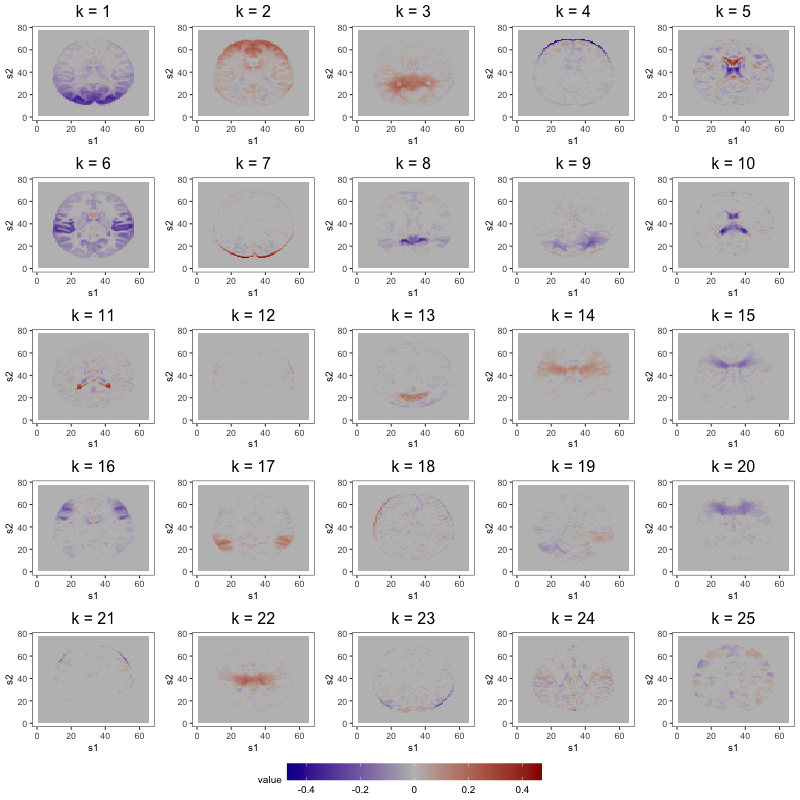}
    \caption{Loadings 1-25 of the functional 50-factor model estimated from ARIMA residuals.}
    \label{fig:rs-ffa-50-1}
\end{figure}

\begin{figure}[!h]
    \centering
    \includegraphics[width=0.8\linewidth]{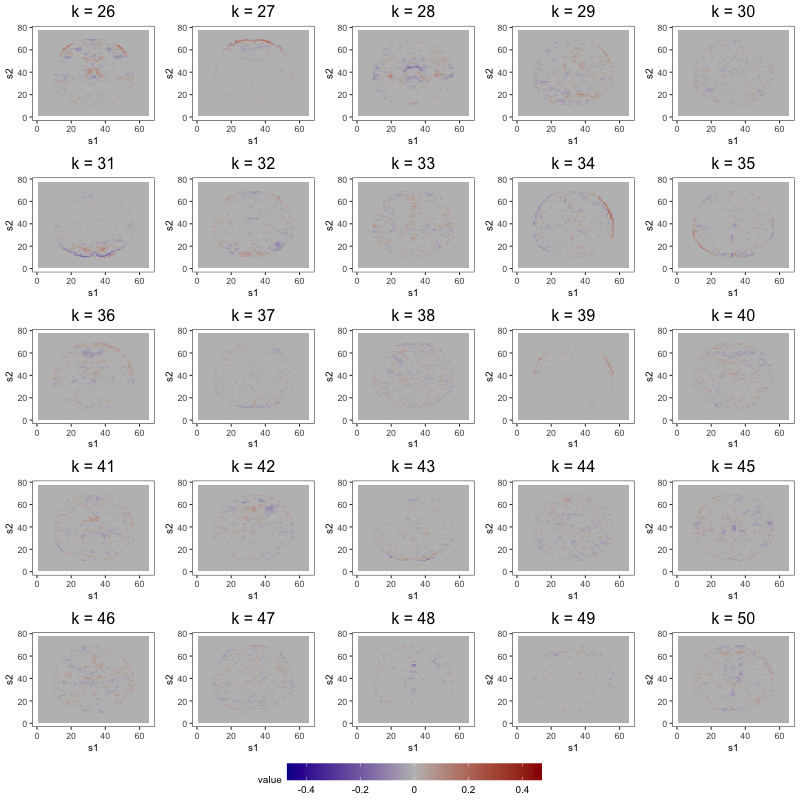}
    \caption{Loadings 26-50 of the functional 50-factor model estimated from ARIMA residuals.}
    \label{fig:rs-ffa-50-2}
\end{figure}

\begin{figure}[!h]
    \centering
    \includegraphics[width=0.8\linewidth]{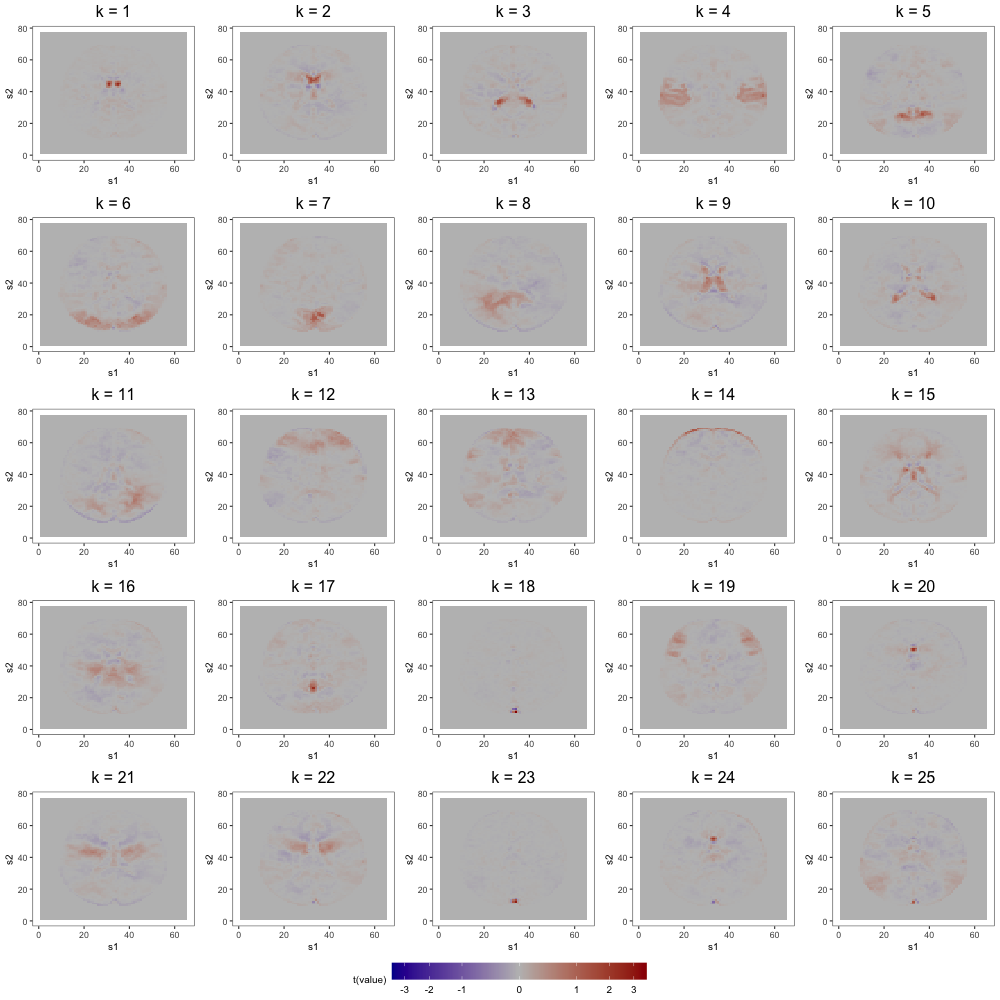}
    \caption{Transformed components ($t(x) = \text{sgn}(x) \log (\abs{x} + 1)$) of the 25-component IC model estimated from the preprocessed scans.}
    \label{fig:rs-ica-25}
\end{figure}

\begin{figure}[!h]
    \centering
    \includegraphics[width=0.8\linewidth]{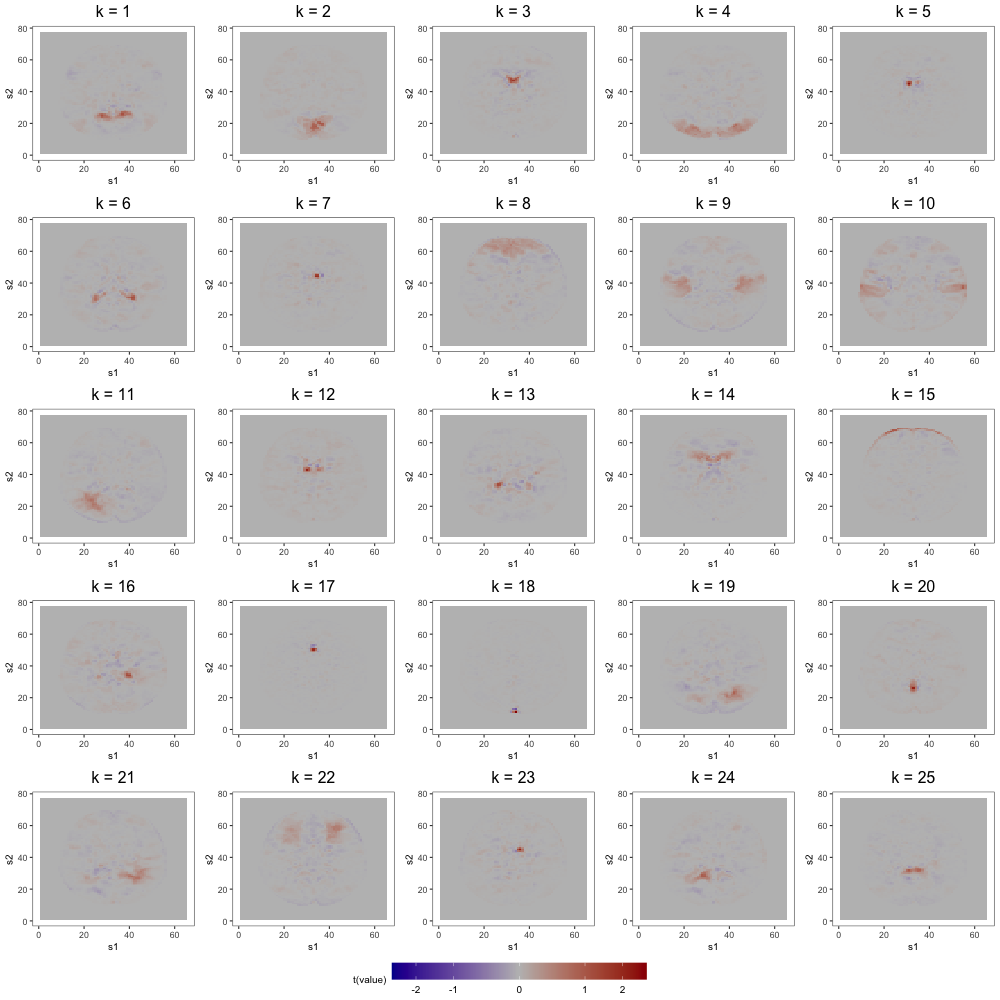}
    \caption{Transformed components ($t(x) = \text{sgn}(x) \log (\abs{x} + 1)$) 1-25 of the 50-component IC model estimated from the preprocessed scans.}
    \label{fig:rs-ica-50-1}
\end{figure}

\begin{figure}[!h]
    \centering
    \includegraphics[width=0.8\linewidth]{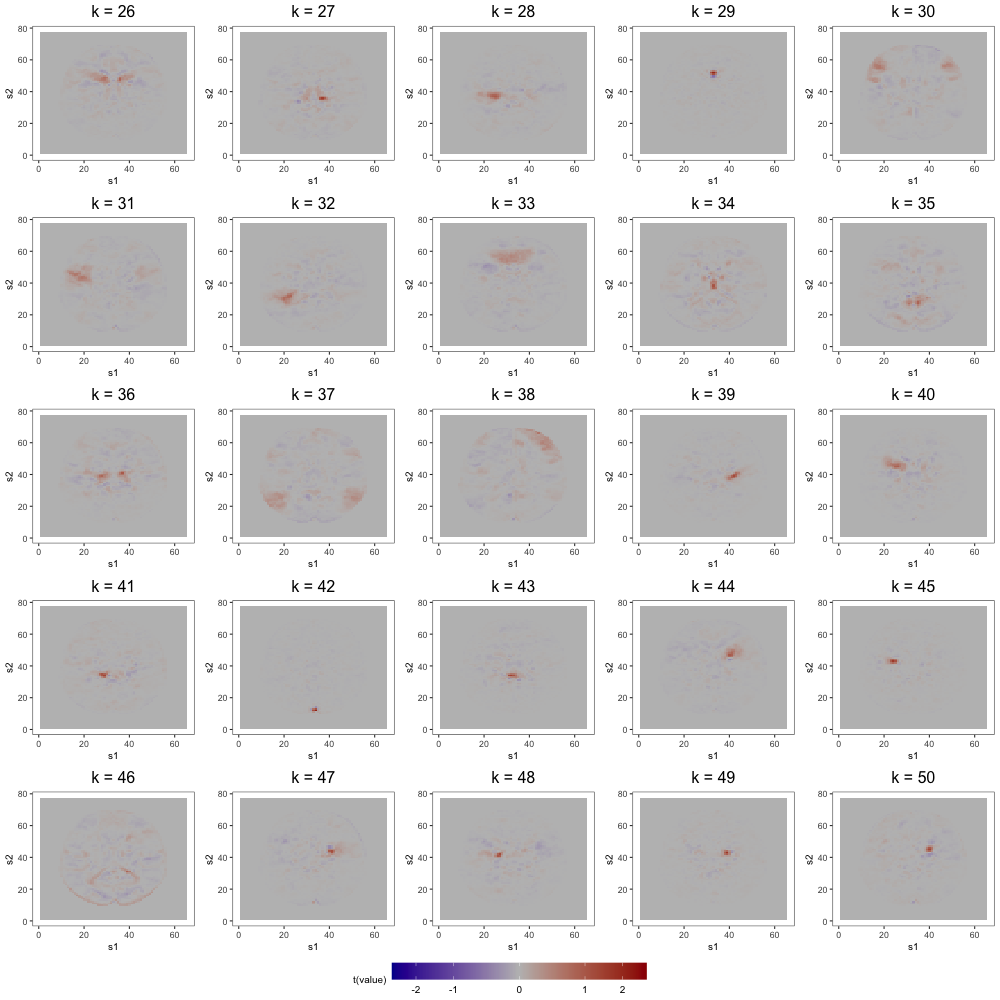}
    \caption{Transformed components ($t(x) = \text{sgn}(x) \log (\abs{x} + 1)$) 26-50 of the 50-component IC model estimated from the preprocessed scans.}
    \label{fig:rs-ica-50-2}
\end{figure}

\subsection{Multi-Scale ICA}
\label{sec:da-supp-iraji}

As described in Section 6 of the manuscript, multi-scale approaches, like that of \cite{iraji-etal-2023}, can remedy ICA's sensitivity to dimension misspecification. Adapting the method of \cite{iraji-etal-2023}
 to the AOMIC data, we first independently generate 25 random subsets of the full dataset, each containing the scans of 150 subjects. Next, we fit IC models of order 10, 20, 30, 40, and 50 to each subset, yielding 25 sets of 150 components. From the collection of all components, we select the 150 with the highest average spatial similarity (calculated by Pearson correlation) across the 25 sets, then identify the components among the 150 that are most distinct from each other (spatial similarity < 0.5). The final 45 components are displayed in Figures \ref{fig:rs-iraji-1} and \ref{fig:rs-iraji-2}. Among this final set is a single component ($k = 10$) containing the ventricle correlations of the first component from the 12-IC model, a structure splintered across several components of higher-order IC models. Comparing Figure 9b of the manuscript to Figures \ref{fig:rs-iraji-1} and \ref{fig:rs-iraji-2} confirms that other low-order structures are similarly preserved. As discussed in Section 6 of the manuscript, this solution to the fragmentation problem comes at the cost of additional computational resources and an interpretation more convoluted than that of approaches based on single IC models.  
 
\begin{figure}[!h]
    \centering
    \includegraphics[width=0.8\linewidth]{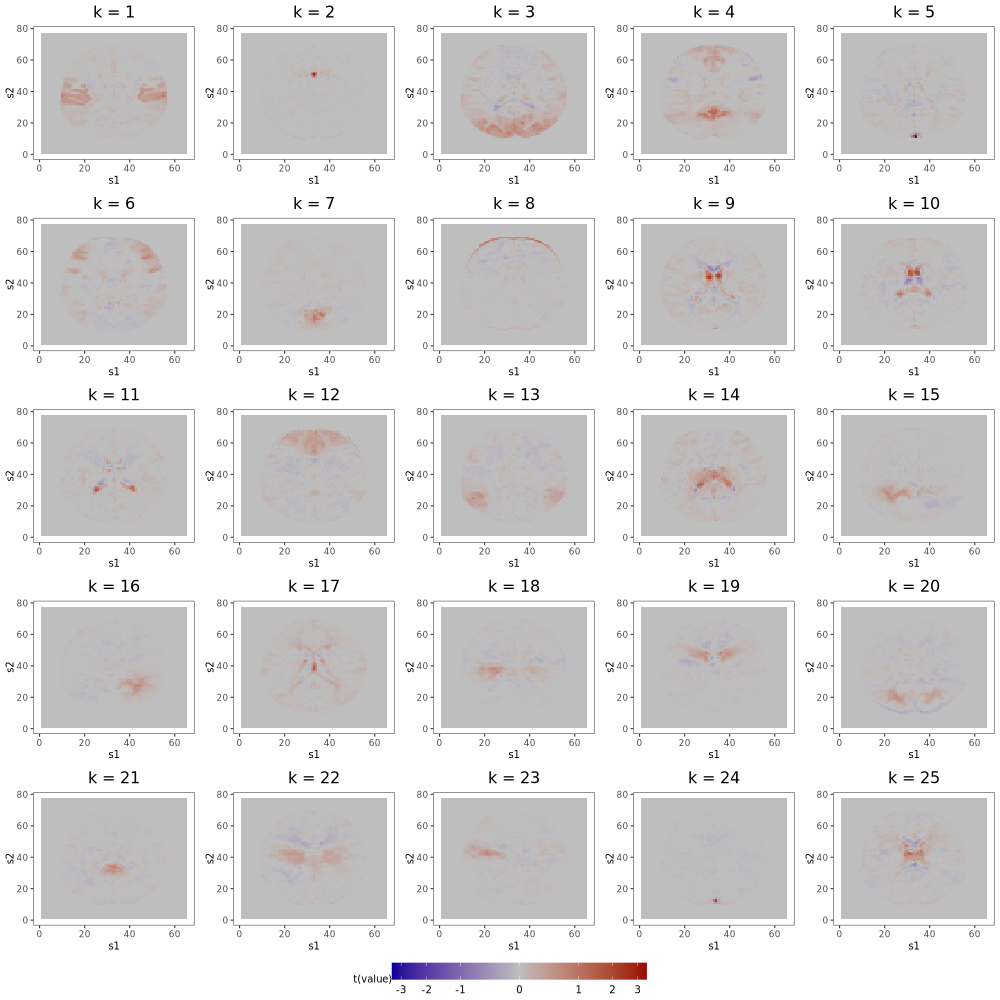}
    \caption{Transformed components ($t(x) = \text{sgn}(x) \log (\abs{x} + 1)$) 1-25 obtained via multi-scale ICA.}
    \label{fig:rs-iraji-1}
\end{figure}

\begin{figure}[!h]
    \centering
    \includegraphics[width=0.8\linewidth]{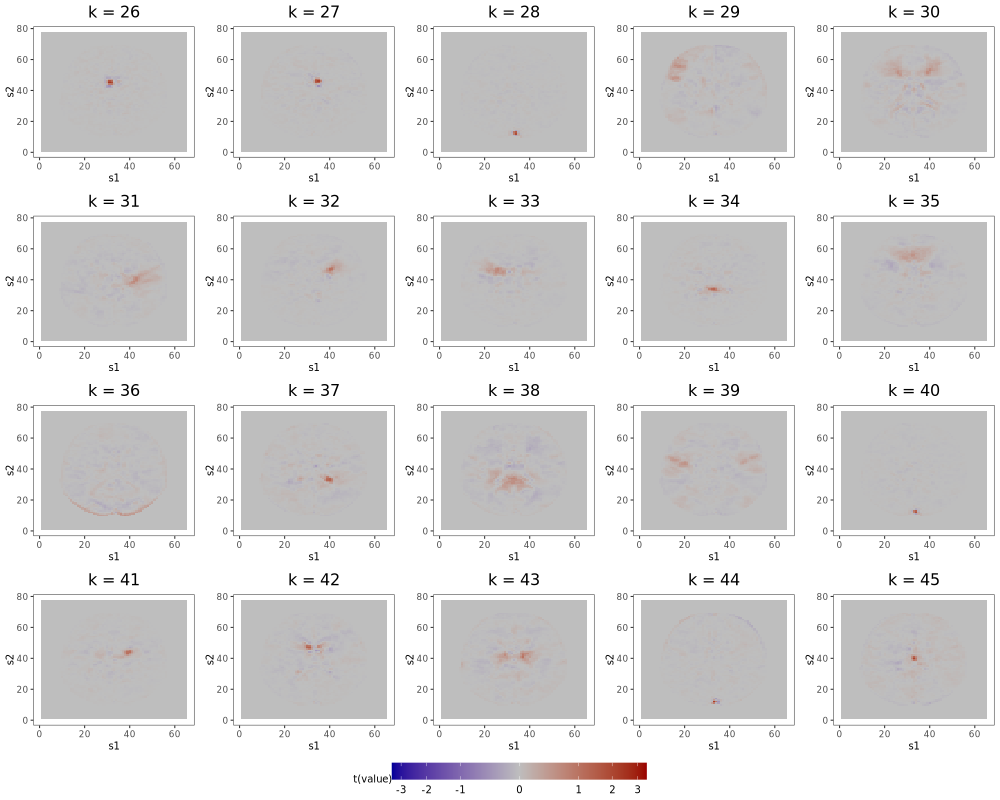}
    \caption{Transformed components ($t(x) = \text{sgn}(x) \log (\abs{x} + 1)$) 26-45 obtained via multi-scale ICA.}
    \label{fig:rs-iraji-2}
\end{figure}

\clearpage

\bibliographystyle{plainnat}
\bibliography{references}

\end{document}